\documentclass[10.4pt,reqno]{amsart}
\usepackage[top=2.5cm, bottom=2.4cm, left=2.5cm, right=2.5cm]{geometry}
\usepackage{dsfont, amssymb,amsmath,amscd,latexsym, amsthm, amsxtra,amsfonts}
\usepackage[all]{xy}
\usepackage[active]{srcltx}
\usepackage{chngpage}
\usepackage{array}
\usepackage{tabularx}
\usepackage{datetime}
\usepackage{float}

\usepackage[round]{natbib}
\usepackage{bbm}
\usepackage{enumerate}
\usepackage{mathrsfs}
\usepackage{graphicx}
\usepackage{comment}
\usepackage{mathtools}
\usepackage[hang,flushmargin]{footmisc}

\usepackage{tikz}
\usetikzlibrary{calc,arrows}
\usepackage{verbatim}
\usepackage{graphicx}
\usepackage{subfigure}

\newtheorem{theorem}{Theorem}[section]

\newtheorem{proposition}[theorem]{Proposition}

\theoremstyle{definition}
\newtheorem{definition}[theorem]{Definition}

\renewcommand{\theequation}{\arabic{section}.\arabic{equation}}

\theoremstyle{definition}

\theoremstyle{definition}
\newtheorem{remark}{Remark}
\theoremstyle{definition}

\DeclareMathOperator*{\diag}{diag}

\newcommand{\rd}{\mathrm{d}}

\renewcommand{\epsilon}{\varepsilon}

\newcommand{\var}{\mathbf{Var}}

\renewcommand{\cite}{\citet*}

\usepackage[pdfstartview=FitH, bookmarksnumbered=true,bookmarksopen=true, colorlinks=true, pdfborder=001, citecolor=blue, linkcolor=blue,urlcolor=blue]{hyperref}
\usepackage{graphics}
\graphicspath{{figures/}}
\usepackage[round]{natbib}
\setcitestyle{authoryear,round}

\begin{document}
	\makeatletter
	\def\@setauthors{%
		\begingroup
		\def\thanks{\protect\thanks@warning}%
		\trivlist \centering\footnotesize \@topsep30\p@\relax
		\advance\@topsep by -\baselineskip
		\item\relax
		\author@andify\authors
		\def\\{\protect\linebreak}%
		{\authors}%
		\ifx\@empty\contribs \else ,\penalty-3 \space \@setcontribs
		\@closetoccontribs \fi
		\endtrivlist
		\endgroup } \makeatother
	\baselineskip 17.3pt
	\title[{{\tiny Robust $N$-insurer games }}]
	{{\tiny Robust mean-variance stochastic differential reinsurance and investment games under volatility risk and model uncertainty}} \vskip 10pt\noindent
		\author[{\tiny  Guohui Guan, Zongxia Liang, Yi Xia}]
		{\tiny {\tiny  Guohui Guan$^{a,\dag}$, Zongxia Liang$^{b,\ddag}$, Yi Xia$^{b,*}$}
				\vskip 10pt\noindent
				{\tiny ${}^a$Center for Applied Statistics and School of Statistics, Renmin University of China, Beijing, 100872, China
						\vskip 10pt\noindent\tiny ${}^b$Department of Mathematical Sciences, Tsinghua
						University, Beijing 100084, China}
				\footnote{ 
						$^{\dag}$  e-mail: guangh@ruc.edu.cn\\
						$^{\ddag}$  e-mail:  liangzongxia@tsinghua.edu.cn\\
						$^*$  Corresponding author, e-mail:  xia-y20@mails.tsinghua.edu.cn}}
	\maketitle
\begin{abstract}	
This paper investigates robust stochastic differential games among insurers under model uncertainty and stochastic volatility. The surplus processes of ambiguity-averse insurers (AAIs) are characterized by drifted Brownian motion with both common and idiosyncratic insurance risks. To mitigate these risks, AAIs can purchase proportional reinsurance. Besides, AAIs allocate their wealth in a financial market consisting of cash, and a stock characterized by the 4/2 stochastic volatility model.  AAIs compete with each other based on relative performance with the mean-variance criterion under the worst-case scenario. This paper  formulates a robust time-consistent mean-field game in a non-linear system. The AAIs seek robust, time-consistent response strategies to achieve Nash equilibrium strategies in the game. We introduce $n$-dimensional extended Hamilton-Jacobi-Bellman-Isaacs (HJBI) equations and corresponding verification theorems under compatible conditions. Semi-closed forms of the robust $n$-insurer equilibrium and mean-field equilibrium are derived, relying on coupled Riccati equations. Suitable conditions are presented to ensure the existence and uniqueness of the coupled Riccati equation as well as the integrability in the verification theorem. As the number of AAIs increases, the results in the $n$-insurer game converge to those in the mean-field game. Numerical examples are provided to illustrate economic behaviors in the games, highlighting the herd effect of competition on the AAIs.
	\vskip 10pt  \noindent
	2020 Mathematics Subject Classification: 91G05, 91B05, 91G10, 49L20, 91A15.
	\vskip 10pt  \noindent
	JEL Classifications: G22, G11, C61, C72.
	\vskip 10 pt  \noindent
	Keywords:  Stochastic differential games; Investment-reinsurance; Mean-variance criterion; Equilibrium strategy;  Model uncertainty; 4/2 model.
\end{abstract}
	\section{\bf Introduction}

 Stochastic differential games provide a framework for analyzing decision-making problems involving multiple players, where choices are influenced not only by individual performance but also by the relative performances of peers.  Based on various research studies, such as \cite{Hopkins} and \cite{demarzo2008relative}, which emphasize the influence of peers' wealth levels on an individual's utility, this paper focuses on the field of stochastic differential games. \cite{elliott1976existence} explores a two-person zero-sum differential game, while \cite{fleming1989existence} investigate the existence of the value for zero-sum, two-player stochastic differential games. \cite{bensoussan2000stochastic} apply dynamic programming method in a class of stochastic games for $N$
players. Due to the importance of relative performance and competition among the players, stochastic differential games have garnered significant attention.

In the context of finance and insurance, there have been numerous studies on reinsurance and investment strategies within the framework of stochastic differential games. Empirical investigations, such as \cite{biggar1998competition} and \cite{murat2002competition}, emphasize the pivotal role of competition in shaping insurers' decisions.  Specifically, \cite{alhassan2018competition} examine  how competition influences risk-taking behavior in emerging insurance markets, while \cite{chang2022does} demonstrates that insurers with varying levels of competitiveness exhibit diverse risk-taking tendencies. \cite{pun2016robust} investigate a non-zero-sum stochastic differential game involving two ambiguity-averse insurers, focusing on the effects of relative performance concerns and ambiguity aversion on reinsurance demand. \cite{asimit2018insurance} study the set of Pareto optimal insurance contracts and the core of an
insurance game. In contrast to numerous works focusing on two-insurer case (e.g., \cite{bensoussan2014class}, \cite{deng2018non},  \cite{zhang2021class}, \cite{bai2022hybrid}, and \cite{cao2023reinsurance}), this paper addresses the practical reality of having more than two insurers in the market. The novel approach integrates model uncertainty and relative performance considerations to explore robust investment-reinsurance decisions for both a finite number ($n$) and an infinite number of insurers. The insurers are faced with common and idiosyncratic insurance risks and volatility risk, with the option of purchasing proportional reinsurance. The primary focus is on deriving robust equilibrium strategies for insurers operating in a competitive setting under  worst-case scenarios. 

While the existing literature on stochastic differential games predominantly centers around It\^{o} (jump) diffusions, recent events such as the COVID-19 pandemic have led to increased volatility in the US equity and options market, as highlighted by \cite{john2021covid} and \cite{carriero2022addressing}. Acknowledging the need to effectively capture this heightened volatility and manage pandemic risks, stochastic volatility models have gained significance. Empirical studies, particularly those exploring volatility smile and volatility clustering, also indicate that volatility is generally non-constant, introducing volatility risk. To comprehensively address volatility risk, this paper adopts the 4/2 stochastic volatility model proposed by \cite{grasselli20174}. Notably, this model encompasses both the Heston model by \cite{heston1993closed} and the 3/2 model by \cite{heston1997simple} as special cases. While Heston's model is widely utilized for derivatives pricing and asset allocation, it often falls short of satisfying the Feller condition when calibrated to real data. In contrast, the 4/2 model proves more effective in capturing the skewness and kurtosis observed in empirical data compared to the Heston model.

In addition, this study examines the effects of ambiguity within the stochastic differential game. Since the pioneering work of \cite{DE}, there has been a growing recognition that individuals exhibit ambiguity attitudes when making decisions. To address this, \cite{GS} propose the max-min expected utility model, which is  widely employed in financial mathematics and actuarial sciences. This model aims to derive robust strategies  under worst-case scenarios. Studies by \cite{HS} and \cite{Anderson} introduce a penalty term based on relative entropy in max-min expected utility problems to measure discrepancies between the  reference model and alternative model. In the context of robust portfolio problems,  \cite{Maenhout2004} demonstrates that a preference for robustness can significantly alter demand for equities in portfolios, affecting the equilibrium equity premium. While robust investment-reinsurance strategies have been extensively studied for insurers (\cite{huang2017robust}, \cite{sun2017robust}, \cite{gu2018optimal}, \cite{wang2022robust}, etc.), these works mainly focus on robust strategies for a single insurer.  Although existing studies indicate that higher ambiguity aversion reduces risk exposure for a single insurer in financial and insurance markets, exploration of ambiguity within stochastic differential games is relatively recent in finance and insurance. Studies by \cite{deng2018non} and \cite{bai2022hybrid} investigate the impact of peers' risk attitudes, revealing that increased risk aversion results in decreased holdings of risky assets. However, the effects of peers' ambiguity attitudes remain unexplored, and this paper aims to fill this gap in the literature.

We focus on establishing robust time-consistent equilibrium strategies for competitive insurers with volatility risk under the mean-variance criterion.  Based on \cite{lacker2019mean} and \cite{kraft2020dynamic}, our research  introduces model uncertainty and the 4/2 stochastic volatility model, which result in robust mean-field game in a non-linear system.  The insurers search for the robust strategies with relative concerns, accounting for model uncertainty and volatility risk. Moreover, this paper show precise definitions of robust $n$-insurer and mean-field equilibria.  By the $n$-dimensional extended HJBI  equations, we derive semi-closed form equilibrium strategies for the robust time-consistent $n$-insurer game under the mean-variance criterion. The results rely on coupled Riccati equations and we show suitable conditions to ensure the existence and uniqueness of the equations. 
Besides, we show the convergence of results from the $n$-insurer game to the mean-field game. Finally, we present numerical examples to show the impacts of competition level, risk aversion coefficients, and ambiguity aversion coefficients. We also find the herd effect of competition in the robust game.

To recap, the main contributions of this paper include the
following aspects:

 (a) It is worth mentioning that in several two-agent games studied in the literature, such as \cite{bensoussan2014class}, \cite{deng2018non}, \cite{kraft2020dynamic}, \cite{bai2022hybrid}, and others, the financial market is characterized by standard Brownian motion, the constant elasticity of variance (CEV)  model, or Heston's stochastic volatility model. In this paper, we  extend \cite{deng2018non} by considering a more general stochastic volatility model. Specifically, we adopt the 4/2 stochastic volatility model, which serves as a unified approach encompassing both the Heston model and the 3/2 model. By employing the 4/2 model, we aim to capture a broader range of market dynamics and address the limitations of both the Heston and the 3/2 models in explaining the implied volatility surface. 
 
 (b)  In contrast to the existing literature, including works such as \cite{yong2002leader}, \cite{aid2020nonzero}, \cite{bai2022hybrid}, and \cite{savku2022stochastic}, this paper focuses on studying the effects of model uncertainty in the $n$-insurer stochastic differential game. We address this by formulating the robust problem as a max-min problem, which requires suitable integrable conditions to ensure the well-posedness of the problem. Furthermore, the non-linearity of the 4/2 model adds complexity to the problem. Unlike models that satisfy the uniform Lipschitz condition, the 4/2 model does not. We establish compatible conditions in Theorem \ref{solution-HJBI} to address this challenge. 
 
 (c) Previous research in the field has predominantly focused on studying stochastic differential games involving only two finite agents, such as \cite{bensoussan2014class}, \cite{deng2018non}, \cite{medhin2020nonzero},  and \cite{zhang2021class}. However, the iteration among a larger number of agents, including $n$ or even an infinite number of agents, is relatively new  in the context of insurance.
 It is worth mentioning that the insurance industry is highly competitive, with a significant number of domestic insurers reported in the United States alone, reaching 5978 by the end of 2021. As such, it becomes increasingly important to understand the dynamics and strategic interactions among multiple insurers. This paper tackles this challenge by formulating the $n$-agent and mean-field games for insurers. Besides, solvable  games are formulated mainly for the linear system, see \cite{bardi2014linear}, \cite{bensoussan2016linear}, \cite{moon2016linear}, \cite{lacker2019mean}, etc. This paper also contributes to the area of mean-field game by explicitly studying the  game in a non-linear system, which arise from the 4/2 model and variance principle.
 
 (d) In contrast to the CARA utility function employed in previous works such as \cite{deng2018non} and \cite{bai2022hybrid}, insurers in our model seek to maximize the mean-variance criterion. The optimal mean-variance investment-reinsurance strategies for a single insurer has been studied in  \cite{zhao2016time} and \cite{wang2019optimal}. The mean-variance criterion contains a non-linear function of the expectation of terminal wealth, which violates Bellman's optimality principle and leads to time inconsistency. To address this issue, we solve the problem based on the time-consistent equilibrium strategy proposed in \cite{bjork2017time}. By employing this approach, we overcome the time inconsistency problem associated with the mean-variance criterion and present $n$-dimensional extended HJBI equations. The results rely on coupled Riccati equations, as stated in Proposition \ref{ricc}, which also discusses their existence. These coupled Riccati equations differentiate themselves from the equations found in the standard geometric Brownian motion or CEV model (\cite{li2014optimal}, \cite{bai2022hybrid}, \cite{wang2022robust}, etc.).
 
(e) We conduct a comprehensive analysis of the impacts of competition within our framework.  \eqref{dn-opt} reveals that the robust equilibrium investment strategy comprises four components: a myopic part, a hedging demand, a myopic part and a hedging demand resulting from competition. On the other hand, the robust equilibrium reinsurance strategy is implicitly determined by two  components: one that disregards competition and another that responds to the actions of others. Similar to the findings in the two-insurer game studied in \cite{deng2018non}, we observe the presence of a herd effect in insurers' decision-making processes. This means that each insurer tends to make decisions by mimicking the actions of their peers. Furthermore, we investigate the effects of risk aversion coefficients and ambiguity aversion coefficients on insurers' strategies. Our analysis reveals that these two coefficients have similar effects on insurers' decision-making. 

The remainder of this paper is organized as follows. Section 2 presents the financial model. In Section 3, we derive the robust time-consistent equilibrium strategies in the $n$-insurer game under the mean-variance criterion. Section 4 shows the results in the robust mean-field game. Section 5 presents the numerical results, and Section 6 is a conclusion. All the proofs are provided in the Appendices.

\section{\bf Financial model}
Consider a complete probability space denoted as $\left(\Omega, \mathcal{F}, \mathbb{P}\right)$, where $\mathbb{P}$ is a reference probability measure. $[0,T]$ is a fixed investment time horizon.  The natural filtration $\mathbb{F}= \left\{\mathcal{F}_{t},{t \in[0, T]}\right\}$ is generated by the standard Brownian motions within the space $\left(\Omega, \mathcal{F}, \mathbb{P}\right)$. We assume that all processes introduces below are well-defined and adapted to $\mathbb{F}$. Additionally, there are no transaction costs, and short sellings are allowed. 


This paper  investigates stochastic differential game among $n$ \footnote{Stochastic differential game among two insurers have been studied in \cite{bensoussan2014class}, \cite{medhin2020nonzero}, \cite{deng2018non}, \cite{bai2022hybrid}. However, the insurance market consists of more than two insurers and is very competitive.
}
insurers with common shock dependence as in \cite{bi2021equilibrium}. To desribe the surplus processes of these insurers, we adopt an approximate model based on the Cramér-Lundberg model. In the Cramér-Lundberg model, the dynamic surplus process $X_i=\left\{X_{i}(t), t \in[0,T] \right\}$ of insurer $i$ ($i \in{1,2, \ldots, n}$) is given by
$$
X_{i}(t)=x^{0}_{i}+p_{i} t-\sum_{j=1}^{K_{i}(t)} Y_{j}^{i},
$$
where $x_i^0 \geqslant 0$ represents the initial surplus, $p_i > 0$ denotes the premium rate, and ${Y_j^i, j=1,2,\ldots}$ constitutes a set of independent identically distributed (i.i.d.) random variables with a common distribution function $F_{Y^i}(\cdot)$. The process $K_i(t)=\hat{N}(t)+N_i(t)$ represents the total number of jumps up to time $t$ for insurer $i$, wherein $N_i(t)$ and $\hat{N}(t)$ are mutually independent Poisson processes with intensities $\lambda_i > 0$ and $\hat{\lambda} > 0$, respectively. It is assumed that $Y^i$ possesses a finite mean $\mu_{i1}=\mathbb{E}[Y^i]$ and a second-order moment $\mu_{i2}=\mathbb{E}[(Y^i)^2]$. The determination of the premium rate $p_i$ adheres to the expected value principle, specifically,
$$
p_{i}=\frac{1}{t}\left(1+\eta_{i}\right) \mathbb{E}\left[\sum_{j=1}^{K_{i}(t)} Y_{j}^{i}\right]=\left(1+\eta_{i}\right)\left(\lambda_{i}+\hat{\lambda}\right) \mu_{i 1},
$$where  $ \eta_{i}>0 $ is the safety loading coefficient of insurer $ i $.

 In this paper, we employ a Brownian motion with drift to approximate the cumulative claims process. More precisely, the cumulative claims process is expressed as:
$$
\sum_{j=1}^{K_{i}(t)} Y_{j}^{i} \approx\left(\hat{\lambda}+\lambda_{i}\right) \mu_{i 1} t-\sqrt{\left(\hat{\lambda}+\lambda_{i}\right) \mu_{i 2}} W_{i}(t).
$$
In this context, $W_i=\{W_{i}(t),t\in[0,T]\}$ refers to a standard Brownian motion, and any two Brownian motions $W_{i}$ and $W_{k}$, where $i \neq k \in\{1,2, \ldots, n\}$, are correlated with a correlation coefficient of
$$
\rho_{i k}:=\frac{\hat{\lambda} \mu_{i 1} \mu_{k 1}}{\sqrt{\left(\hat{\lambda}+\lambda_{i}\right)\left(\hat{\lambda}+\lambda_{k}\right) \mu_{i 2} \mu_{k 2}}}.
$$ 
Equivalently, we assume that
 $$ W_i(t)=\sqrt{\frac{\hat{\lambda}\mu_{i 1}^2}{(\hat{\lambda}+\lambda_i)\mu_{i 2}}} \tilde{W}(t)+\sqrt{1-\frac{\hat{\lambda}\mu_{i 1}^2}{(\hat{\lambda}+\lambda_i)\mu_{i 2}}}\hat{W}_i(t),$$ where $\tilde{W}=\{\tilde{W}(t),t\in[0,T]\}$, $\hat{W}_i=\{\hat{W}_{i}(t),t\in[0,T]\}$, $1\leqslant i\leqslant n $, are standard Brownian motions and are independent with each other. 
Denote $\hat{\textbf{W}}(t)= (\hat{W}_1(t),\hat{W}_2(t),\cdots,\hat{W}_n(t))^T $.

Then the surplus process  of insurer $i$ is approximated by \begin{equation*}\begin{aligned}
	X_{i}(t)=&x^{0}_{i}+p_{i} t-\left(\hat{\lambda}+\lambda_{i}\right) \mu_{i 1} t+\sqrt{\hat{\lambda}}\mu_{i 1} \tilde{W}(t)+\sqrt{(\hat{\lambda}+\lambda_{i})\mu_{i 2}-\hat{\lambda}\mu_{i 1}^2}\hat{W}_i(t)\\
=&x^{0}_{i}+\eta_{i}\left(\hat{\lambda}+\lambda_{i}\right) \mu_{i 1} t+\sqrt{\hat{\lambda}}\mu_{i 1} \tilde{W}(t)+\sqrt{(\hat{\lambda}+\lambda_{i})\mu_{i 2}-\hat{\lambda}\mu_{i 1}^2}\hat{W}_i(t).	\end{aligned}
\end{equation*}

In this study, insurers manage their risk by  purchasing proportional reinsurance from a reinsurer and engaging in investments in the financial market. In the context of the reinsurance model, the proportional reinsurance level of insurer $i$ at time $t$ is denoted by $a_i(t) \in [0,1]$. The reinsurance premium rate paid by insurer $i$ at time $t$ is denoted as $\hat{p}_i(t)$ and is calculated using the variance principle, relying on the expected loss and variance of the underlying risk. Specifically,
$$
\hat{p}_i(t) = (1-a_i(t))\left(\lambda_{i}+\hat{\lambda}\right)+\hat{\eta} (1-a_i(t))^2 \left(\lambda_{i}+\hat{\lambda}\right) \mu_{i 2},
$$
where $\hat{\eta} > 0$ is the safety loading coefficient of the reinsurer.

Apart from their involvement in the insurance market, the insurers in this study also engage in investments within a financial market comprising both cash and a stock. The risk-free rate in the financial market is represented by $r\geqslant 0$, denoting the return on cash. The stock price is denoted by $S={S(t), t\in [0,T]}$ and is characterized using the $4/2$ stochastic volatility model, expressed as:
\begin{equation}
	\left\{\begin{aligned}
		&\frac{\rd S(t)}{S(t)}=\left(r +m({a}{Z(t)}+{{b}})\right)\rd t+{\Sigma}(t)\rd {W}(t), \quad S_0=s_0,\\
			&\rd Z(t)={\kappa}(\bar{Z}-Z(t))\rd t+{\nu}\sqrt{Z(t)}\left[{\rho}\rd {W}(t)+\sqrt{1-{\rho}^2}\rd{B}(t)\right],\quad Z(0)=z^0,
	\end{aligned}\right.
\end{equation}
where\begin{equation*}
	{\Sigma}(t)={a}\sqrt{Z(t)}+\frac{{b}}{\sqrt{Z(t)}}.
\end{equation*}	 

The processes $W=\{W(t), t\in [0,T]\}$ and $B=\{B(t), t\in [0,T]\}$ are independent standard Brownian motions. To ensure the strict positivity of the process $Z(t)$, we impose the Feller condition $2\kappa\bar{Z} > \nu^2$. The 4/2 stochastic volatility model, introduced by \cite{grasselli20174}, encompasses both the Heston model by \cite{heston1993closed} and the 3/2 model by \cite{heston1997simple} as special cases.

Let $\pi_i(t)$ denote the amount invested in the stock by insurer $i$ at time $t$, while the remaining portion $X_i(t)-\pi_i(t)$ is allocated in cash. Considering both reinsurance and investment, the non-linear\footnote{The non-linearity of the system is a result of the variance principle and the 4/2 model. In contrast to solvable mean-field games found in linear systems like those discussed in \cite{bardi2014linear}, \cite{bensoussan2016linear}, \cite{moon2016linear}, \cite{lacker2019mean}, and others, we introduce a solvable mean-field game within a non-linear system that incorporates quadratic and square root terms.} surplus process $X_i$ for insurer $i$ is expressed as follows:
$$
\begin{aligned}
	\mathrm{d} X_{i}(t)
	=&\left[r X_{i}(t)+p_i-\hat{p}_i(t)-a_i(t)\left(\lambda_{i}+\hat{\lambda}\right) \mu_{i 1}+\pi_{i}(t)m\sqrt{Z(t)}\Sigma(t)\right] \mathrm{d} t\\&+a_{i}(t) \left(\sqrt{\hat{\lambda}}\mu_{i 1} \mathrm{d}\tilde{W}(t)+\sqrt{(\hat{\lambda}+\lambda_{i})\mu_{i 2}-\hat{\lambda}\mu_{i 1}^2}\mathrm{d}\hat{W}_i(t)\right)+\pi_{i}(t){\Sigma}(t)\rd {W}(t),\\
	=&\left[r X_{i}(t)+\eta_{i}\left(\lambda_{i}+\hat{\lambda}\right) \mu_{i 1}-\hat{\eta}(1-a_i(t))^2\left(\lambda_{i}+\hat{\lambda}\right) \mu_{i 2}+\pi_{i}(t)m\sqrt{Z(t)}\Sigma(t)\right] \mathrm{d} t\\&+a_{i}(t) \left(\sqrt{\hat{\lambda}}\mu_{i 1} \mathrm{d}\tilde{W}(t)+\sqrt{(\hat{\lambda}+\lambda_{i})\mu_{i 2}-\hat{\lambda}\mu_{i 1}^2}\mathrm{d}\hat{W}_i(t)\right)+\pi_{i}(t){\Sigma}(t)\rd {W}(t).
\end{aligned}
$$

\section{\bf Problem formulation}

In the context of financial and insurance modeling, the prevalence of model uncertainty or ambiguity arises from the inherent difficulty in accurately estimating model parameters. Consequently, it becomes imperative for insurers to incorporate model uncertainty into their risk management strategies. Insurers often exhibit ambiguity aversion, necessitating the adoption of a robust risk management framework. Model uncertainly is ignored in \cite{medhin2020nonzero}, \cite{deng2018non}, \cite{bai2022hybrid}, etc.  \cite{deng2018non} find the herd effect in the competition. However, it is crucial to investigate the herd effect in competition while taking into account the presence of model uncertainty.  In this study, we focus on ambiguity-averse insurers (AAIs) and address model uncertainty by formulating the max-min robust problem for AAIs. The uncertainty w.r.t.  the reference model is characterized by a class of probability measures that are equivalent to the reference measure $\mathbb{P}$, as follows:
$$
\mathcal{Q}^{i}:=\left\{\mathbb{Q}^{\varphi_i,\chi_i,\phi_i,\vartheta_{i}} \mid \mathbb{Q}^{\varphi_i,\chi_i,\phi_i,\vartheta_{i}} \sim \mathbb{P}\right\},
$$
where functions  $\varphi_{i}(t)$, $ \chi_{i}(t),\phi_i(t),\vartheta_{i}(t)=(\vartheta_{i,1}(t),\vartheta_{i,2}(t),\cdots,\vartheta_{i,n}(t))^T$ satisfy the following conditions:
\begin{enumerate}
	\item $\varphi_{i}(t), \chi_{i}(t),\phi_i(t),\vartheta_{i}(t)$ are progressively measurable w.r.t. $\mathbb{F}$.
	\item 	The Novikov's condition is satisfied:$$
	\mathbb{E}\left[\exp \left(\frac{1}{2} \int_{0}^{T}\left[{\varphi_{i}^2(t)}+\chi_{i}^2(t)+\phi_i^2(t)+\vartheta_{i}(t)^T\vartheta_{i}(t)\right] \mathrm{d} t\right)\right]<+\infty.
	$$
\end{enumerate} 

For each $\varphi_{i}(t), \chi_{i}(t),\phi_i(t),\vartheta_{i}(t)$, the real valued process $\Theta^{\varphi_i,\chi_i,\phi_i,\vartheta_{i}}=\left\{\Theta^{\varphi_i,\chi_i,\phi_i,\vartheta_{i}}(t),{t \in[0, T]}\right\}$ is defined by
$$
\begin{aligned}
		\Theta^{\varphi_i,\chi_i,\phi_i,\vartheta_{i}}(t)=& \exp \left(\int_{0}^{t} \varphi_{i}(s) \mathrm{d} W(s)+\int_{0}^{t} \chi_{i}(s) \mathrm{d}B(s)-\frac{1}{2} \int_{0}^{t}[ \varphi_{i}^2(s) +\chi_{i}^2 (s)]\mathrm{d} s\right)\\
		&\times \exp \left(\int_{0}^{t} \phi_{i}(s) \mathrm{d} \tilde{W}(s)+\int_{0}^{t} \vartheta_{i}(t)^T \mathrm{d}\hat{\textbf{W}}(s)-\frac{1}{2} \int_{0}^{t} [\phi_{i}^2(s) +\vartheta_{i}(t)^T\vartheta_{i}(t)]\mathrm{d} s\right).
\end{aligned}
$$
Then $\Theta^{\varphi_i,\chi_i,\phi_i,\vartheta_{i}}$ is a $\mathbb{P}$-martingale and there is an equivalent probability measure $ \mathbb{Q}^{\varphi_i,\chi_i,\phi_i,\vartheta_{i}}$ defined by
$$
\left.\frac{\mathrm{d} \mathbb{Q}^{\varphi_i,\chi_i,\phi_i,\vartheta_{i}}}{\mathrm{~d} \mathbb{P}}\right|_{\mathcal{F}_{T}}=\Theta^{\varphi_i,\chi_i,\phi_i,\vartheta_{i}}(T) .
$$
It follows from Girsanov's Theorem that the  processes
\begin{equation*}
	\begin{aligned}
	 &{{W}}^{ \mathbb{Q}^{\varphi_i,\chi_i,\phi_i,\vartheta_{i}}}=\left\{{{W}}^{ \mathbb{Q}^{\varphi_i,\chi_i,\phi_i,\vartheta_{i}}}(t),t \in[0, T]\right\},
	 &&B^{ \mathbb{Q}^{\varphi_i,\chi_i,\phi_i,\vartheta_{i}}}=\left\{{{B}}^{ \mathbb{Q}^{\varphi_i,\chi_i,\phi_i,\vartheta_{i}}}(t),t \in[0, T]\right\},\\
	& \tilde{W}^{ \mathbb{Q}^{\varphi_i,\chi_i,\phi_i,\vartheta_{i}}}=\left\{\tilde{W}^{ \mathbb{Q}^{\varphi_i,\chi_i,\phi_i,\vartheta_{i}}}(t),t \in[0, T]\right\},
	&& \hat{\textbf{W}}^{ \mathbb{Q}^{\varphi_i,\chi_i,\phi_i,\vartheta_{i}}}=\left\{\hat{\textbf{W}}^{ \mathbb{Q}^{\varphi_i,\chi_i,\phi_i,\vartheta_{i}}}(t),t \in[0, T]\right\}
\end{aligned}
\end{equation*}
%
%
defined in  (\ref{newBrown}) are standard Brownian motions under $ \mathbb{Q}^{\varphi_i,\chi_i,\phi_i,\vartheta_{i}}$, where
\begin{equation}\label{newBrown}
	\begin{aligned}
	&\mathrm{d} {{W}}^{ \mathbb{Q}^{\varphi_i,\chi_i,\phi_i,\vartheta_{i}}}(t)=\mathrm{d}W(t)-\varphi_i(t) \mathrm{d} t, 
	&&\mathrm{d} B^{ \mathbb{Q}^{\varphi_i,\chi_i,\phi_i,\vartheta_{i}}}(t)=\mathrm{d} B(t)-\chi_{i}(t) \mathrm{d} t,\\
	&\mathrm{d} {\tilde{W}}^{ \mathbb{Q}^{\varphi_i,\chi_i,\phi_i,\vartheta_{i}}}(t)=\mathrm{d}\tilde{W}(t)-\phi_i(t) \mathrm{d} t, 
	&&\mathrm{d} {\hat{\textbf{W}}}^{  \mathbb{Q}^{\varphi_i,\chi_i,\phi_i,\vartheta_{i}}}(t)=\mathrm{d}\hat{\textbf{W}}(t)-\vartheta_{i}(t) \mathrm{d} t.
\end{aligned}
\end{equation}
Under the equivalent probability measure $\mathbb{Q}^{\varphi_i,\chi_i,\phi_i,\vartheta_{i}}$ and the investment-reinsurance strategies $\left\{(\pi_{k},a_k)_{k=1}^{n}\right\}$, the wealth process of AAI $k$ ($k=1,2,\dots,n$) can be expressed as follows:
\begin{equation}\label{SDE-X}
	\begin{aligned}
		&\mathrm{d} X_{k}(t)\\
		=&\left[r X_{k}(t)+\eta_{k}\left(\lambda_{k}\!+\!\hat{\lambda}\right) \!\mu_{k 1}\!-\!\hat{\eta}(1\!-\!a_k(t))^2\left(\lambda_{k}\!+\!\hat{\lambda}\right)\! \mu_{k 2}\!+\!\pi_{k}(t)m\sqrt{Z(t)}\Sigma(t)\right]\!\mathrm{d} t,\\
		&+\!a_{k}(t)\!\left(\!\sqrt{\hat{\lambda}}\mu_{k1}(\mathrm{d} {\tilde{W}}^{ \mathbb{Q}^{\varphi_i,\chi_i,\phi_i,\vartheta_{i}}}\!(t)\!+\!\phi_{i}(t)\rd t)\!+\!\sqrt{(\hat{\lambda}\!+\!\lambda_{k})\mu_{k 2}\!-\!\hat{\lambda}\mu_{k 1}^2}(\mathrm{d} {\hat{\textbf{W}}}_k^{ \mathbb{Q}^{\varphi_i,\chi_i,\phi_i,\vartheta_{i}}}(t)\!+\!\vartheta_{i,k}\!(t) \mathrm{d} t )\!\right)\\
		&+\pi_{k}(t){\Sigma}(t)(\rd W^{ \mathbb{Q}^{\varphi_i,\chi_i,\phi_i,\vartheta_{i}}}(t)+\varphi_i(t)\rd t),
	\end{aligned}
\end{equation}where
\begin{equation}\label{z}
	\rd Z(t)\!=\!{\kappa}(\bar{Z}\!-\!Z(t))\rd t\!+\!{\nu}\sqrt{Z(t)}\left[{\rho}(\rd W^{ \mathbb{Q}^{\varphi_i,\chi_i,\phi_i,\vartheta_{i}}}(t)\!+\!\varphi_i(t)\rd t)\!+\!\sqrt{1\!-\!{\rho}^2}(\rd B^{ \mathbb{Q}^{\varphi_i,\chi_i,\phi_i,\vartheta_{i}}}(t)\!+\!\chi_i(t)\rd t)\right]\!.
\end{equation}

The AAIs compete with one another and are concerned with their relative performance. The average wealth of the AAIs is denoted by $\bar{X}(T)=\frac{1}{n}\sum_{k=1}^{n} X_k(T)$. AAI $i$ is specifically interested in her relative wealth, which is expressed as $$ Y_i(t):=X_{i}(t) {-\theta_{i}}\bar{X}(t)$$ with the initial value $ y_i^0=x^0_i-\frac{\theta_i}{n}\sum_{k=1}^nx^0_k$. $\theta_{i} \in [0,1]$ denotes the risk preference of AAI $i$ concerning her own wealth versus relative wealth. Larger values of $\theta_i$ indicate a greater emphasis on relative performance. When $\theta_i$ is higher, the AAI places more importance on her relative wealth and is more concerned with her performance relative to the peers.

Based on \eqref{SDE-X}, $Y_i=\{Y_i(t),t\in[0,T]\}$ satisfies the following stochastic differential equation
\begin{equation}\label{SDE-y}
	\begin{aligned}
		\mathrm{d} Y_i(t)
		=&r Y_i(t)\rd t+\frac{\theta_i}{n}\sum_{k\neq i}\left[\eta_{k}\left(\lambda_{k}\!+\!\hat{\lambda}\right) \!\mu_{k 1}\!-\!\hat{\eta}(1\!-\!a_k(t))^2\left(\lambda_{k}\!+\!\hat{\lambda}\right)\! \mu_{k 2}\!+\!\pi_{k}(t)m\sqrt{Z(t)}\Sigma(t)\right]\!\mathrm{d} t\\
		&+\left(1-\frac{\theta_i}{n}\right)\left[\eta_{i}\left(\lambda_{i}\!+\!\hat{\lambda}\right) \!\mu_{i 1}\!-\!\hat{\eta}(1\!-\!a_i(t))^2\left(\lambda_{i}\!+\!\hat{\lambda}\right)\! \mu_{i 2}\!+\!\pi_{i}(t)m\sqrt{Z(t)}\Sigma(t)\right]\!\mathrm{d} t\\
		&\!+\!\frac{\theta_i}{n}\!\sum_{k\neq i}a_{k}(t)\!\left(\!\sqrt{\hat{\lambda}}\mu_{k1}(\mathrm{d} {\tilde{W}}^{ \mathbb{Q}^{\varphi_i,\chi_i,\phi_i,\vartheta_{i}}}(t)\!+\!\phi_{i}(t)\rd t)\!+\!\sqrt{(\hat{\lambda}+\lambda_{k})\mu_{k 2}-\hat{\lambda}\mu_{k 1}^2}(\mathrm{d} {\hat{\textbf{W}}}_k^{ \mathbb{Q}^{\varphi_i,\chi_i,\phi_i,\vartheta_{i}}}(t)\!+\!\vartheta_{i,k}(t) \mathrm{d} t )\right)\\
		&\!+\!\left(\!1\!-\!\frac{\theta_i}{n}\right)\!a_{i}(t)\!\left(\!\sqrt{\hat{\lambda}}\mu_{i1}(\mathrm{d} {\tilde{W}}^{ \mathbb{Q}^{\varphi_i,\chi_i,\phi_i,\vartheta_{i}}}(t)\!+\!\phi_{i}(t)\rd t)\!+\!\sqrt{(\hat{\lambda}+\lambda_{i})\mu_{i 2}-\hat{\lambda}\mu_{i 1}^2}(\mathrm{d} {\hat{\textbf{W}}}_i^{ \mathbb{Q}^{\varphi_i,\chi_i,\phi_i,\vartheta_{i}}}(t)\!+\!\vartheta_{i,i}(t) \mathrm{d} t )\right)\\
		&+\frac{\theta_i}{n}\sum_{k\neq i}\pi_{k}(t){\Sigma}(t)(\rd W^{ \mathbb{Q}^{\varphi_i,\chi_i,\phi_i,\vartheta_{i}}}(t)+\varphi_i(t)\rd t)\!+\!\left(1\!-\!\frac{\theta_i}{n}\right)\!\pi_{i}(t){\Sigma}(t)(\rd W^{ \mathbb{Q}^{\varphi_i,\chi_i,\phi_i,\vartheta_{i}}}(t)\!+\!\varphi_i(t)\rd t).
	\end{aligned}
\end{equation}
We now introduce the robust game for AAIs under the mean-variance criterion. For an AAI $i$ who is concerned with relative performance, her preference without ambiguity can be expressed as
$$
\mathbb{E}[Y_i(T)]-\frac{\delta_{i}}{2}\var(Y_i(T)),
$$
where $\delta_{i}>0$ represents AAI $i$'s risk aversion coefficient.

In our work, the AAIs aim to search for the robust strategies under worst-case scenarios. The penalty of $\mathbb{Q}^{\varphi_i,\chi_i,\phi_i,\vartheta_{i}}$, denoted by $\gamma\left( \mathbb{Q}^{\varphi_i,\chi_i,\phi_i,\vartheta_{i}}\right)$, measures the distance between this probability measure and the reference measure $\mathbb{P}$. To account for ambiguity aversion, we use the following relative penalty function:
$$
\gamma\left( \mathbb{Q}^{\varphi_i,\chi_i,\phi_i,\vartheta_{i}}\right)=\frac{1}{2 \Psi_i}\int_{0}^{T} \left[{\varphi_{i}^{2}(t)}+{\chi_{i}^2(t)}+{\phi_{i}^{2}(t)}+{\vartheta_{i}(t)^T\vartheta_{i}(t)}\right] \mathrm{d} t,
$$
where $\Psi_i>0$ is a constant representing  AAI $i$'s level of ambiguity aversion.
 
Then the preference of AAI $i$ at time $t$ given states $ Y_i(t)=y_i$ and $Z(t)=z $ under the equivalent probability measure  $\mathbb{Q}^{\varphi_i,\chi_i,\phi_i,\vartheta_{i}}\in \mathcal{Q}^{i}$ is as follows
\begin{equation}\label{equ:ji}
	\begin{aligned}
&J_{i}\left(\left\{(\pi_{k},a_k)_{k=1}^{n}\right\},\left(\varphi_i,\chi_i,\phi_i,\vartheta_{i}\right),t,y_i,z\right)\\=&\mathbb{E}_{t,y_i,z}^{ \mathbb{Q}^{\varphi_i,\chi_i,\phi_i,\vartheta_{i}}}\left[Y_i(T)\right]
-{\delta_i\over 2}\var_{t,y_i,z}^{ \mathbb{Q}^{\varphi_i,\chi_i,\phi_i,\vartheta_{i}}}\left[Y_i(T)\right]
\\&+
\mathbb{E}_{t,y_i,z}^{ \mathbb{Q}^{\varphi_i,\chi_i,\phi_i,\vartheta_{i}}}\left[ \frac{1}{2 \Psi_i}\int_{t}^{T}  \left[{\varphi_{i}^{2}(s)}+{\chi_{i}^2(s)}+{\phi_{i}^{2}(s)}+{\vartheta_{i}(s)^T\vartheta_{i}(s)}\right] \mathrm{d} s\right], 
	\end{aligned}
\end{equation}where $ \mathbb{E}_{t,y_i,z}^{ \mathbb{Q}^{\varphi_i,\chi_i,\phi_i,\vartheta_{i}}}[\cdot] $ and $ \var_{t,y_i,z}^{ \mathbb{Q}^{\varphi_i,\chi_i,\phi_i,\vartheta_{i}}}[\cdot] $ mean the conditional expectation and the conditional variance under probability measure $ \mathbb{Q}^{\varphi_i,\chi_i,\phi_i,\vartheta_{i}}$ given states $ Y_i(t)=y_i$ and $Z(t)=z $, respectively.

To ensure that the verification theorem holds, we establish the admissible set of strategies and density generators for the AAIs:
\begin{equation*}
	\begin{aligned}
		\mathscr{U}=&\Big\{\{(\pi_k,a_k)_{k=1}^n\}:(\pi_k,a_k) \text{ is progressively measurable w.r.t. }\mathbb{F}\text{ and }0\leqslant a_k(t)\leqslant1, \forall 1\leqslant k\leqslant n,\\
		&\exists \{\mathcal{C}_k\}_{k=1}^n\subseteq\mathbb{R}_+, \sup_{1\leqslant k\leqslant n}\mathcal{C}_k<\infty, \text{ such that } \pi_k(t)=\ell_k(t)\frac{Z(t)}{aZ(t)+b},|\ell_k(t)|\leqslant \mathcal{C}_k, \forall t\in[0,T]\Big\},\\
		\mathscr{U}_i=&\Big\{(\pi_i,a_i):(\pi_i,a_i) \text{ is progressively measurable w.r.t. }\mathbb{F}\text{ and }0\leqslant a_i(t)\leqslant1,\\
		&\exists \mathcal{C}_i\in\mathbb{R}_+, \text{ such that } \pi_i(t)=\ell_i(t)\frac{Z(t)}{aZ(t)+b},|\ell_i(t)|\leqslant \mathcal{C}_i, \forall t\in[0,T]\Big\},\\
		\mathscr{U}_{-i}=&\Big\{\{(\pi_k,a_k)_{k\neq i}\}:(\pi_k,a_k) \text{ is progressively measurable w.r.t. }\mathbb{F}\text{ and }0\leqslant a_k(t)\leqslant1, \forall k\neq i,\\
		&\exists \{\mathcal{C}_k\}_{k\neq i}\subseteq\mathbb{R}_+, \sup_{k\neq i}\mathcal{C}_k<\infty, \text{ such that } \pi_k(t)=\ell_k(t)\frac{Z(t)}{aZ(t)+b},|\ell_k(t)|\leqslant \mathcal{C}_k, \forall t\in[0,T]\Big\},\\
		\mathscr{A}=&\Big\{\left(\varphi, \chi,\phi,\vartheta\right):\left(\varphi, \chi,\phi,\vartheta\right) \text{ is progressively measurable w.r.t. }\mathbb{F}\text{ and }\left(\varphi(t), \chi(t)\right)\\&\quad=\left(\hslash(t)\!\sqrt{Z(t)},\hbar(t)\!\sqrt{Z(t)} \right),
		|\hslash(t)|, |\hbar(t)|\leqslant \mathcal{C},\sup|\phi|<\infty,\sup|\vartheta|<\infty,~ \forall t\in[0,T]\Big\},
	\end{aligned}
\end{equation*} 
where $ \mathcal{C} $ is a constant and $ \mathcal{C}^2<\frac{\kappa^2}{2\nu^2} $.
\begin{remark}
	For $ \left(\varphi, \chi,\phi,\vartheta\right)\in\mathscr{A} $, we see  \begin{equation*}
		\begin{aligned}
			&\mathbb{E}\left[\exp \left(\frac{1}{2} \int_{0}^{T}\left[{\varphi^2(t)}+\chi^2(t)+\phi^2(t)+\vartheta(t)^T\vartheta(t)\right] \mathrm{d} t\right)\right]\\
			=&\mathbb{E}\left[\exp \left(\frac{1}{2} \int_{0}^{T}\left[{\hslash^2(t)}+\hbar^2(t)\right]Z(t) \mathrm{d} t+\frac{1}{2} \int_{0}^{T}[\phi^2(t)+\vartheta(t)^T\vartheta(t)] \mathrm{d} t\right)\right]\\
			\leqslant&\mathbb{E}\left[\exp\left(\frac{T}{2}\sup|\phi|^2+\frac{T}{2}\sup\|\vartheta\|^2\right)\times\exp \left( \mathcal{C}^2\int_{0}^{T}Z(t) \mathrm{d} t\right)\right].
		\end{aligned}
	\end{equation*}
	By Proposition 5.1 in \cite{Kraft}, if condition\begin{equation*}
		\mathcal{C}^2<\frac{\kappa^2}{2\nu^2}
	\end{equation*}is satisfied, then \begin{equation*}
		\mathbb{E}\left[\exp \left( \mathcal{C}^2\int_{0}^{T}Z(t) \mathrm{d} t\right)\right]<\infty.
	\end{equation*}
	Thus \begin{equation*}
		\mathbb{E}\left[\exp \left(\frac{1}{2} \int_{0}^{T}\left[{\varphi^2(t)}+\chi^2(t)+\phi^2(t)+\vartheta(t)^T\vartheta(t)\right] \mathrm{d} t\right)\right]<\infty.
	\end{equation*}
	The Novikov's condition holds and $\mathbb{Q}^{\varphi,\chi,\phi,\vartheta}$ is a well-defined probability  measure which is equivalent to $ \mathbb{P} $.
\end{remark}

 Denote the objective function of AAI $i$ when AAI $k$ takes strategy $ (\pi_{k},a_k), \forall 1\leqslant k\leqslant n $, under the worst-case scenario as 
$$
\bar{J}_{i}\left(\left\{(\pi_{k},a_k)_{k=1}^{n}\right\},t,y_i,z\right)=\inf_{\left(\varphi_i,\chi_i,\phi_i,\vartheta_{i}\right)\in \mathscr{A}}J_{i}\left(\left\{(\pi_{k},a_k)_{k=1}^{n}\right\},\left(\varphi_i,\chi_i,\phi_i,\vartheta_{i}\right),t,y_i,z\right).
$$
If \begin{equation*}
\left(\varphi_i^\dag,\chi_i^\dag,\phi_i^\dag,\vartheta_{i}^\dag\right)=\arg	\inf_{\left(\varphi_i,\chi_i,\phi_i,\vartheta_{i}\right)\in \mathscr{A}}J_{i}\left(\left\{(\pi_{k},a_k)_{k=1}^{n}\right\},\left(\varphi_i,\chi_i,\phi_i,\vartheta_{i}\right),t,y_i,z\right),
\end{equation*}then we call $ \left(\varphi_i^\dag,\chi_i^\dag,\phi_i^\dag,\vartheta_{i}^\dag\right) $ the worst-case scenario density generator associated with the strategy.

Then  the goal of AAI $i$ is to search the robust investment-reinsurance response strategies when the peers' strategies are given under the worst-case scenario, which is formulated as follows:
\begin{equation}\nonumber
\sup_{(\pi_i,a_i)\in\mathscr{U}_i}\bar{J}_{i}\left(\left\{(\pi_{k},a_k)_{k=1}^{n}\right\},t,y_i,z\right).
\end{equation}

 The inclusion of the variance operator in the objective function introduces time-inconsistency, posing challenges for the application of dynamic programming methods. To address this issue, we adopt the approach proposed in \cite{bjork2017time}. This method facilitates the derivation of time-consistent strategies for the AAIs under the mean-variance criterion.
We  first introduce the concept of the robust time-consistent  response strategy for each AAI, which is widely called  ``weak equilibrium strategy'', see \cite{EkelandIvar}, \cite{ZhouZhou} and \cite{YFY}. However, it's essential to clarify that this strategy is defined when the strategies of others are given. To avoid any misunderstandings regarding the Nash equilibrium of the $n$-insurer game and mean-field equilibrium of the mean-field game, we specifically refer to this strategy as the robust time-consistent response strategy or simply the best response strategy.
\begin{definition}{\label{def}} For any fixed $(t,y_i,z)\in [0,T]\times \mathbb{R}\times \mathbb{R}_+$, consider an admissible  strategy $ (\pi_i^\dag,a_i^\dag) $ for AAI $ i $ given that AAI $k$ takes admissible strategy $ (\pi_{k},a_k), k\neq i $. For any $(\pi_i,a_i)\in \mathscr{U}_i$ and $h>0$,  define a new $ h $-perturbed strategy $(\pi_i^h,a_i^h)$ by
	\begin{equation*}
		(\pi_{i}^h(s),a_{i}^h(s))=\left\{
		\begin{aligned}
			&(\pi_i(s),a_i(s)),&&t\leqslant s<t+h,
			\\&(\pi^\dag_i(s),a_i^\dag(s)),&& t+h\leqslant s\leqslant T.
		\end{aligned}
		\right.
	\end{equation*}
	If for all $(\pi_i,a_i)\in \mathscr{U}_i$,   we have 
	\begin{eqnarray}{\label{mv1}}
		\mathop {\lim \inf }\limits_{h\rightarrow 0^+} \frac{\bar{J}_{i}\left(\left\{(\pi_{k},a_k)_{k\neq i},(\pi_{i}^\dag,a_i^\dag)\right\},t,y_i,z\right)-\bar{J}_{i}\left(\left\{(\pi_{k},a_k)_{k\neq i},(\pi_i^h,a_i^h)\right\},t,y_i,z\right)}{h}\geqslant 0,
	\end{eqnarray}
	then  $ (\pi_i^\dag,a_i^\dag) $  is the \textbf{robust time-consistent  investment-reinsurance {response} strategy} or simply the \textbf{best response strategy} for AAI $ i $ and the robust time-consistent response value function of AAI $i$ is given by $\bar{J}_{i}\left(\left\{(\pi_{k},a_k)_{k\neq i},(\pi_{i}^\dag,a_i^\dag)\right\},t,y_i,z\right)$.
\end{definition}

Next, we  present the definition of a robust Nash equilibrium game for the $n$-insurer game under insurance and volatility risks for the AAIs.
\begin{definition}\label{def1}
	We say that   $\left\{(\pi_{k}^*,a_k^*)_{k=1}^{n}\right\}\in\mathscr{U}$ is a {\textbf{robust time-consistent Nash equilibrium}} or simply the \textbf{robust equilibrium} if $\forall 1\leqslant i\leqslant n$, $ (\pi_{i}^*,a_i^*) $ is the best response strategy for AAI $ i $ when AAI $k$ takes strategy $ (\pi^*_{k},a^*_k), k\neq i $. We call $ (\pi_{i}^*,a_i^*) $ the \textbf{robust  equilibrium strategy} for AAI $ i $ and  and denote the robust time-consistent value function of of AAI $i$ as $V^{(i)}(t,y,z):= \bar{J}_{i}\left(\left\{(\pi_{k}^*,a_k^*)_{k=1}^{n}\right\},t,y,z\right) $.
\end{definition}


\section{\bf Robust $n$-insurer game}

In this section, we focus on obtaining the robust equilibrium strategies for the AAIs. Based on  \cite{bjork2017time} and \cite{Chi2018}, we derive the $n$-deimension extended HJBI equations for the robust mean-variance game. To begin with, we introduce the infinitesimal generator $\mathcal{A}^{\{(\pi_k,a_k)\}_{k=1}^n,(\varphi_i,\chi_i,\phi_i,\vartheta_{i})}$ of system \eqref{z}--\eqref{SDE-y}  for all $ (t,y,z) \in[0,T]\times\mathbb{R}\times\mathbb{R}_+$ and $f(t,y,z)\in C^{1,2,2}([0,T]\times\mathbb{R}\times\mathbb{R}_+)$, 
\begin{equation}\label{equ:af}
\hspace{-5pt}	\begin{aligned}
		&\mathcal{A}^{\{(\pi_k,a_k)\}_{k=1}^n,(\varphi_i,\chi_i,\phi_i,\vartheta_{i})}f(t,y,z)\\
		=&f_t\!+\!\left[{\kappa}(\bar{Z}\!-\!z)+{\nu}\sqrt{z}({\rho}\varphi_i\!+\!\sqrt{1-{\rho}^2}\chi_i)\right]f_{z}\!+\!\nu^2\frac{1}{2}zf_{zz}\!+\!\nu\sqrt{z}\rho\sigma((1\!-\!\frac{\theta_{i}}{n})\pi_i\!-\!\frac{\theta_{i}}{n}\sum_{k\neq i}\pi_{k}(t))f_{yz}+ryf_y\\
		&+(1-\frac{\theta_i}{n})\left[\eta_{i}\left(\lambda_{i}+\hat{\lambda}\right) \mu_{i 1}-\hat{\eta}(1-a_i(t))^2\left(\lambda_{i}+\hat{\lambda}\right) \mu_{i 2}+\pi_{i}(t)\left(m+\frac{\varphi_i(t)}{\sqrt{z}}\right)({a}{z}+{{b}})\right] f_y\\
		&-\frac{\theta_{i}}{n}\sum_{k\neq i}\left[\eta_{k}\left(\lambda_{k}+\hat{\lambda}\right) \mu_{k 1}-\hat{\eta}(1-a_k(t))^2\left(\lambda_{k}+\hat{\lambda}\right) \mu_{k 2}+\pi_{k}(t)\left(m+\frac{\varphi_i(t)}{\sqrt{z}}\right)({a}{z}+{{b}})\right]f_y\\
		&+(1-\frac{\theta_i}{n})a_{i}(t)\left(\sqrt{\hat{\lambda}}\mu_{i1}\phi_{i}(t)+\sqrt{(\hat{\lambda}+\lambda_{i})\mu_{i 2}-\hat{\lambda}\mu_{i 1}^2}\vartheta_{i,i}(t)\right)f_y\\
		&-\frac{\theta_{i}}{n}\sum_{k\neq i}a_{k}(t) \left(\sqrt{\hat{\lambda}}\mu_{k1}\phi_{i}(t)+\sqrt{(\hat{\lambda}+\lambda_{k})\mu_{k 2}-\hat{\lambda}\mu_{k 1}^2}\vartheta_{i,k}(t) \right)f_y\\
		&+\!\frac{1}{2}(1\!-\!\frac{\theta_i}{n})^2a_{i}^2(t) \left(\hat{\lambda}\!+\!\lambda_{i}\right) \mu_{i 2}f_{yy}\!+\!\frac{\theta_{i}^2}{2n^2}\sum_{k\neq i}(a_{k}(t))^2 \left(\left(\hat{\lambda}\!+\!\lambda_{k}\right) \mu_{k 2}-\hat{\lambda}\mu_{k1}^2\right)f_{yy}\\
		&\!+\!\hat{\lambda}\frac{\theta_i^2}{2n^2}\!\left[\sum_{k\neq i}a_k(t)\mu_{k 1}\right]^2\!\!\!\!f_{yy}\!-\!\hat{\lambda}\frac{\theta_{i}}{n}(1\!-\!\frac{\theta_i}{n})a_i(t)\mu_{i 1}\sum_{k\neq i}a_{k}(t)\mu_{k 1}f_{yy}\!+\!\frac{1}{2}\sigma^2((1\!-\!\frac{\theta_{i}}{n})\pi_i\!-\!\frac{\theta_{i}}{n}\sum_{k\neq i}\pi_{k}(t))^2f_{yy}.
	\end{aligned}
\end{equation} 
Denote\begin{equation*}
	\begin{aligned}
		&	\mathcal{L}_{i}\left({\{(\pi_k,a_k)_{k=1}^n\},(\varphi,\chi,\phi,\vartheta)},f,g, h,(t,y,z)\right)\\=&\mathcal{A}^{\{(\pi_k,a_k)_{k=1}^n\},(\varphi,\chi,\phi,\vartheta)}f(t,y,z)+  \frac{\varphi^{2}(t)}{2 h} +\frac{\chi^2(t)}{2 h} +  \frac{\phi^{2}(t)}{2 h} +\frac{\vartheta(t)^T\vartheta(t)}{2 h} \\
		&+\delta_{i} g(t,y,z) \mathcal{A}^{\{(\pi_k,a_k)_{k=1}^n\},(\varphi,\chi,\phi,\vartheta)}g(t,y,z)-\frac{\delta_{i}}{2}\mathcal{A}^{\{(\pi_k,a_k)_{k=1}^n\},(\varphi,\chi,\phi,\vartheta)}g^2(t,y,z).
	\end{aligned}
\end{equation*}

Similar to the derivation in \cite{bjork2017time} and \cite{wang2022robust}, we introduce the $n$-dimensional extended HJBI equations for the robust time-consistent equilibrium strategy and value function, as defined in Definition \ref{def}. The details of the derivation are omitted for brevity. 
\begin{proposition}[extended HJBI equations]\label{prop:hjbi}
	The extended HJBI equations of AAI $i$ for the two-tuple $ (V,\Upsilon) $ where $ V\in C^{1,2,2}([0,T]\times\mathbb{R}\times\mathbb{R}_+) $ and $ \Upsilon \in C^{1,2,2}([0,T]\times\mathbb{R}\times\mathbb{R}_+)$ are defined as
	\begin{equation}\label{HJBI-MV}
		\begin{aligned}
			\sup _{(\pi_{i},a_i)\in\mathscr{U}_i} \inf _{\left(\varphi_i,\chi_i,\phi_i,\vartheta_{i}\right)\in\mathscr{A}}&\mathcal{L}_{i}\left({\left\{(\pi_{k}^*,a_k^*)_{k\neq i},(\pi_{i},a_i)\right\},(\varphi_i,\chi_i,\phi_i,\vartheta_{i})},V,\Upsilon,\Psi_i,(t,y,z)\right)  =0
		\end{aligned}
	\end{equation}
	with boundary condition $ V(T,y,z)=y $
	and \begin{equation}\label{g-MV}
		\mathcal{A}^{\left\{(\pi_{k}^*,a_k^*)_{k\neq i},(\pi_{i}^\circ,a_i^\circ)\right\},(\varphi_i^\circ,\chi_i^\circ,\phi_i^\circ,\vartheta_{i}^\circ)}\Upsilon(t,y,z)=0
	\end{equation}with boundary condition $ \Upsilon(T,y,z)=y $, where\begin{equation}\label{HJBI-opt-MV}
		\begin{aligned}
		&	(\pi^\circ_i,a^\circ_i)=\arg	\sup _{(\pi_i,a_i)\in\mathscr{U}_i} \inf _{(\varphi_i,\chi_i,\phi_i,\vartheta_{i})\in\mathscr{A}}\mathcal{L}_{i}\left({\left\{(\pi_{k}^*,a_k^*)_{k\neq i},(\pi_{i},a_i)\right\},(\varphi_i,\chi_i,\phi_i,\vartheta_{i})},V,\Upsilon,\Psi_i,(t,y,z)\right),\\
		&	( \varphi^\circ_{i},\chi^\circ_{i},\phi_i^\circ,\vartheta_{i}^\circ)=\arg	 \inf _{(\varphi_i,\chi_i,\phi_i,\vartheta_{i})\in\mathscr{A}}\mathcal{L}_{i}\left({\left\{(\pi_{k}^*,a_k^*)_{k\neq i},(\pi_{i}^\circ,a_i^\circ)\right\},(\varphi_i,\chi_i,\phi_i,\vartheta_{i})},V,\Upsilon,\Psi_i,(t,y,z)\right).
		\end{aligned}
	\end{equation}
\end{proposition}

We call $ 	(\pi^\circ_i,a^\circ_i) $ the candidate response strategy of AAI $ i $ and call $ ( \varphi^\circ_{i},\chi^\circ_{i},\phi_i^\circ,\vartheta_{i}^\circ) $ the worst-case scenario density generator  associated with the candidate  strategy  of AAI $ i $. Assuming that the two-tuple $ (v^{(i)},\varUpsilon^{(i)}) $ satisfies (\ref{HJBI-MV}) and (\ref{g-MV}), we refer to $ (v^{(i)},\varUpsilon^{(i)}) $ as the candidate value function tuple, with $ v^{(i)}$ being the candidate value function.

To begin with, we will use the extended HJBI equations \eqref{HJBI-MV}-\eqref{g-MV} to derive the candidate response strategies, the associated worst-case scenario density generators, and the candidate value function, as outlined in Theorem \ref{solution-HJBI}. Subsequently, we proceed to prove that under specific conditions the candidate response strategies are the robust time-consistent equilibrium strategies in Theorem \ref{Verification}.

We first list some compatible conditions about the parameters.
Let \begin{equation*}
	\left\{\begin{aligned}
		&K_{i,1}=\diag(-\frac{1}{2}\nu^2(\rho^2\frac{\Psi_i\delta_{i}}{\Psi_i+\delta_{i}}+(1-\rho^2)\Psi_i), \nu^2\rho^2\frac{\delta_{i} \Psi_i}{(\Psi_i+{\delta_{i}})^2}),\\
		&K_{i,2}=\diag(-\frac{1}{2}\nu^2\delta_{i}(1-\rho^2\frac{\delta_{i}}{\Psi_i+\delta_{i}} ),\nu^2\rho^2\Psi_i\frac{\delta_{i} \Psi_i}{(\Psi_i+{\delta_{i}})^2}), \\
		&K_{i,1,2}=\diag(-\nu^2\left[ \rho^2\Psi_i\frac{2\Psi_i{\delta_{i}}}{(\Psi_i+{\delta_{i}})^2}+(1-\rho^2)\Psi_i\right],\nu^2\rho^2\frac{\Psi_i\delta_{i}}{\Psi_i+\delta_{i}}),\\
		&\bar{B}_{i,1}=\diag(-(\kappa+m\nu\rho\frac{\Psi_i}{\Psi_i+\delta_{i}}), -\left[\kappa +m\nu\rho\frac{{\delta_{i}}^2+(\Psi_i)^2}{(\Psi_i+{\delta_{i}})^2}\right]), \\
		&\bar{B}_{i,2}=\diag (-m\nu\rho\frac{\delta_{i}}{\Psi_i+\delta_{i}},-m\nu\rho\frac{2\Psi_i{\delta_{i}}}{(\Psi_i+{\delta_{i}})^2}), \\
		&G_i=\diag(\frac{m^2}{2(\Psi_i+\delta_{i})}, \frac{\delta_{i} m^2}{(\Psi_i+\delta_{i})^2}),
	\end{aligned}\right.
\end{equation*}

and $ H=\begin{pmatrix}
	0&1\\1&0
\end{pmatrix} $, where $ \diag(a,b)= \begin{pmatrix}
	a&0\\0&b
\end{pmatrix}$. Define 
\begin{equation*}
\alpha_{i,1}=\|K_{i,1}\|_{\sup}+\|K_{i,2}\|_{\sup}+\|K_{i,1,2}\|_{\sup},\quad\alpha_{i,2}=\|\bar{B}_{i,1}\|_{\sup}+\|\bar{B}_{i,2}\|_{\sup},\quad	\alpha_{i,3}=\|G_{i}\|_{\sup}
\end{equation*}
and
\begin{equation*}
 \Delta_i = \alpha_{i,2}^2-4\alpha_{i,1}\alpha_{i,3} ,\quad \varsigma_{i,1}=\frac{-\alpha_{i,2}+\sqrt{\Delta_i}}{2\alpha_{i,1}} ,\quad \varsigma_{i,2}=\frac{-\alpha_{i,2}-\sqrt{\Delta_i}}{2\alpha_{i,1}},
\end{equation*}
where $ \|\diag(a,b)\|_{\sup}= \max\{|a|,|b|\}$.

Then we have the following results.
\begin{proposition}\label{ricc}
$\forall 1\leqslant i\leqslant n$, if one of the following three compatible conditions holds, then  we can define the corresponding $ U_i(t) $ and we will find that the matrix Riccati equation (\ref{mat-ricc}) admits a solution $ F_i$ with $ \|F_i(t)\|_{\sup}\leqslant U_i(t) $ for $ t\in[0,T] $.\begin{equation}\label{mat-ricc}
		\left\{\begin{aligned}
			&\frac{\rd F_i}{\rd t}=F_iK_{i,1}F_{i}+HF_iHK_{i,2}F_iH+\bar{B}_{i,1}F_i+HF_iH\bar{B}_{i,2}+HF_iK_{i,1,2}HF_i+ G_i,\\&F_i(0)=\diag(0,0).
		\end{aligned}\right.
	\end{equation}
\begin{itemize}
	\item [(i)]If $ \Delta_i =0 $ and $ T<\frac{2}{\alpha_{i,2}} $, then define 
	\begin{equation*}
		 U_i(t)=\varsigma_{i,1}+\frac{\alpha_{i,2}}{\alpha_{i,1}(2-\alpha_{i,2}t)}.
	\end{equation*}
	\item [(ii)]If $ \Delta_i >0 $ and $ T<\frac{1}{\sqrt{\Delta_i}}\log(\frac{\varsigma_{i,2}}{\varsigma_{i,1}}) $, then define
	\begin{equation*}
	U_i(t)=\frac{\alpha_{i,3}}{\alpha_{i,1}}\frac{1-e^{\alpha_{i,1}(\varsigma_{i,1}-\varsigma_{i,2})t}}{\varsigma_{i,2}-\varsigma_{i,1}e^{\alpha_{i,1}(\varsigma_{i,1}-\varsigma_{i,2})t}}.
	\end{equation*}
	\item [(iii)]If $ \Delta_i <0 $ and $ T<\frac{1}{\sqrt{|\Delta_i|}}(\pi+2\arctan(\frac{\Re{\varsigma_{i,1}}}{\Im{\varsigma_{i,1}}})) $, where $ \Re{\varsigma_{i,1}} $  represents the real part of $ \varsigma_{i,1} $ and $ \Im{\varsigma_{i,1}} $ represents the imaginary part of $ \varsigma_{i,1} $, then define
	\begin{equation*}
		U_i(t)=\Re{\varsigma_{i,1}}+\Im{\varsigma_{i,1}}\tan\left(\alpha_{i,1}t\Im{\varsigma_{i,1}}-\arctan\left(\frac{\Re{\varsigma_{i,1}}}{\Im{\varsigma_{i,1}}}\right)\right).
	\end{equation*}
\end{itemize}
\end{proposition}
\begin{proof}
 (\ref{mat-ricc}) is a matrix Riccati equation with constant coefficients, it is easy to see  \begin{equation*}
		\frac{	\rd\|F_i(t)\|_{\sup}}{\rd t}\leqslant \alpha_{i,1}\|F_i(t)\|_{\sup}^2+\alpha_{i,2}\|F_i(t)\|_{\sup}+\alpha_{i,3}.
	\end{equation*}
Then based on  \cite{Papavassilopoulos1979}, this proposition holds.
\end{proof}
Proposition \ref{ricc} will play an important role in ensuring the existence of the solution in later calculations.
Based on this proposition, we see  \begin{equation*}
	\max_{t\in[0,T]}\|F_i(t)\|_{\sup}\leqslant U_i(T).
\end{equation*}

Next, we solve the extended HJBI equations in Proposition \ref{prop:hjbi}.
\begin{theorem}[Solution to the extended HJBI equations]\label{solution-HJBI}
	$ \forall 1\leqslant i\leqslant n $, suppose that the parameters satisfy one of the  three compatible conditions in Proposition \ref{ricc} and
	\begin{equation}\label{condition1}
		\sup_{1\leqslant i\leqslant n}m+\nu\rho U_i(T)\Psi_i<\frac{\kappa}{\sqrt{2}\nu},	\quad\sup_{1\leqslant i\leqslant n}\nu\sqrt{1-\rho^2} U_i(T)\Psi_i<\frac{\kappa}{\sqrt{2}\nu}.
	\end{equation}
When other AAI $k$ ($k\neq i$) adopts strategy $\left({\pi}^*_k, {a}^*_k\right)\in\mathscr{U}_k$, the  candidate response investment-reinsurance strategies are
 \begin{equation}\label{solution:pia}
	\left\{\begin{aligned}
		&{\pi}^\circ_i(t)=\frac{\theta_i}{n-\theta_i}\sum_{k\neq i}\pi^*_k+ \frac{ n\sqrt{Z(t)}}{(n-\theta_i)\Sigma(t)}S_i(t),\\
		&a_i^\circ(t)=(\hat{a}^\circ_i(t)\vee0)\wedge1\\
		&\qquad=\left(\left(\frac{Q_i(t)}{R_i(t)}\frac{1}{n}\sum_{k\neq i}\mu_{k 1}a_k^*(t)+\frac{P_i(t)}{R_i(t)}\right)\vee0\right)\wedge1\\
		&\qquad=\left(\frac{Q_i(t)}{R_i(t)}\frac{1}{n}\sum_{k\neq i}\mu_{k 1}a_k^*(t)+\frac{P_i(t)}{R_i(t)}\right)\wedge1,
	\end{aligned}\right.
\end{equation}
	and the associated worst-case scenario density generators $ (\varphi_i^\circ,\chi_i^\circ,\phi_i^\circ,\vartheta_{i}^\circ) $ are given as follows:\begin{equation*}
	\left\{\begin{aligned}
		&{\varphi^\circ_{i}}(t)=-(\nu \rho v_{i,3}+S_i(t) v_{i,2})\Psi_i\sqrt{Z(t)},\\
		&{\chi^\circ_{i}}(t)=-\nu\sqrt{1-\rho^2} v_{i,3}\Psi_i\sqrt{Z(t)},\\
		&\phi^\circ_{{i}}(t)=-\sqrt{\hat{\lambda}}\left[(1-\frac{\theta_i}{n})\mu_{i1}a_{i}^\circ(t)-\frac{\theta_{i}}{n}\sum_{k\neq i}\mu_{k1}a_{k}^*(t)\right]v_{i,2}\Psi_i,\\
		&\vartheta^\circ_{i,i}(t)=-(1-\frac{\theta_i}{n})a_{i}^\circ(t)\sqrt{(\hat{\lambda}+\lambda_{i})\mu_{i 2}-\hat{\lambda}\mu_{i 1}^2}v_{i,2}\Psi_i,\\
		&\vartheta^\circ_{{i,k}}(t)=\frac{\theta_i}{n}a_{k}^*(t)\sqrt{(\hat{\lambda}+\lambda_{k})\mu_{k 2}-\hat{\lambda}\mu_{k 1}^2}v_{i,2}\Psi_i,~k\neq i,
	\end{aligned}\right.
\end{equation*}
	where 
	 \begin{equation}\label{substitute}
		\left\{\begin{aligned}
			&S_i(t)=\frac{m -\nu\rho( v_{i,3}\Psi_i+\delta_{i} \varUpsilon_{i,3})}{(\Psi_i+{\delta_{i}})e^{r(T-t)}},\\
			&R_i(t)=(\lambda_i+\hat{\lambda})\mu_{i2}((1-\frac{\theta_i}{n})(\delta_{i}\varUpsilon_{i,2}^2+\Psi_iv_{i,2}^2)+2\hat{\eta}v_{i,2}),\\
			&Q_i(t)=\hat{\lambda}\theta_i\mu_{i 1}(v_{i,2}^2\Psi_i+\delta_{i}\varUpsilon_{i,2}^2),\\& P_i(t)=2\hat{\eta}(\lambda_i+\hat{\lambda})\mu_{i 2}v_{i,2}.
	\end{aligned}\right.\end{equation}
	Besides, the   candidate value function tuple $ (v^{(i)},\varUpsilon^{(i)}) $  is given by\begin{equation*}
		v^{(i)}(t, y,z)= v_{i,1}(t)+ yv_{i,2}(t)+v_{i,3}(t)z,\quad \varUpsilon^{(i)}(t, y,z)= \varUpsilon_{i,1}(t)+ y\varUpsilon_{i,2}(t)+\varUpsilon_{i,3}(t)z,
	\end{equation*}where 
	\begin{align}
		&~v_{i,2}(t)=\varUpsilon_{i,2}(t)=e^{r(T-t)}, \qquad \diag(v_{i,3}(t),\varUpsilon_{i,3}(t) )=F_i(T-t),\label{F_i}
	\end{align}
and $ F_i(\cdot) $ is the solution to (\ref{mat-ricc}). By utilizing \eqref{varUpsilon}, we can obtain $\varUpsilon_{i,1}(t)$. Similarly, by using  \eqref{v1}, we can obtain $v_{i,1}(t)$. Then the   candidate value function tuple follows.
\end{theorem}
\begin{proof}
	See Appendix \ref{proof-solution-HJBI}.
\end{proof}

Theorem \ref{solution-HJBI} outlines the candidate response strategies, the associated worst-case scenario density generators, and the candidate value function tuples. However, to prove that $(\pi^*_i,a^i)=(\pi^\circ_i,a^\circ_i)$, $ \left(\varphi_{i}^*, \chi_{i}^*,\phi_i^*,\vartheta_{i}^*\right)= \left(\varphi_i^\circ,\chi_i^\circ,\phi_i^\circ,\vartheta_{i}^\circ\right)$, and $V^{(i)}(t, y,z)=v^{(i)}(t, y,z)$, it is imperative to verify that the candidate response strategies, associated worst-case scenario density generators, and candidate value function indeed represent the robust equilibrium strategies, associated worst-case scenario density generators, and value function, respectively. This verification is presented in the following theorem.

\begin{theorem}[Verification theorem]\label{Verification}

Assuming that the parameters meet one of the three compatible conditions in Proposition \ref{ricc} and (\ref{condition1}) holds,
	then we have $(\pi^*_i,a^*_i)=(\pi^\circ_i,a^\circ_i)$, $ \left(\varphi_{i}^*, \chi_{i}^*,\phi_i^*,\vartheta_{i}^*\right)= \left(\varphi_i^\circ,\chi_i^\circ,\phi_i^\circ,\vartheta_{i}^\circ\right)$, $ V^{(i)}(t, y,z)=v^{(i)}(t, y,z) $, $  \forall 1\leqslant i\leqslant n $, i.e., the candidate solutions obtained in Theorem \ref{solution-HJBI} indeed solve the robust $n$-insurer game.  
\end{theorem}
\begin{proof}
See Appendix \ref{proof-Verification}.
\end{proof}

Using Theorem \ref{Verification}, we can determine the optimal response strategy for AAI $i$, denoted by \eqref{solution:pia}. By solving these equations simultaneously for $1 \leqslant i \leqslant n$, we  obtain the robust equilibrium strategies for all the AAIs. 
\begin{theorem}\label{thm1}
	If $n\neq\sum_{k=1}^n\theta_k$, the robust $n$-insurer equilibrium  exists. The equilibrium investment  $ \pi^*_i$, $ \forall 1\leqslant i\leqslant n $, are given by 
	\begin{equation}\label{n-opt}
		\pi_i^*(t)=\left(S_i(t)+\theta_i\frac{\sum_{k=1}^nS_k(t)}{n-\sum_{k=1}^n\theta_k}\right)\frac{Z(t)}{aZ(t)+b}
	\end{equation}
and the equilibrium reinsurance proportions $ a^*_i $, $ \forall 1\leqslant i\leqslant n $, are  determined by 
\begin{equation}\label{a_star}
	a_i^*(t)=\left(\frac{Q_i(t)}{R_i(t)}\frac{1}{n}\sum_{k\neq i}\mu_{k 1}a_k^*(t)+\frac{P_i(t)}{R_i(t)}\right)\wedge1.
\end{equation}
The associated worst-case model uncertainty functions $ \left(\varphi_{i}^*, \chi_{i}^*,\phi_i^*,\vartheta_{i}^*\right)$,  $ \forall 1\leqslant i\leqslant n $, are given by 
\begin{equation*}
	\left\{\begin{aligned}
		&{\varphi^*_{i}}(t)=-(\nu \rho v_{i,3}+S_i(t) v_{i,2})\Psi_i\sqrt{Z(t)},\\
		&{\chi^*_{i}}(t)=-\nu\sqrt{1-\rho^2} v_{i,3}\Psi_i\sqrt{Z(t)},\\
		&\phi^*_{{i}}(t)=-\sqrt{\hat{\lambda}}\left[(1-\frac{\theta_i}{n})\mu_{i1}a_{i}^*(t)-\frac{\theta_{i}}{n}\sum_{k\neq i}\mu_{k1}a_{k}^*(t)\right]v_{i,2}\Psi_i,\\
		&\vartheta^*_{i,i}(t)=-(1-\frac{\theta_i}{n})a_{i}^*(t)\sqrt{(\hat{\lambda}+\lambda_{i})\mu_{i 2}-\hat{\lambda}\mu_{i 1}^2}v_{i,2}\Psi_i,\\
		&\vartheta^*_{{i,k}}(t)=\frac{\theta_i}{n}a_{k}^*(t)\sqrt{(\hat{\lambda}+\lambda_{k})\mu_{k 2}-\hat{\lambda}\mu_{k 1}^2}v_{i,2}\Psi_i,~k\neq i,
	\end{aligned}\right.
\end{equation*}
where $ S_i(t) $, $ R_i(t) $,  $ Q_i(t) $, $ P_i(t) $,  $ v_{i,j}(t) $ and  $ \varUpsilon_{i,j}(t), 1\leqslant j\leqslant 3 $, are given in Theorem  \ref{solution-HJBI}.

Moreover, if $n=\sum_{k=1}^n\theta_k$, the robust $n$-insurer equilibrium does not exist.
\end{theorem}

\eqref{n-opt} can be rewritten as
	\begin{equation}\label{dn-opt}
		\pi_i^*(t)=\frac{e^{-r(T-t)}Z(t)}{aZ(t)+b}
		\left[ 
		\frac{m}{\Psi_i+\delta_{i}}-{\nu\rho( v_{i,3}\Psi_i+\delta_{i} \varUpsilon_{i,3})\over \Psi_i+{\delta_{i}}}+{\theta_i\over n-\sum_{k=1}^n\theta_k}\sum_{k=1}^n\frac{m}{\Psi_k+\delta_{k}}-{\theta_i\over n-\sum_{k=1}^n\theta_k}\sum_{k=1}^n{\nu\rho( v_{k,3}\Psi_k+\delta_{k} \varUpsilon_{k,3})\over \Psi_k+{\delta_{k}}}
		\right].
	\end{equation}
We see that the robust equilibrium investment strategy consists of four parts: a myopic demand as the Merton's fraction, a hedging demand, a myopic demand caused by competition and a hedging demand caused by competition. Particularly, the first two terms rely on the individual's risk aversion and ambiguity aversion coefficients. The last two terms rely on all the insurers' competition coefficient,  risk aversion and ambiguity aversion coefficients. In the case of no competition, the last two terms vanish.  \eqref{a_star} shows that the robust equilibrium reinsurance strategy is composed of a part ignoring competition and a part caused by competition. 

\begin{remark}\label{Randomness}
	Based on Theorems \ref{solution-HJBI} and \ref{Verification}, we see that the randomness of $ \pi^*_{i}(t) $ only depends on $ \frac{{Z(t)}}{a{Z(t)}+b} $,  the randomness of $ {\varphi^*_{i}(t)} $ and $ {\chi^*_{i}(t)} $ depends on $ Z(t) $ and $ {\chi^*_{i}(t)} $  is  proportional to $ \sqrt{Z(t)} $.
	 $ \phi_i^*(t),\vartheta_{i}^*(t) $ and $ a_i^*(t) $ are deterministic.	However, $ a_i^*(t) $ is implicit, $ \phi_i^*(t)$ and $\vartheta_{i}^*(t) $ rely on $ a_i^*(t) $,  we can get the approximate  numerical solution by the following procedure.
	
	First solve \eqref{a_i}:
	\begin{equation}\label{a_i}
		\check{a}_i(t)=\frac{Q_i(t)}{R_i(t)}\frac{1}{n}\sum_{k\neq i}\mu_{k 1}\check{a}_k(t)+\frac{P_i(t)}{R_i(t)},\quad \forall1\leqslant i\leqslant n,~~ t\in[0,T].
	\end{equation}
	\eqref{a_i} is equivalent to 
	\begin{equation}\label{a_i_2}
		\mu_{i 1}	\check{a}_i(t)=\frac{n\mu_{i 1}{Q_i(t)}}{{\mu_{i1}Q_i(t)}+nR_i(t)}\frac{1}{n}\sum_{k=1}^n\mu_{k 1}\check{a}_k(t)+\frac{n\mu_{i 1}{P_i(t)}}{{\mu_{i1}Q_i(t)}+nR_i(t)},\quad \forall1\leqslant i\leqslant n,~~ t\in[0,T].
	\end{equation}
	The explicit solution to \eqref{a_i_2} is given by
	\begin{equation*}
		\check{a}_i(t)=\frac{n{Q_i(t)}}{{\mu_{i1}Q_i(t)}+nR_i(t)}\frac{\sum_{k=1}^n\frac{n\mu_{k 1}{P_k(t)}}{{\mu_{k1}Q_k(t)}+nR_k(t)}}{n-\sum_{k=1}^n\frac{n\mu_{k 1}{Q_k(t)}}{{\mu_{k1}Q_k(t)}+nR_k(t)}}+\frac{n{P_i(t)}}{{\mu_{i1}Q_i(t)}+nR_i(t)}, \forall1\leqslant i\leqslant n, t\in[0,T].
	\end{equation*}
	$ \forall1\leqslant i\leqslant n, t\in[0,T] $, denote $$ a_i^{(0)}(t)=\check{a}_i(t),\quad a_i^{(l+1)}(t)=\left(\frac{Q_i(t)}{R_i(t)}\frac{1}{n}\sum_{k\neq i}\mu_{k 1}a_k^{(l)}(t)+\frac{P_i(t)}{R_i(t)}\right)\wedge1,\quad\forall l\geqslant1,$$
	then we see the sequence $ \{a_i^{(l)}(t)\}_{l=1}^{\infty} $ is non-negative and monotonically non-increasing. Therefore, according to the monotone convergence theorem, $ \{a_i^{(l)}(t)\}_{l=1}^{\infty} $ converges. Let $\displaystyle\mathring{a}_i(t)=\lim_{l\to\infty}a_i^{(l)}(t)  $, then we see 
	\begin{equation*}
		\mathring{a}_i(t)=\left(\frac{Q_i(t)}{R_i(t)}\frac{1}{n}\sum_{k\neq i}\mu_{k 1}\mathring{a}_k(t)+\frac{P_i(t)}{R_i(t)}\right)\wedge1,\quad \forall1\leqslant i\leqslant n,~~ t\in[0,T],
	\end{equation*}
	thus $a_i^*(t)= \mathring{a}_i(t), \forall1\leqslant i\leqslant n, t\in[0,T] $.
	
	In practice, when there exists $ j $ such that $ a_i^{(j+1)}(t) = a_i^{(j)}(t)$ or $ |a_i^{(j+1)}(t) -a_i^{(j)}(t)| $ is sufficiently small $\forall 1\leqslant i\leqslant n $, then we can use $ a_i^{(j+1)}(t) $ as a   numerical approximation for $ a^*_i(t) $.
\end{remark}

\section{\bf  Robust mean-field game}

This section investigates the asymptotic behavior of the robust $ n $-insurer game as $ n $ approaches infinity. In the context of this game, we can use a type vector $ U=(x^0,\lambda,\mu_{1},\mu_{ 2},\eta,\theta,\delta,\Psi) $ from the value space $ \mathbb{U} $ to represent the AAI's preferences or information. Furthermore, the type vectors of all AAIs collectively generate an empirical measure, representing the probability measure on the
$\mathbb{U}$
given by
$$
\mathcal{D}_n(A)=\frac{1}{n} \sum_{i=1}^n 1_A\left(u_i\right), \text { for Borel sets } A \subset\mathbb{U}.
$$ 
Suppose that $\mathcal{D}_n$ has a weak limit $\mathcal{D}$: $ \int_{\mathbb{U}} f\rd \mathcal{D}_n\to\int_{\mathbb{U}} f\rd \mathcal{D}$ for every bounded continuous function $f$. After establishing the robust equilibrium in the robust $n$-insurer game, our subsequent goal is to identify the robust mean-field equilibrium, assuming the distribution of the type vector $U$ follows $\mathcal{D}$. Referring to \eqref{n-opt} and \eqref{a_star}, we anticipate that
\begin{equation}\label{m-pi}
	\lim_{n\to\infty}\pi^*_i(t)=\lim_{n\to\infty}\left(S_i(t)+\theta_i\frac{\sum_{k=1}^nS_k(t)}{n-\sum_{k=1}^n\theta_k}\right)\frac{Z(t)}{aZ(t)+b}=\left[\theta_i\frac{M_1}{1-\bar{\theta}}+S_i(t)\right]\frac{{Z(t)}}{a{Z(t)}+b},
\end{equation}
where \begin{equation*}
	M_1=\lim_{n\to\infty}\frac{1}{n}\sum_{k=1}^nS_k(t)=\mathbb{E}[\frac{m -\nu\rho( \Psi v_{U,3}+\delta \varUpsilon_{U,3})}{\Psi+{\delta}}]\frac{1}{e^{r(T-t)}},~~\bar{\theta}=\lim_{n\to\infty}\frac{1}{n}\sum_{k=1}^n\theta_k=\mathbb{E}[\theta].
\end{equation*}
$ v_{U,3}(\cdot) $ and $ \varUpsilon_{U,3}(\cdot) $ are similar to that in (\ref{F_i}) and (\ref{mat-ricc}) for AAI with type vector $ U $,
and \begin{equation}\label{m-a}
	\lim_{n\to\infty}a_i^*(t)=\left(\frac{Q_i(t)}{R_i(t)}M_2+\frac{P_i(t)}{R_i(t)}\right)\wedge1,
\end{equation}
where \begin{equation*}
	M_2=\lim_{n\to\infty}\frac{1}{n}\sum_{k=1}^n\mu_{k 1}a_k^*(t)=\mathbb{E}\left[\mu_{ 1}a^*_U(t)\right]
\end{equation*}
is obtained by solving\begin{equation*}
	M_2=\!	\lim_{n\to\infty}\frac{1}{n}\sum_{i=1}^n\mu_{i 1}a_i^*(t)\!=\!	\lim_{n\to\infty}\frac{1}{n}\sum_{i=1}^n\mu_{i 1}\!\left(\frac{Q_i(t)}{R_i(t)}M_2+\frac{P_i(t)}{R_i(t)}\right)\wedge1=\mathbb{E}\!\left[\mu_{ 1}\!\left(\frac{Q_U(t)}{R_U(t)}M_2+\frac{P_U(t)}{R_U(t)}\right)\!\wedge1\right]\!.
\end{equation*}
Expressing the investment and reinsurance strategies for a robust equilibrium with an infinite number of AAIs, we can use the formulations \eqref{m-pi}--\eqref{m-a}. In the subsequent discussion, we will introduce the mean-field game and illustrate that \eqref{m-pi}--\eqref{m-a} serve as solutions to the mean-field game.

Assuming that the distribution $\mathcal{D}$ is a general distribution, let $\left(\Omega, \mathcal{F}, \mathbb{P}\right)$ support a random vector $ U=(x^0,\lambda,\mu_{1},\mu_{ 2},\eta,\theta,\delta,\Psi):\Omega\to\mathbb{U} $ with distribution $ \mathcal{D} $. We also assume that all the relevant Brownian motions are adapted to the smallest filtration $\hat{\mathbb{F}}=\{\hat{\mathcal{F}}_t,t\in[0,T]\}$ with the usual assumptions, ensuring $U$ is $\hat{\mathcal{F}}_0$-measurable. The natural filtration generated by the Brownian motions $ \tilde{W},{W}$, and ${B}$ is denoted by $\mathbb{F}^{\tilde{W},{W},{B}}=\{\mathcal{F}_t^{\tilde{W},{W},{B}},t\in[0,T]\}$.

Under the measure of  the AAI with type vector $ u_0=(x_0^0,\lambda_0,\mu_{ 01},\mu_{ 02},\eta_0,\theta_0,\delta_0,\Psi_{0}) $,
 the surplus process of a representative AAI with type $ U\sim\mathcal{D} $ is described as:
\begin{equation}\label{equ:xu}
\hspace{-5pt}	\left\{\!\begin{aligned}
		\mathrm{d} X_{U}(t)
		=&\left[r X_{U}(t)+\eta\left(\lambda+\hat{\lambda}\right) \mu_{ 1}-\hat{\eta}(1-a_{U}(t))^2\left(\lambda+\hat{\lambda}\right) \mu_{ 2}+\pi_{U}(t)m({a}{Z(t)}+{{b}})\right] \mathrm{d} t\\&+\!a_{U}(t)\!\left(\!\sqrt{\hat{\lambda}}\mu_{1}(\rd  \tilde{W}^{\mathbb{Q}^{u_0}}(t)\!+\!\phi_{u_0}(t)\rd t)\!+\!\sqrt{(\hat{\lambda}\!+\!\lambda)\mu_{ 2}\!-\!\hat{\lambda}\mu_{ 1}^2}(\rd \hat{W}_U^{\mathbb{Q}^{u_0}}(t)\!+\!\vartheta_{u_0}^U(t)\rd t)\!\right)\\&+\pi_{U}(t){\Sigma}(t)(\rd {W}^{\mathbb{Q}^{u_0}}(t)+\varphi_{u_0}(t)\rd t),
		\\	\rd Z(t)\!=\!{\kappa}&(\bar{Z}\!-\!Z(t))\rd t\!+\!{\nu}\sqrt{Z(t)}\left[{\rho}(\rd {W}^{\mathbb{Q}^{u_0}}(t)\!+\!\varphi_{u_0}(t)\rd t)\!+\!\sqrt{1\!-\!{\rho}^2}(\rd{B}^{\mathbb{Q}^{u_0}}(t)\!+\!\chi_{u_0}(t)\rd t)\right]\!,
	\end{aligned}\right.
\end{equation}
where \begin{equation*}
	{\Sigma}(t)={a}\sqrt{Z(t)}+\frac{{b}}{\sqrt{Z(t)}}.
\end{equation*}

Similar to the robust $ n $-insurer game, we initially establish the admissible set of strategies and corresponding density generators. Certain notations remain consistent with the previous section. Define:
\begin{equation}
	\begin{aligned}
		\mathscr{U}\!\!=&\Big\{(\pi_{U},a_{U}):U\sim\mathcal{D}, (\pi_{U},a_{U}) \text{ is progressively measurable w.r.t. }\hat{\mathbb{F}}\text{ and }0\leqslant a_{U}(t)\leqslant1,\\
		&\exists \mathcal{C}_{U}\in\mathbb{R}_+, \sup_{U}\mathcal{C}_{U}<\infty, \text{ such that } \pi_{U}(t)=\ell_{U}(t)\frac{Z(t)}{aZ(t)+b},|\ell_{U}(t)|\leqslant \mathcal{C}_{U}, \forall t\in[0,T]\Big\},\\
		\mathscr{U}_{u_0}&\!=\!\Big\{(\pi_{u_0},a_{u_0}):(\pi_{u_0},a_{u_0}) \text{ is progressively measurable w.r.t. }{\mathbb{F}}\text{ and }0\leqslant a_{u_0}(t)\leqslant1,\\
		&\quad\exists \mathcal{C}_{u_0}\in\mathbb{R}_+, \text{ such that } \pi_{u_0}(t)=\ell_{u_0}(t)\frac{Z(t)}{aZ(t)+b},|\ell_{u_0}(t)|\leqslant \mathcal{C}_{u_0}, \forall t\in[0,T]\Big\},\\
			\mathscr{A}\!\!=&\Big\{\left(\varphi, \chi,\phi,\vartheta\right):\left(\varphi, \chi,\phi,\vartheta\right) \text{ is progressively measurable w.r.t. }\mathbb{F}\text{ and }\left(\varphi(t), \chi(t)\right)\\&\ =\!\left(\hslash(t)\!\sqrt{Z(t)},\hbar(t)\!\sqrt{Z(t)} \right),
		|\hslash(t)|, |\hbar(t)|\leqslant \mathcal{C},\sup|\phi|<\infty,\sup|\vartheta|<\infty,~ \forall t\in[0,T]\Big\}.
	\end{aligned}
\end{equation}

The AAIs aim to identify the robust mean-field equilibrium strategies within the admissible set. Analogous to Definition \ref{def1}, the definition of the robust mean-field equilibrium in the mean-field game is presented below:

When the AAIs have type distribution $ U \sim \mathcal{D} $ and take strategy $ (\pi_{U},a_U) $, for AAI with type vector $ u_0 $, we define the relative performance at time $t$  as $ Y^{(\pi_{U},a_U)}_{u_0}(t)=X_{u_0}(t) {-\theta_{0}}\bar{X}_{u_0}(t),$  $ \forall  u_0\in\mathbb{U} $, where $ \bar{X}_{u_0}(t)=\mathbb{E}^{\mathbb{Q}^{u_0}}\left[X_{U}(T)\Big|\mathcal{F}_T^{\tilde{W},{W},{B}}\right] $,  
and  define the preference under the equivalent probability measure  $\mathbb{Q}^{u_0}$ given states $ Y_{u_0}^{(\pi_{U},a_U)}(t)=y$ and $Z(t)=z $  as follows:
\begin{equation*}
	\begin{aligned}
		&J_{u_0}\left((\pi_{u_0},a_{u_0}),\left(\varphi_{u_0}, \chi_{u_0},\phi_{u_0},\vartheta_{u_0}\right),(\pi_{U},a_U),t,y,z\right)\\
		=&\mathbb{E}^{\mathbb{Q}^{u_0}}_{t,y,z}\left[Y^{(\pi_{U},a_U)}_{u_0}(T)\right]
		-{\delta_0\over 2}\var^{\mathbb{Q}^{u_0}}_{t,y,z}\left[Y^{(\pi_{U},a_U)}_{u_0}(T)\right]
		\\
		&+\mathbb{E}^{\mathbb{Q}^{u_0}}_{t,y,z}\left[\frac{1}{2 \Psi_{0}}\int_{t}^{T} \left( {\varphi_{u_0}^{2}(s)}+{\chi_{u_0}^2(s)}+{\phi_{u_0}^{2}(s)}+{\mathbb{E}^{\mathbb{Q}^{u_0}}\left[(\vartheta_{u_0}^U(s))^2\Big|\mathcal{F}_T^{\tilde{W},{W},{B}}\right]}\right)\mathrm{d} s \right],
	\end{aligned}
\end{equation*}
where $ \mathbb{E}^{\mathbb{Q}^{u_0}}_{t,y,z}[\cdot] $ and $\var^{\mathbb{Q}^{u_0}}_{t,y,z}[\cdot]$ represent the conditional expectation and the conditional variance under probability measure $ \mathbb{Q}^{u_0}$ given states $ Y_{u_0}^{(\pi_{U},a_U)}(t)=y$ and $Z(t)=z $, respectively.  Define the objective function of an AAI with type vector $u_0$ under the worst-case scenario in a similar manner as:
$$
\begin{aligned}
	&\bar{J}_{u_0}\left((\pi_{u_0},a_{u_0}),(\pi_{U},a_U),t,y,z\right)\\=&\inf_{\left(\varphi_{u_0}, \chi_{u_0},\phi_{u_0},\vartheta_{u_0}\right)\in \mathscr{A}}J_{u_0}\left((\pi_{u_0},a_{u_0}),\left(\varphi_{u_0}, \chi_{u_0},\phi_{u_0},\vartheta_{u_0}\right),(\pi_{U},a_U),t,y,z\right).
\end{aligned}
$$
If \begin{equation*}
	\left(\varphi_{u_0}^\dag,\chi_{u_0}^\dag,\phi_{u_0}^\dag,\vartheta_{u_0}^\dag\right)=\arg\!\inf_{\left(\varphi_{u_0}, \chi_{u_0},\phi_{u_0},\vartheta_{u_0}\right)\in \mathscr{A}}\!J_{u_0}\left((\pi_{u_0},a_{u_0}),\!\left(\varphi_{u_0}, \chi_{u_0},\phi_{u_0},\!\vartheta_{u_0}\right)\!,(\pi_{U},a_U),t,y,z\right)\!,
\end{equation*}then we call $ \left(\varphi_{u_0}^\dag,\chi_{u_0}^\dag,\phi_{u_0}^\dag,\vartheta_{u_0}^\dag\right)$ the worst-case scenario density generator associated with the strategy.

The objective function parallels the one introduced in the preceding section, featuring the variance operator and exhibiting time inconsistency. To address this, we adopt the time-consistent equilibrium strategy outlined in \cite{bjork2017time} to tackle the time-inconsistent optimization problem. Analogous to Definitions \ref{def}--\ref{def1}, defining the robust mean-field equilibrium of the mean-field game requires establishing the best response strategy.
 \begin{definition}\label{def:mfrs}
For any fixed $(t,y,z)\in [0,T]\times \mathbb{R}\times \mathbb{R}_+$, consider an admissible  strategy $ (\pi_{u_0}^\dag,a_{u_0}^\dag) $ for AAI with type vector $ u_0 $ given that the AAIs have type distribution $ U \sim \mathcal{D} $ and take strategy $ (\pi_{U},a_U) $. For any $(\pi,a)\in \mathscr{U}_{u_0}$ and $h>0$,  define a new $ h $-perturbed strategy $(\pi^h,a^h)$ by
 	\begin{equation*}
 		(\pi^h(s),a^h(s))=\left\{
 		\begin{aligned}
 			&(\pi(s),a(s)),&&t\leqslant s<t+h,
 			\\&(\pi^\dag_{u_0}(s),a_{u_0}^\dag(s)),&& t+h\leqslant s\leqslant T.
 		\end{aligned}
 		\right.
 	\end{equation*}
 	If for all $(\pi,a)\in \mathscr{U}_{u_0}$,   we have 
 	\begin{eqnarray}{\label{mv1}}
 		\mathop {\lim \inf }\limits_{h\rightarrow 0^+} \frac{\bar{J}_{u_0}\!\left((\pi_{u_0}^\dag,a_{u_0}^\dag),(\pi_{U},a_U),t,y,z\right)\!-\bar{J}_{u_0}\!\left((\pi^h,a^h),(\pi_{U},a_U),t,y,z\right)\!}{h}\geqslant 0,
 	\end{eqnarray}
 	then  $ (\pi_{u_0}^\dag,a_{u_0}^\dag)$  is called the \textbf{robust time-consistent Nash equilibrium investment-reinsurance {response} strategy} or simply the \textbf{best response strategy} and the robust time-consistent response value function of AAI with type vector $ u_0 $ is given by $\bar{J}_{u_0}\!\left((\pi_{u_0}^\dag,a_{u_0}^\dag),(\pi_{U},a_U),t,y,z\right)\!$.
 	
 \end{definition}
\begin{definition}\label{def:mfe}
	We call $ (\pi_{U}^*,a_{U}^* ), U\sim \mathcal{D}$ a {\textbf{robust (time-consistent) mean-field equilibrium}} if $(\pi_{U}^*,a_{U}^*)\in\mathscr{U}$, and   for all  $ u_0\in\mathbb{U} $,
$ (\pi_{u_0}^*,a_{u_0}^*)$  is  the {best response strategy} given that the AAIs have type distribution $ U \sim \mathcal{D} $ and take strategy $ (\pi^*_{U},a^*_U) $.  We then call $(\pi_{u_0}^*,a_{u_0}^*) $ the \textbf{robust (time-consistent) mean-field equilibrium strategy}  for AAI with type vector $ u_0 $ and call $V^{u_0}(t,y,z):= \bar{J}_{u_0}\!\left((\pi_{u_0}^*,a_{u_0}^*),(\pi^*_{U},a^*_U),t,y,z\right)\! $ the robust time-consistent value  function.

\end{definition}

Definition \ref{def:mfe} establishes that the robust mean-field equilibrium is attained by identifying a fixed point. With AAIs characterized by a type distribution $ U \sim \mathcal{D} $ and employing strategies $ (\pi^*_{U},a^*_U) $, an AAI with type vector $ u_0 $ aims to find the best response investment-reinsurance strategy $(\pi^*_{u_0},a_{u_0}^*)$ under the worst-case scenario. This iterative process is carried out by each insurer in the system.
If we obtain a fixed point in the meaning that $(\pi^*_{u_0},a_{u_0}^*)=(\pi^*_{U},a^*_U)|_{U=u_0} ,\forall u_0\in\mathbb{U}$, then $ (\pi_{U}^*,a_{U}^* ), U \sim \mathcal{D}$ is referred to as the robust mean-field equilibrium.

%

Based on \eqref{m-pi} and \eqref{m-a}, we expect 
\begin{equation}\label{m-opt}
	\left\{\begin{aligned}
		\pi_{U}^*(t)&=\left[\theta_i\frac{\mathbb{E}[\frac{m -\nu\rho( \Psi v_{U,3}+\delta \varUpsilon_{U,3})}{(\Psi+{\delta})}]}{e^{r(T-t)}(1-\mathbb{E}[{\theta}])}+S_U(t)\right]\frac{{Z(t)}}{a{Z(t)}+b},\\
		a_U^*(t)&=\left(\frac{Q_U(t)}{R_U(t)}\bar{\Omega}(t)+\frac{P_U(t)}{R_U(t)}\right)\wedge1,
	\end{aligned}\right.
\end{equation}
where
 $\overline{\Omega}(t)=\mathbb{E}[\mu_{ 1}a_{U}^*(t)] $	is obtained by $\overline{\Omega}(t)=\mathbb{E}[\mu_{ 1}\left(\frac{Q_U(t)}{R_U(t)}\bar{\Omega}(t)+\frac{P_U(t)}{R_U(t)}\right)\wedge1] $ and \begin{equation}\label{SRQP}
 	\left\{\begin{aligned}
 		&S_U(t)=\frac{m -\nu\rho( v_{U,3}\Psi+\delta \varUpsilon_{U,3})}{(\Psi+{\delta})e^{r(T-t)}},\\
 		&R_U(t)=(\lambda+\hat{\lambda})\mu_{2}((\delta+\Psi)e^{2r(T-t)}+2\hat{\eta}e^{r(T-t)}),\\
 		&Q_U(t)=\hat{\lambda}\theta\mu_{ 1}(\Psi+\delta)e^{2r(T-t)},\\& P_U(t)=2\hat{\eta}(\lambda+\hat{\lambda})\mu_{ 2}e^{r(T-t)}.
 	\end{aligned}\right.\end{equation}
The expression \eqref{m-opt} represents the robust mean-field equilibrium, achieved by taking the limit as $n$ approaches infinity in the competition system with $n$ AAIs. Subsequently, we aim to concurrently determine the best response strategy for all $u_0\in\mathbb{U} $ to establish the robust mean-field equilibrium $ (\pi_{U}^*,a_{U}^* ), U\sim \mathcal{D}$.

Based on the dynamic \eqref{equ:xu} of $ X_{U}$, under the measure of $ u_0 $, $Y^{(\pi_{U},a_U)}_{u_0}$ evolves as
\begin{equation}
	\begin{aligned}
	&	\mathrm{d} Y^{(\pi_{U},a_U)}_{u_0}(t)\\
		=&r Y^{(\pi_{U},a_U)}_{u_0}(t)\mathrm{d} t+\left(\eta_0\left(\lambda_0+\hat{\lambda}\right) \mu_{0 1}-\theta_0\mathbb{E}\left[\eta\left(\lambda+\hat{\lambda}\right) \mu_{ 1}\right]\right)\rd t\\
		&-\hat{\eta}(1-a_{u_0}(t))^2\left(\lambda_0+\hat{\lambda}\right) \mu_{ 02}\rd t+\theta_0\hat{\eta}\mathbb{E}^{\mathbb{Q}^{u_0}}\left[(1-a_{U}(t))^2\left(\lambda+\hat{\lambda}\right) \mu_{ 2}\Big|\mathcal{F}_T^{\tilde{W},{W},{B}}\right]\rd t\\
		&+\left(\pi_{u_0}(t)-\theta_0\mathbb{E}^{\mathbb{Q}^{u_0}}\left[\pi_{U}(t)\Big|\mathcal{F}_T^{\tilde{W},{W},{B}}\right]\right)\left({\Sigma}(t)\rd {W}^{\mathbb{Q}^{u_0}}(t)+\left(m\sqrt{Z(t)}+{\varphi_{u_0}(t)}\right)\Sigma(t) \mathrm{d} t\right)\\
		&+\sqrt{\hat{\lambda}}\left(a_{u_0}(t)\mu_{01}-\theta_0\mathbb{E}^{\mathbb{Q}^{u_0}}\left[a_{U}(t)\mu_{1}\Big|\mathcal{F}_T^{\tilde{W},{W},{B}}\right]\right)(\rd  \tilde{W}^{\mathbb{Q}^{u_0}}(t)+\phi_{u_0}(t)\rd t)\\
		&+a_{u_0}(t)\sqrt{(\hat{\lambda}+\lambda_0)\mu_{0 2}-\hat{\lambda}\mu_{0 1}^2}(\rd \hat{W}_{u_0}^{\mathbb{Q}^{u_0}}(t)+\vartheta_{u_0}^{u_0}(t)\rd t)
		\\&-\theta_0\mathbb{E}^{\mathbb{Q}^{u_0}}\left[a_{U}(t)\sqrt{(\hat{\lambda}+\lambda)\mu_{ 2}-\hat{\lambda}\mu_{ 1}^2}\vartheta_{u_0}^U(t)\Big|\mathcal{F}_T^{\tilde{W},{W},{B}}\right]\rd t.
	\end{aligned}
\end{equation}
To obtain the  best response strategy, we first define the  infinitesimal generator $\mathcal{A}^{(\pi_{U},a_U)}_{u_0}$ based on the dynamic of $ Y^{(\pi_{U},a_U)}_{u_0} $ and $Z$ for  $ u_0=(x_0^0,\lambda_0,\mu_{ 01},\mu_{ 02},\eta_0,\theta_0,\delta_0,\Psi_{0}) $. For $ U\sim \mathcal{D} $,  $\forall (t,y,z) \in[0,T]\times\mathbb{R}\times\mathbb{R}_+$ and $f(t,y,z)\in C^{1,2,2}([0,T]\times\mathbb{R}\times\mathbb{R}_+)$, 
\begin{equation*}
	\begin{aligned}
		&\mathcal{A}^{(\pi_{U},a_U)}_{u_0}f(t,y,z)\\
		=&f_t+\left[{\kappa}(\bar{Z}-z)+{\nu}\sqrt{z}({\rho}\varphi_{u_0}+\sqrt{1-{\rho}^2}\chi_{u_0})\right]f_{z}+\frac{1}{2}\nu^2zf_{zz}\\
		&+\rho\nu (az+b)\left(\pi_{u_0}(t)-\theta_0\mathbb{E}^{\mathbb{Q}^{u_0}}\left[\pi_{U}(t)\Big|\mathcal{F}_T^{\tilde{W},{W},{B}}\right]\right)f_{yz}\\
		&+r yf_y+\left(\eta_0\left(\lambda_0+\hat{\lambda}\right) \mu_{0 1}-\theta_0\mathbb{E}\left[\eta\left(\lambda+\hat{\lambda}\right) \mu_{ 1}\right]\right)f_y\\
		&-\hat{\eta}(1-a_{u_0}(t))^2\left(\lambda_0+\hat{\lambda}\right) \mu_{0 2}f_y+\theta_0\mathbb{E}^{\mathbb{Q}^{u_0}}\left[\hat{\eta}(1-a_{U}(t))^2\left(\lambda+\hat{\lambda}\right) \mu_{ 2}\Big|\mathcal{F}_T^{\tilde{W},{W},{B}}\right]f_y\\
		&+\left(\pi_{u_0}(t)-\theta_0\mathbb{E}^{\mathbb{Q}^{u_0}}\left[\pi_{U}(t)\Big|\mathcal{F}_T^{\tilde{W},{W},{B}}\right]\right)\left(m\sqrt{z}+{\varphi_{u_0}(t)}\right)\sigma f_y\\
		&+\sqrt{\hat{\lambda}}\left(a_{u_0}(t)\mu_{01}-\theta_0\mathbb{E}^{\mathbb{Q}^{u_0}}\left[a_{U}(t)\mu_{1}\Big|\mathcal{F}_T^{\tilde{W},{W},{B}}\right]\right)\phi_{u_0}(t)f_y\\
		&+a_{u_0}(t)\sqrt{(\hat{\lambda}+\lambda_0)\mu_{0 2}-\hat{\lambda}\mu_{0 1}^2}\vartheta_{u_0}^{u_0}(t)f_y
		-\theta_0\mathbb{E}^{\mathbb{Q}^{u_0}}\left[a_{U}(t)\sqrt{(\hat{\lambda}+\lambda)\mu_{ 2}-\hat{\lambda}\mu_{ 1}^2}\vartheta_{u_0}^{U}(t)\Big|\mathcal{F}_T^{\tilde{W},{W},{B}}\right]f_y\\
		&+\frac{1}{2}{\hat{\lambda}}\left(a_{u_0}(t)\mu_{01}-\theta_0\mathbb{E}^{\mathbb{Q}^{u_0}}\left[a_{U}(t)\mu_{1}\Big|\mathcal{F}_T^{\tilde{W},{W},{B}}\right]\right)^2f_{yy}+\frac{1}{2}\left(\pi_{u_0}(t)-\theta_0\mathbb{E}^{\mathbb{Q}^{u_0}}\left[\pi_{U}(t)\Big|\mathcal{F}_T^{\tilde{W},{W},{B}}\right]\right)^2\!{\sigma}^2f_{yy}\\
		&+\frac{1}{2}a_{u_0}^2(t){\left((\hat{\lambda}+\lambda_0)\mu_{0 2}-\hat{\lambda}\mu_{ 01}^2\right)}f_{yy},
	\end{aligned}
\end{equation*}
where $ \sigma=a\sqrt{z}+\frac{b}{\sqrt{z}} $.
We also define \begin{equation*}
\hspace{-5pt}	\begin{aligned}
		&	\mathcal{L}^{(\pi_{U},a_U)}_{u_0}\left((\pi_{u_0},a_{u_0}),\left(\varphi_{u_0}, \chi_{u_0},\phi_{u_0},\vartheta_{u_0}\right),f,g, h,(t,y,z)\right)\\=~&\mathcal{A}^{(\pi_{U},a_U)}_{u_0}f(t,y,z)+  \frac{\varphi_{u_0}^{2}(t)}{2 h} +\frac{\chi_{u_0}^2(t)}{2 h} +  \frac{\phi_{u_0}^{2}(t)}{2 h} +\frac{1}{2 h}\mathbb{E}^{\mathbb{Q}^{u_0}}\left[(\vartheta_{u_0}^U(t))^2\Big|\mathcal{F}_T^{\tilde{W},{W},{B}}\right] \\
		&+\delta_{0} g(t,y,z) \mathcal{A}^{(\pi_{U},a_U)}_{u_0}g(t,y,z)-\frac{\delta_{0}}{2}\mathcal{A}^{(\pi_{U},a_U)}_{u_0}g^2(t,y,z).
	\end{aligned}
\end{equation*}

 The extended HJBI equations  are presented as follows.
\begin{proposition}[extended HJBI equations]\label{def:hjbi2}
	The extended HJBI equations of AAI with type $ u_0=(x_0^0,\lambda_0,\mu_{ 01},$ $\mu_{ 02},\eta_0,\theta_0,\delta_0,\Psi_{0}) $ for the two-tuple $ (V, \Upsilon) $ are defined as
	\begin{equation}\label{HJBI-MV2}
		\begin{aligned}
			\sup _{(\pi_{u_0},a_{u_0})\in\mathscr{U}_{u_0}} \inf _{\left(\varphi_{u_0}, \chi_{u_0},\phi_{u_0},\vartheta_{u_0}\right)\in\mathscr{A}}&\mathcal{L}^{(\pi^*_{U},a^*_U)}_{u_0}\left((\pi_{u_0},a_{u_0}),(\varphi_{u_0}, \chi_{u_0},\phi_{u_0},\vartheta_{u_0}),V,\Upsilon,\Psi_0,(t,y,z)\right)  =0,
		\end{aligned}
	\end{equation}
	with boundary condition $ V(T,y,z)=y $
	and \begin{equation}\label{g-MV2}
		\mathcal{A}^{u_0}\Upsilon(t,y,z)=0
	\end{equation}with boundary condition $ \Upsilon(T,y,z)=y $, where\begin{equation}\label{HJBI-opt-MV2}
		\begin{aligned}
			&(\pi^\circ_{u_0},a^\circ_{u_0})\!=\!\arg\!	\sup _{(\pi_{u_0},a_{u_0})\in\mathscr{U}_{u_0}} \inf _{(\varphi_{u_0}, \chi_{u_0},\phi_{u_0},\vartheta_{u_0})\in\mathscr{A}}\!\!\mathcal{L}^{(\pi^*_{U},a^*_U)}_{u_0}\!\left((\pi_{u_0},a_{u_0}),(\varphi_{u_0}, \chi_{u_0},\phi_{u_0},\vartheta_{u_0}),V,\Upsilon,\Psi_0,(t,y,z)\right),\\
			&( \varphi^\circ_{u_0},\chi^\circ_{u_0},\phi_{u_0}^\circ,\vartheta_{u_0}^\circ)\!=\!\arg	\! \inf _{(\varphi_{u_0}, \chi_{u_0},\phi_{u_0},\vartheta_{u_0})\in\mathscr{A}}\!\!\mathcal{L}^{(\pi^*_{U},a^*_U)}_{u_0}\!\left((\pi_{u_0},a_{u_0}),(\varphi_{u_0}, \chi_{u_0},\phi_{u_0},\vartheta_{u_0}),V,\Upsilon,\Psi_0,(t,y,z)\right).
		\end{aligned}
	\end{equation}
\end{proposition}
Similarly, we define $ (\pi^\circ_{u_0},a^\circ_{u_0}) $ as the candidate response strategies for the AAI with type vector $ {u_0} $. We also designate $ ( \varphi^\circ_{u_0},\chi^\circ_{u_0},\phi_{u_0}^\circ,\vartheta_{u_0}^\circ) $ as the density generator associated with the worst-case scenario for the candidate strategies. Assuming that the pair $ (v^{u_0}, \varUpsilon^{u_0}) $ solves (\ref{HJBI-MV2})--(\ref{g-MV2}), we refer to it as the candidate value function tuple, with $ v^{u_0}$ being the candidate value function. Analogous to Theorems \ref{solution-HJBI}--\ref{Verification}, we also have the following results. The proofs follow a similar approach to the previous section and are omitted.
		For AAI with type vector $ u=(x^0,\lambda,\mu_{ 1},\mu_{ 2},\eta,\theta,\delta,\Psi)$,
		let \begin{equation*}
			\left\{\begin{aligned}
				&K_{u,1}=\diag(-\frac{1}{2}\nu^2(\rho^2\frac{\Psi\delta}{\Psi+\delta}+(1-\rho^2)\Psi), \nu^2\rho^2\frac{\delta \Psi}{(\Psi+{\delta})^2}),\\
				&K_{u,2}=\diag(-\frac{1}{2}\nu^2\delta(1-\rho^2\frac{\delta}{\Psi+\delta} ),\nu^2\rho^2\Psi\frac{\delta \Psi}{(\Psi+{\delta})^2}), \\
				&K_{u,1,2}=\diag(-\nu^2\left[ \rho^2\Psi\frac{2\Psi{\delta}}{(\Psi+{\delta})^2}+(1-\rho^2)\Psi\right],\nu^2\rho^2\frac{\Psi\delta}{\Psi+\delta}),\\
				&\bar{B}_{u,1}=\diag(-(\kappa+m\nu\rho\frac{\Psi}{\Psi+\delta}), -\left[\kappa +m\nu\rho\frac{{\delta}^2+(\Psi)^2}{(\Psi+{\delta})^2}\right]), \\
				&\bar{B}_{u,2}=\diag (-m\nu\rho\frac{\delta}{\Psi+\delta},-m\nu\rho\frac{2\Psi{\delta}}{(\Psi+{\delta})^2}),\\
				&G_u=\diag(\frac{m^2}{2(\Psi+\delta)}, \frac{\delta m^2}{(\Psi+\delta)^2}),
			\end{aligned}\right.
		\end{equation*}
		
		and $ H=\begin{pmatrix}
			0&1\\1&0
		\end{pmatrix} $, where $ \diag(a,b)= \begin{pmatrix}
			a&0\\0&b
		\end{pmatrix}$. Define 
		\begin{equation*}
			\alpha_{u,1}=\|K_{u,1}\|_{\sup}+\|K_{u,2}\|_{\sup}+\|K_{u,1,2}\|_{\sup},\quad\alpha_{u,2}=\|\bar{B}_{u,1}\|_{\sup}+\|\bar{B}_{u,2}\|_{\sup},\quad	\alpha_{u,3}=\|G_{u}\|_{\sup}
		\end{equation*}
	and
		\begin{equation*}
			\Delta_u = \alpha_{u,2}^2-4\alpha_{u,1}\alpha_{u,3} , \varsigma_{u,1}=\frac{-\alpha_{u,2}+\sqrt{\Delta_u}}{2\alpha_{u,1}} , \varsigma_{u,2}=\frac{-\alpha_{u,2}-\sqrt{\Delta_u}}{2\alpha_{u,1}}, 
		\end{equation*}
		where $ \|\diag(a,b)\|_{\sup}= \max\{|a|,|b|\}$. Then we have the following results.
		\begin{proposition}\label{ricc2}
			For every $u\in\mathbb{U}$, if any of the following three compatible conditions holds, then we can define the corresponding $ U_u(t) $. Subsequently, we observe that the matrix Riccati equation (\ref{mat-ricc2}) has a solution $ F_u$ with $ \|F_u(t)\|_{\sup}\leqslant U_u(t) $ for $ t\in[0,T] $.\begin{equation}\label{mat-ricc2}
			\left\{\begin{aligned}
					&\frac{\rd F_u}{\rd t}=F_uK_{u,1}F_{u}+HF_uHK_{u,2}F_uH+\bar{B}_{u,1}F_u+HF_uH\bar{B}_{u,2}+HF_uK_{u,1,2}HF_u+ G_u,\\&F_u(0)=\diag(0,0).
				\end{aligned}\right.
			\end{equation}
			\begin{itemize}
				\item [(i)]If $ \Delta_u =0 $ and $ T<\frac{2}{\alpha_{u,2}} $, then define 
				\begin{equation*}
					U_u(t)=\varsigma_{u,1}+\frac{\alpha_{u,2}}{\alpha_{u,1}(2-\alpha_{u,2}t)}.
				\end{equation*}
				\item [(ii)]If $ \Delta_u >0 $ and $ T<\frac{1}{\sqrt{\Delta_u}}\log(\frac{\varsigma_{u,2}}{\varsigma_{u,1}}) $, then define
				\begin{equation*}
					U_u(t)=\frac{\alpha_{u,3}}{\alpha_{u,1}}\frac{1-e^{\alpha_{u,1}(\varsigma_{u,1}-\varsigma_{u,2})t}}{\varsigma_{u,2}-\varsigma_{u,1}e^{\alpha_{u,1}(\varsigma_{u,1}-\varsigma_{u,2})t}}.
				\end{equation*}
				\item [(iii)]If $ \Delta_u <0 $ and $ T<\frac{1}{\sqrt{|\Delta_u|}}(\pi+2\arctan(\frac{\Re{\varsigma_{u,1}}}{\Im{\varsigma_{u,1}}})) $, where $ \Re{\varsigma_{u,1}} $  represents the real part of $ \varsigma_{u,1} $ and $ \Im{\varsigma_{u,1}} $ represents the imaginary part of $ \varsigma_{u,1} $, then define
				\begin{equation*}
					U_u(t)=\Re{\varsigma_{u,1}}+\Im{\varsigma_{u,1}}\tan\left(\alpha_{u,1}t\Im{\varsigma_{u,1}}-\arctan\left(\frac{\Re{\varsigma_{u,1}}}{\Im{\varsigma_{u,1}}}\right)\right).
				\end{equation*}
			\end{itemize}
		\end{proposition}
		\begin{theorem}\label{solution-MHJBI}
	
Assume that the parameters fulfill one of the three compatible conditions outlined in Proposition \ref{ricc2}, and
		\begin{equation}\label{condition}
			\sup_{u_0\in\mathbb{U}}m+\nu\rho U_{u_0}(T)\Psi_0<\frac{\kappa}{\sqrt{2}\nu},	\quad\sup_{u_0\in\mathbb{U}}\nu\sqrt{1-\rho^2} U_{u_0}(T)\Psi_0<\frac{\kappa}{\sqrt{2}\nu}.
		\end{equation}
	Then the candidate value function tuple $ (v^{u},\varUpsilon^{u}) $ of  AAI with type vector $u$  is  given by\begin{equation*}
			v^{u}(t, y,z)= v_{u,1}(t)+ yv_{u,2}(t)+v_{u,3}(t)z,\quad \varUpsilon^{u}(t, y,z)= \varUpsilon_{u,1}(t)+ y\varUpsilon_{u,2}(t)+\varUpsilon_{u,3}(t)z,
		\end{equation*}where 
		\begin{equation*}
			\begin{aligned}
				v_{u,2}(t)=\varUpsilon_{u,2}(t)=&e^{r(T-t)},
			\end{aligned}
		\end{equation*}
		\begin{equation}\label{F_u}
			\diag(v_{u,3}(t),\varUpsilon_{u,3}(t) )=F_u(T-t).
		\end{equation}
		$ F_u(t) $ is the solution to (\ref{mat-ricc2}), $ \varUpsilon_{u,1}(t) $ and $ v_{u,1}(t) $ can  be obtained similarly based on \eqref{varUpsilon} and \eqref{v1}. Then the candidate value function tuple follows.
		\end{theorem}
		
Theorem \ref{solution-MHJBI} provides solutions to the extended HJBI equations, and the subsequent theorems confirm that the candidate response strategies, the corresponding worst-case scenario density generators, and the candidate value functions indeed constitute the robust equilibrium strategies, the associated worst-case scenario density generators, and value functions, respectively.
		
		\begin{theorem}[Verification theorem]\label{mVerification}
			$ \forall u_0\in\mathbb{U} $, the candidate value function   $v^{u_0},  $ is the value function tuple, that is, the equation  $ V^{u_0}(t, y,z)=v^{u_0}(t, y,z) $ holds.
		\end{theorem}
		Utilizing Theorems \ref{solution-MHJBI}--\ref{mVerification}, we derive the robust time-consistent mean-field equilibrium investment and reinsurance strategies as follows.
		\begin{theorem}	
			If $1\neq \mathbb{E}[\theta]$, the robust time-consistent  mean-field equilibrium strategy $ (\pi^*_{U},a^*_{U}) $,  $ U\sim\mathcal{D} $ exists and is given by \eqref{m-opt}. 
			Furthermore,  the associated worst-case scenario density generator  is $ \left(\varphi_{U}^*, \chi_{U}^*,\phi_{U}^*,\vartheta_{U}^*\right)$,  $ U\sim\mathcal{D} $, when $ U=u_0 $, the explicit forms  are given by
			\begin{equation*}
				\left\{\begin{aligned}
					&{\varphi^*_{u_0}}(t)=-(\nu \rho v_{u_0,3}+S_{u_0}(t) v_{u_0,2})\Psi_0\sqrt{Z(t)},\\
					&{\chi^*_{u_0}}(t)=-\nu\sqrt{1-\rho^2} v_{u_0,3}\Psi_0\sqrt{Z(t)},\\
					&\phi^*_{u_0}(t)=-\sqrt{\hat{\lambda}}\left[\mu_{01}a_{u_0}^*(t)-{\theta_{0}}\mathbb{E}[\mu_{1}a_{U}^*(t)]\right]v_{u_0,2}\Psi_0,\\
					&\vartheta_{u_0}^{u_0,*}(t)=-a_{u_0}^*(t)\sqrt{(\hat{\lambda}+\lambda_{0})\mu_{0 2}-\hat{\lambda}\mu_{0 1}^2}v_{u_0,2}\Psi_0,\\
					&\vartheta_{u_0}^{u_1,*}(t)=0,u_1\in\mathbb{U},u_1\neq u_0,
				\end{aligned}\right.
			\end{equation*}
			where $ S_{u_0}(t) $, $ R_{u_0}(t) $,  $ Q_{u_0}(t) $, $ P_{u_0}(t) $,  $ v_{u_0,j}(t) $ and  $ \varUpsilon_{u_0,j}(t), 1\leqslant j\leqslant 3 $, are given in \eqref{SRQP} and Theorem  \ref{solution-MHJBI}.
			
			If $1= \mathbb{E}[\theta]$, the robust  mean-field equilibrium does not exist.
		\end{theorem}
	The above paragraph demonstrates the congruence of results between the mean-field and $n$-insurer scenarios. As a result, a closed-form solution for the robust equilibrium investment strategy can be derived, while the robust equilibrium reinsurance strategy can be computed numerically.
	
	\section{\bf Numerical analysis}

In this section, we examine the impact of various parameters on equilibrium strategies, specifically the portfolio processes and reinsurance proportions for the AAIs. To illustrate portfolio behavior, we utilize the deterministic coefficients ${\pi^*(t)}/ \frac{{Z(t)}}{a{Z(t)}+b} $, as detailed in Remark \ref{Randomness}. Similarly, $ a^*(t) $ is used to depict the reinsurance behavior of the AAIs. We set $ T=5 $, $ r=0.02 $. Unless specified otherwise, we employ the real-data estimates provided in \cite{42model}, derived from S\&P 500 and VIX data spanning January 2010 to December 2019. Specifically, $ \kappa=7.3479 $, $ \nu=0.6612 $, $ m=2.9428 $, $ \rho=-0.7689 $, $ a=0.9051 $ and $ b=0.0023 $. To illustrate the impact of competition, we consider two insurers with different payment patterns for AAI 1 and AAI 2, aiming to investigate their effects on equilibrium reinsurance strategies. Specifically, we set the payment of AAI 1 to be low frequency but high payment, while the payment of AAI 2 is high frequency but low payment. For this scenario, we choose the following parameter values: $ \lambda_1 = 0.9,$ $ \lambda_2 = 2.4,$ $\hat{\lambda} = 0.6,$ $ \eta_1 = 0.2,$ $ \eta_2 = 0.2,$ and $ \hat{\eta} = 0.25. $ Additionally, we set $ \mu_{11} = 1,$ $ \mu_{12} = 2,$ $ \mu_{21} = 1/2,$ $ \mu_{22} = 1/2.$ Furthermore, the default values of the other parameters are as follows: $ \delta_1 = 2,$ $ \delta_2 = 3,$ $ \theta_1 = 0.7,$ $ \theta_2 = 0.7,$ $ \Psi_{1} = 5,$ $ \Psi_{2} = 7. $

\begin{figure}[H]
	\centering
	\begin{minipage}[t]{0.45\linewidth}
		\centering
		\includegraphics[width=0.9\linewidth]{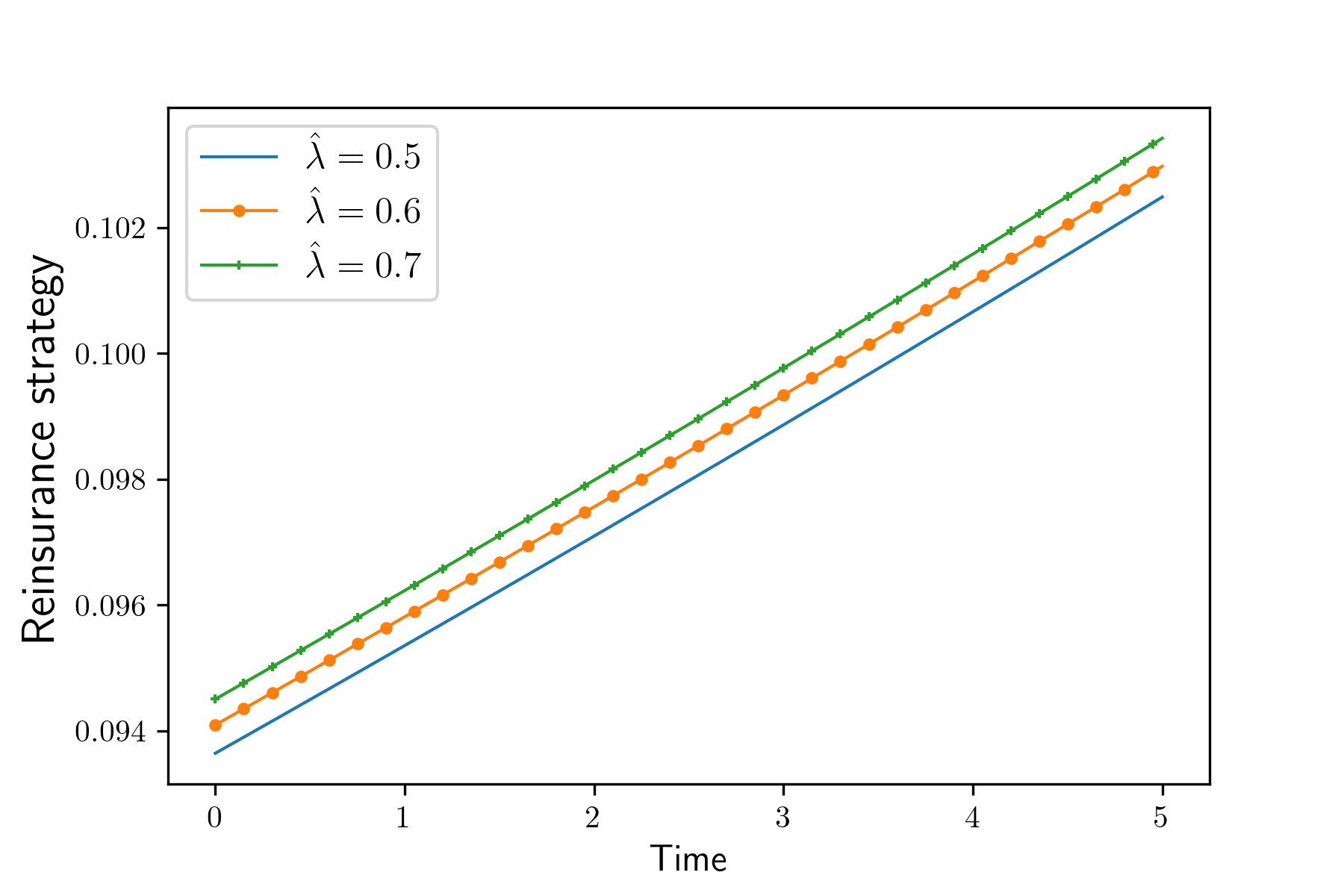}
		\caption{Effects of $ \hat{\lambda} $ on $ a^*_1(t) $}
		\label{lambda_hat_Reinsurance}
	\end{minipage}
	\begin{minipage}[t]{0.45\linewidth} 
		\centering
		\includegraphics[width=0.9\linewidth]{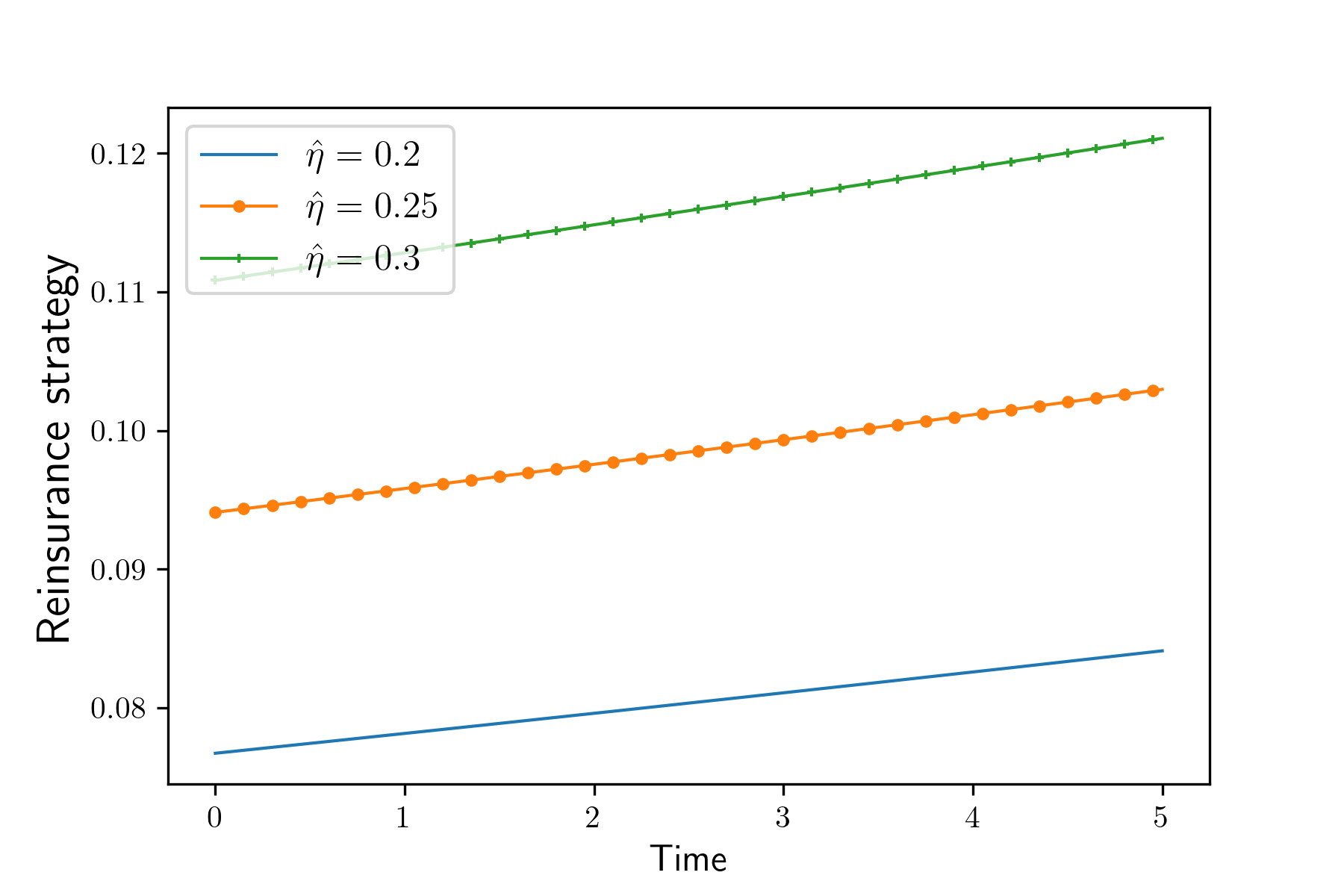}
		\caption{Effects of $ \hat{\eta} $ on $ a^*_1(t) $}
		\label{eta_hat_Reinsurance}
	\end{minipage}
\end{figure}

Our analysis primarily aims to understand how competition, risk aversion, and ambiguity aversion impact the equilibrium reinsurance strategies of AAIs. Figs.~\ref{lambda_hat_Reinsurance}-\ref{eta_hat_Reinsurance} showcase the influence of two crucial parameters, $\hat{\lambda}$ and $\hat{\eta}$, on reinsurance strategies. The parameter $\hat{\lambda}$ signifies the intensity of the common insurance business, reflecting AAIs' profitability potential by engaging in more insurance business. With premiums calculated based on the expected value principle, a higher $\hat{\lambda}$ indicates greater profit potential from writing more insurance policies. Consequently, the reinsurance strategy increases with $\hat{\lambda}$ (Fig.~\ref{lambda_hat_Reinsurance}). On the other hand, $\hat{\eta}$ characterizes the reinsurance premium, representing the cost of transferring insurance risk to a reinsurer. Higher premiums for sharing insurance risk may lead insurers to opt for retaining more risk themselves. Fig.~\ref{eta_hat_Reinsurance} illustrates the positive relationship between $\hat{\eta}$ and the reinsurance strategy. It's worth noting that the independence between the financial market and the insurance market implies that parameters $\hat{\lambda}$ and $\hat{\eta}$ in the insurance market do not impact the robust equilibrium investment strategy, as observed in \eqref{n-opt}. This suggests that the AAIs' investment decisions remain uninfluenced by the intensity of the common insurance business or the reinsurance premium. Instead, their investment strategies depend on other factors, such as risk preferences and financial market dynamics.


\begin{figure}[H]
	\centering
	\begin{minipage}[t]{0.45\linewidth} 
		\centering
		\includegraphics[width=0.9\linewidth]{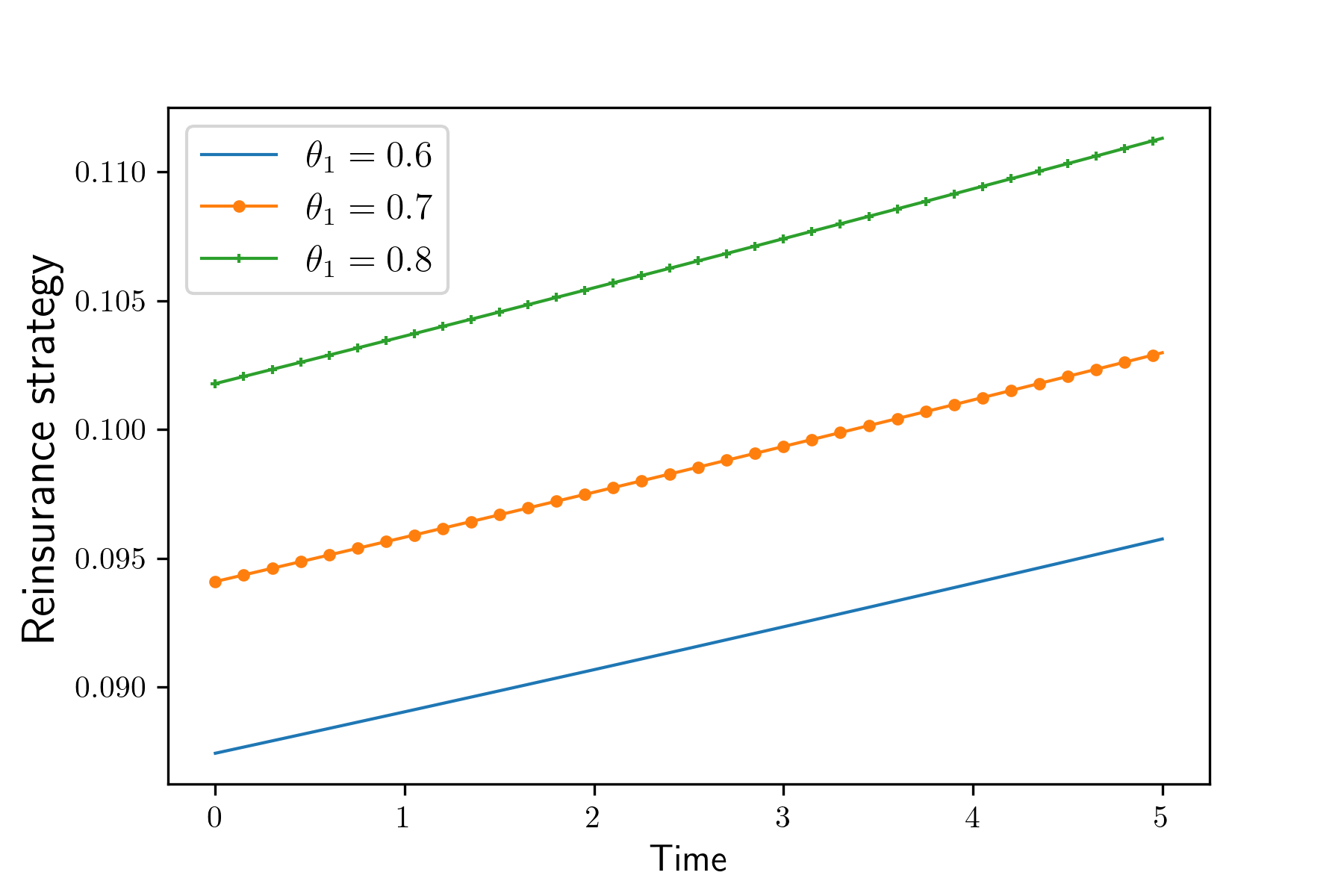}
		\caption{Effects of $ \theta_{1}$ on $ a^*_1(t) $}
		\label{theta1_Reinsurance}
	\end{minipage}
	\begin{minipage}[t]{0.45\linewidth}
		\centering
		\includegraphics[width=0.9\linewidth]{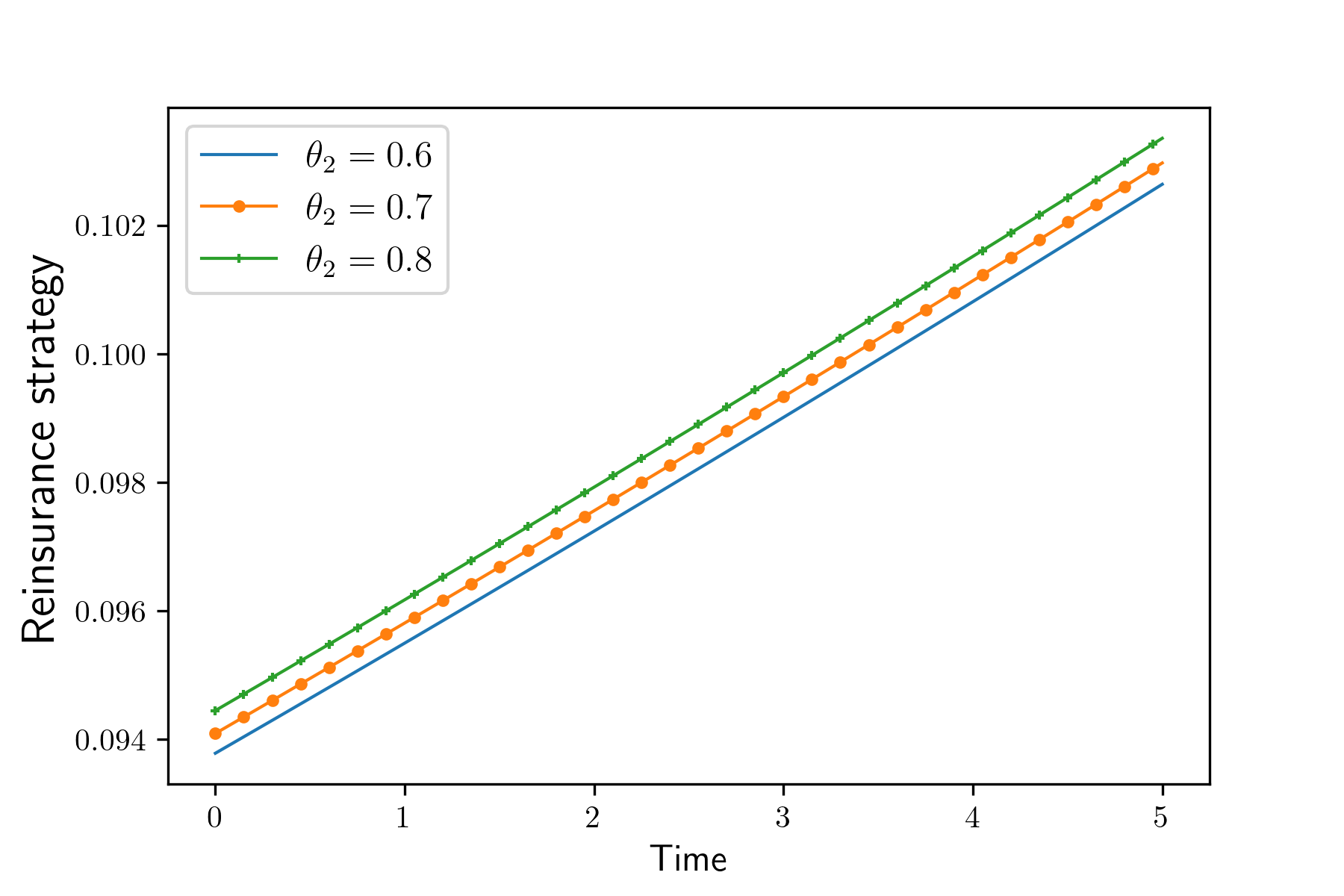}
		\caption{Effects of $ \theta_{2}$ on $ a^*_1(t) $}
		\label{theta2_Reinsurance}
	\end{minipage}
\end{figure}

\begin{figure}[H]
	\centering
	\begin{minipage}[t]{0.45\linewidth} 
		\centering
		\includegraphics[width=0.9\linewidth]{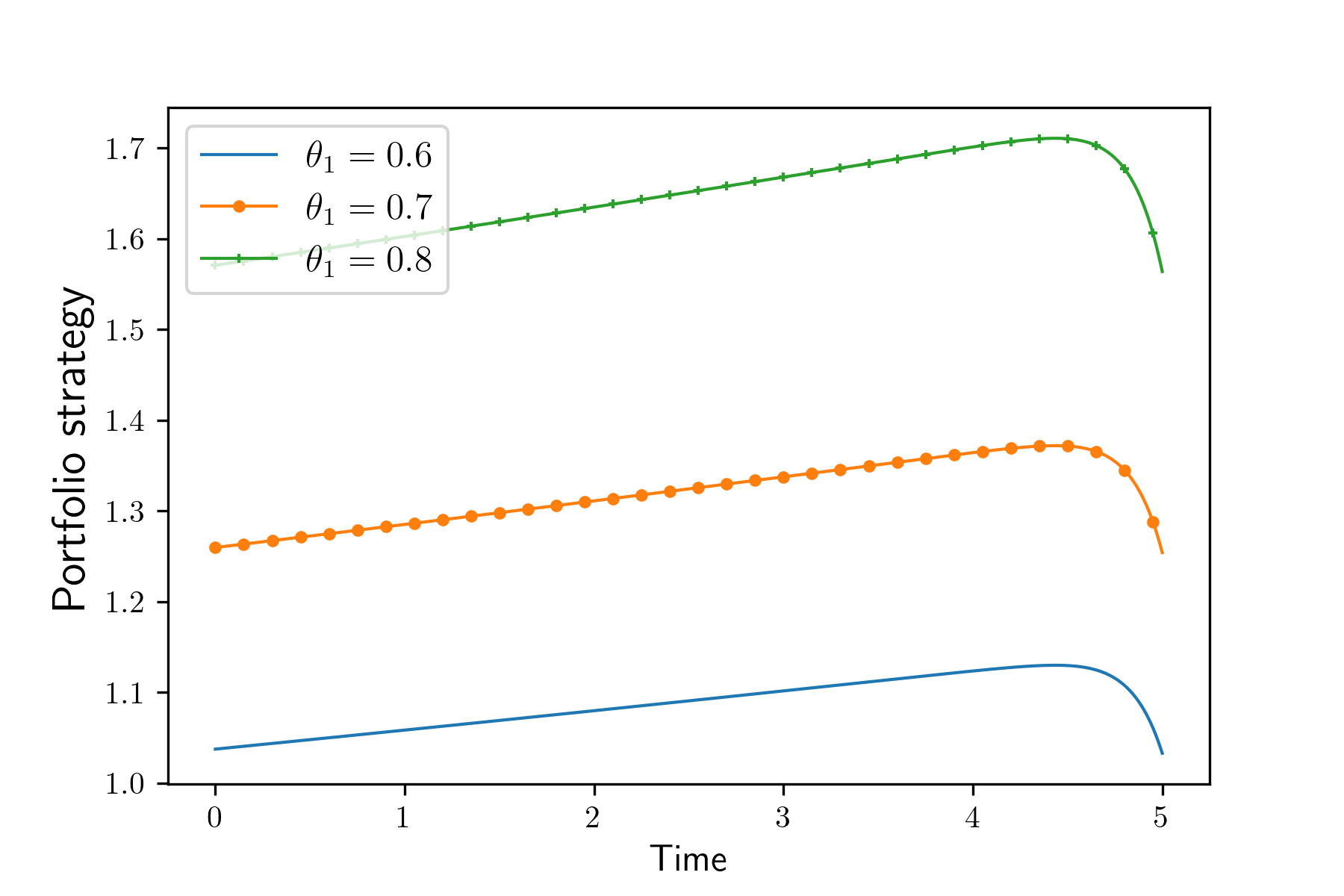}
		\caption{Effects of $ \theta_{1}$ on $ \pi^*_1(t) $}
		\label{theta1_Portfolio}
	\end{minipage}
	\begin{minipage}[t]{0.45\linewidth}
		\centering
		\includegraphics[width=0.9\linewidth]{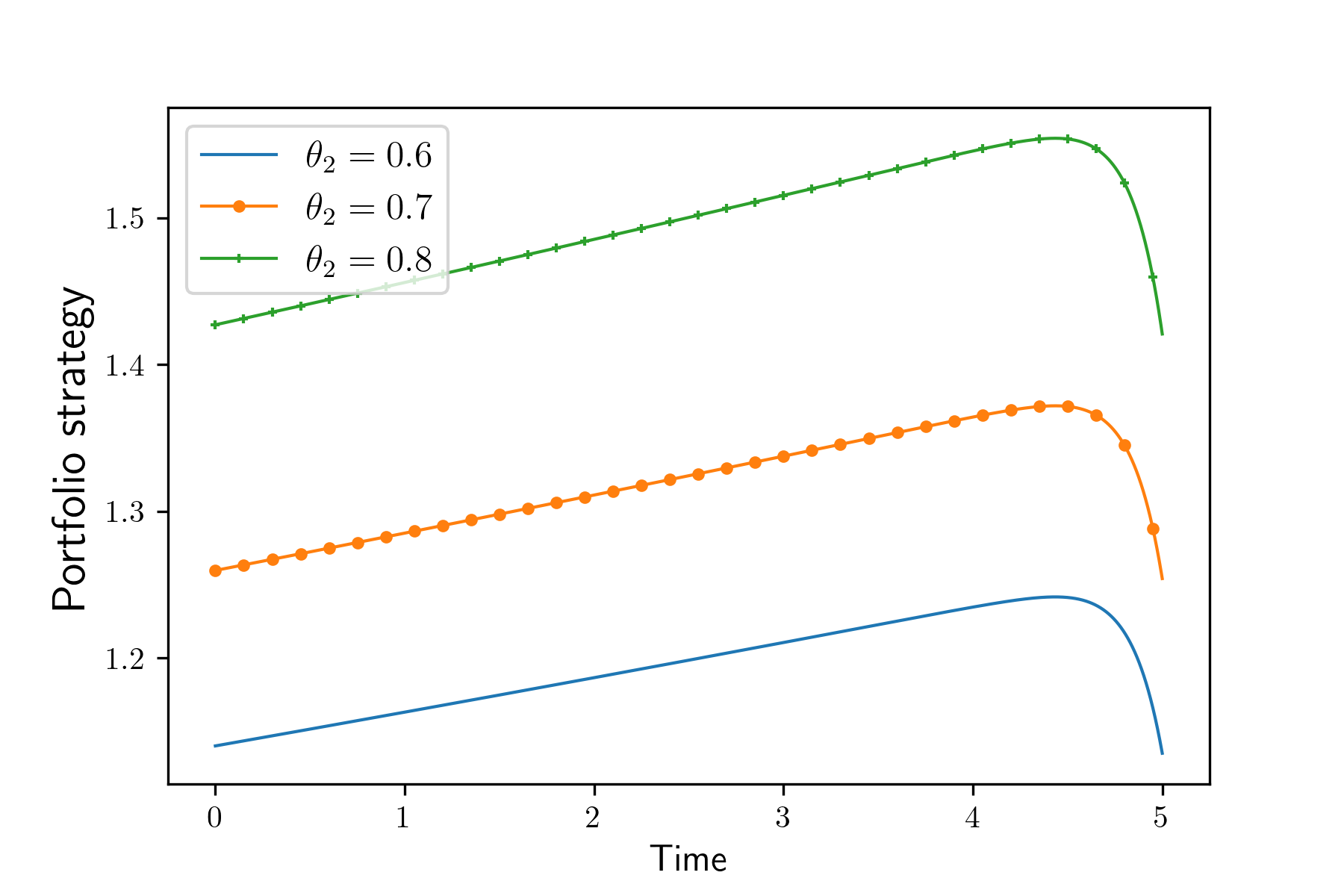}
		\caption{Effects of $ \theta_{2}$ on $ \pi^*_1(t) $}
		\label{theta2_Portfolio}
	\end{minipage}
\end{figure}

The parameters $\theta_1$ and $\theta_2$ reflect the levels of competition among the AAIs. Figs.~\ref{theta1_Reinsurance} and \ref{theta2_Reinsurance} show the impact of competition on AAI 1's reinsurance strategy. When AAI 1 prioritizes relative performance, she  tends to adopt a more risk-seeking approach to outperform others. Consequently, in Fig.~\ref{theta1_Reinsurance}, AAI 1's reinsurance strategy increases with $\theta_1$, indicating a greater acceptance of insurance risk when focusing more on relative performance. Fig.~\ref{theta2_Reinsurance} illustrates AAI 2's influence on AAI 1's reinsurance strategy. We observe that $a_1^*$ has a positive relationship with $\theta_2$. As $\theta_2$ rises, AAI 2 selects a larger $a_2^*$. In response, AAI 1 also adjusts to a larger $a_1^*$ in this competitive scenario. Similarly, Figs.~\ref{theta1_Portfolio} and \ref{theta2_Portfolio} depict the effects of competition on the robust equilibrium investment strategies. We note that $\pi^*_1$ increases with both $\theta_1$ and $\theta_2$. In a competitive setting, if one AAI prioritizes relative performance and adopts a risk-seeking strategy, others tend to allocate more to the risky asset. This reveals that competition in either the insurance market or the financial market drives insurers to adopt more aggressive investment strategies.
\begin{figure}[H]
	\centering
	\begin{minipage}[t]{0.45\linewidth} 
		\centering
		\includegraphics[width=0.9\linewidth]{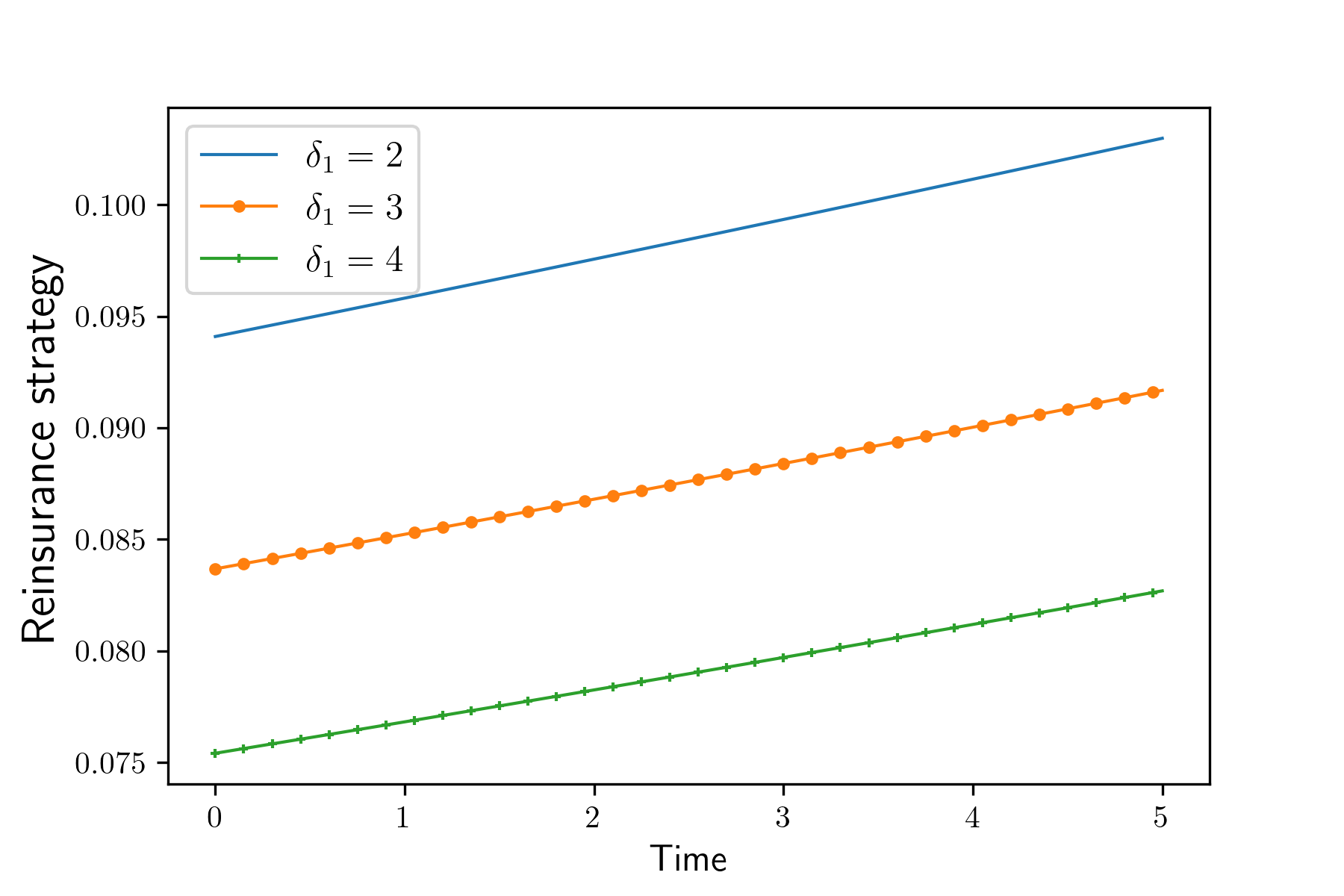}
		\caption{Effects of $ \delta_{1}$ on $ a^*_1(t) $}
		\label{delta1_Reinsurance}
	\end{minipage}
	\begin{minipage}[t]{0.45\linewidth}
		\centering
		\includegraphics[width=0.9\linewidth]{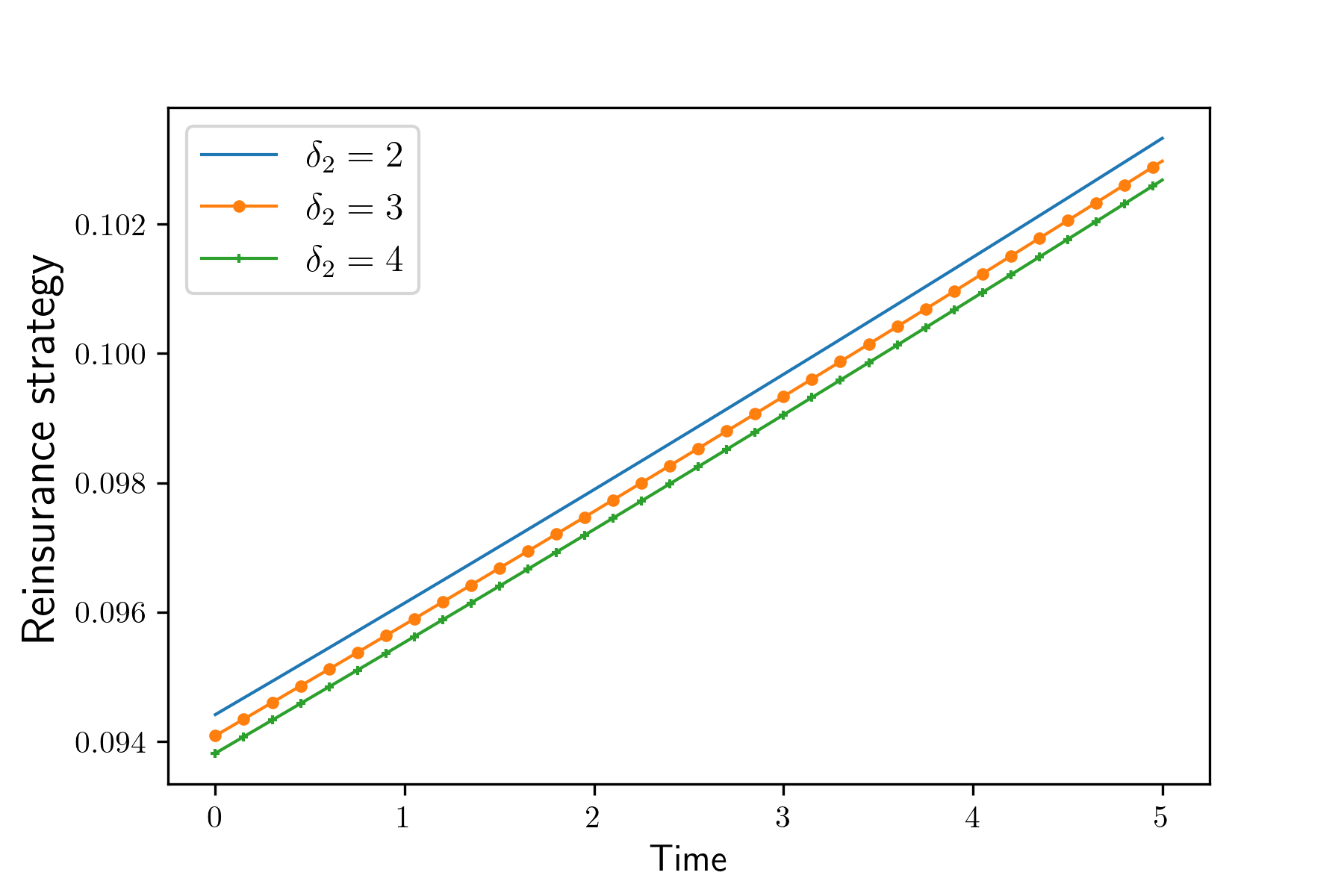}
		\caption{Effects of $ \delta_{2}$ on $ a^*_1(t) $}
		\label{delta2_Reinsurance}
	\end{minipage}
\end{figure}

The parameters $\delta_1$ and $\delta_2$ represent the risk aversion levels of AAIs. Fig.~\ref{delta1_Reinsurance} illustrates the effect of $\delta_1$ on AAI 1's reinsurance strategy in a robust equilibrium. As $\delta_1$ increases, AAI 1 becomes more risk-averse, resulting in a reduced willingness to take on insurance risk independently. Consequently, AAI 1 transfers more risk to the reinsurer, leading to a lower retention level $a_1^*$, as depicted in Fig.~\ref{delta1_Reinsurance}. The influence of AAI 2's risk aversion parameter, $\delta_2$, on AAI 1's reinsurance strategy is portrayed in Fig.~\ref{delta2_Reinsurance}. Higher values of $\delta_2$ indicate that AAI 2 is more willing to share insurance risk. AAI 1 responds accordingly in the competition, leading to a decrease in $a_1^*$ as AAI 2's risk aversion parameter increases (see Fig.~\ref{delta2_Reinsurance}).

\begin{figure}[H]
	\centering
	\begin{minipage}[t]{0.45\linewidth} 
		\centering
		\includegraphics[width=0.9\linewidth]{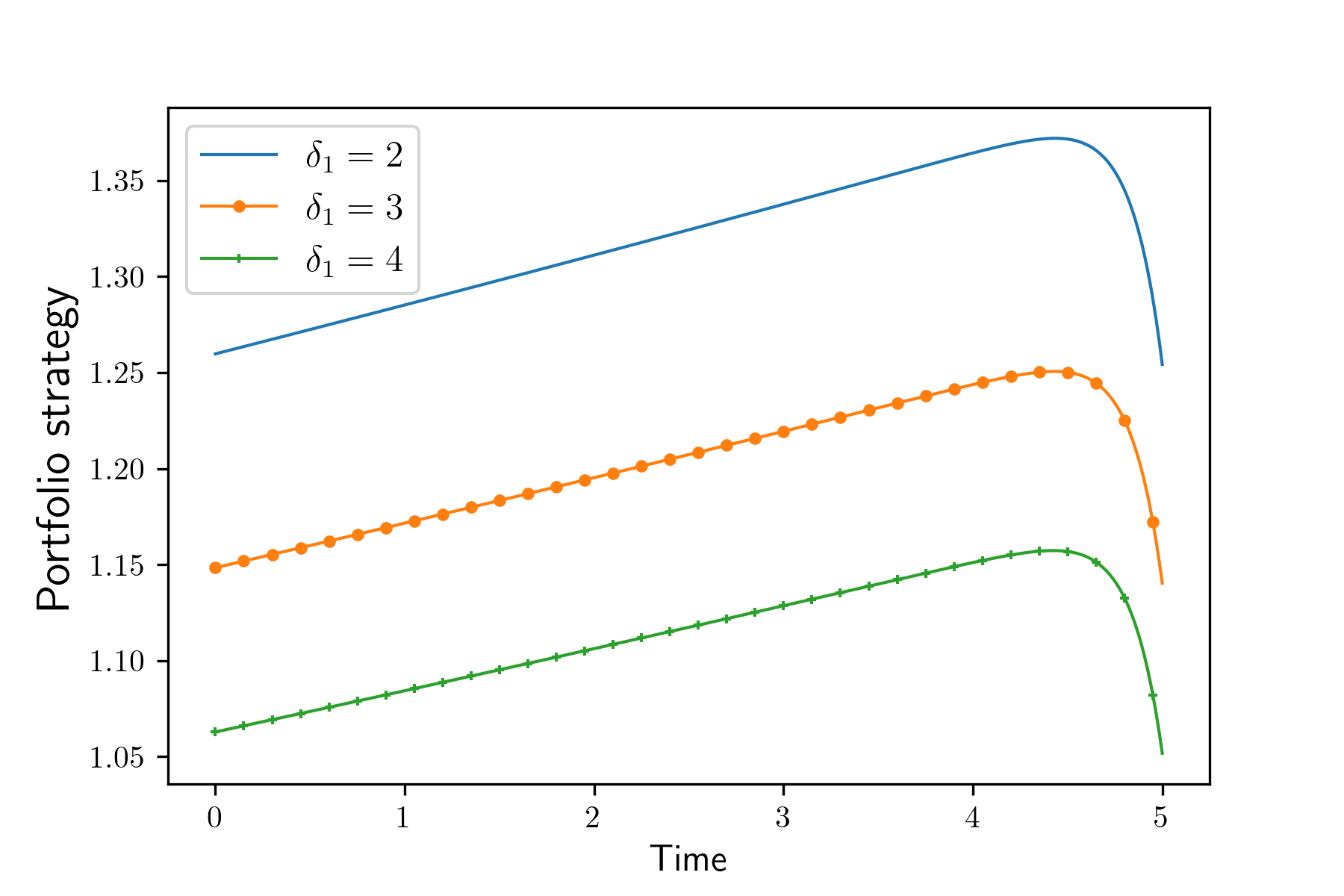}
		\caption{Effects of $ \delta_{1}$ on $ \pi^*_1(t) $}
		\label{delta1_Portfolio}
	\end{minipage}
	\begin{minipage}[t]{0.45\linewidth}
		\centering
		\includegraphics[width=0.9\linewidth]{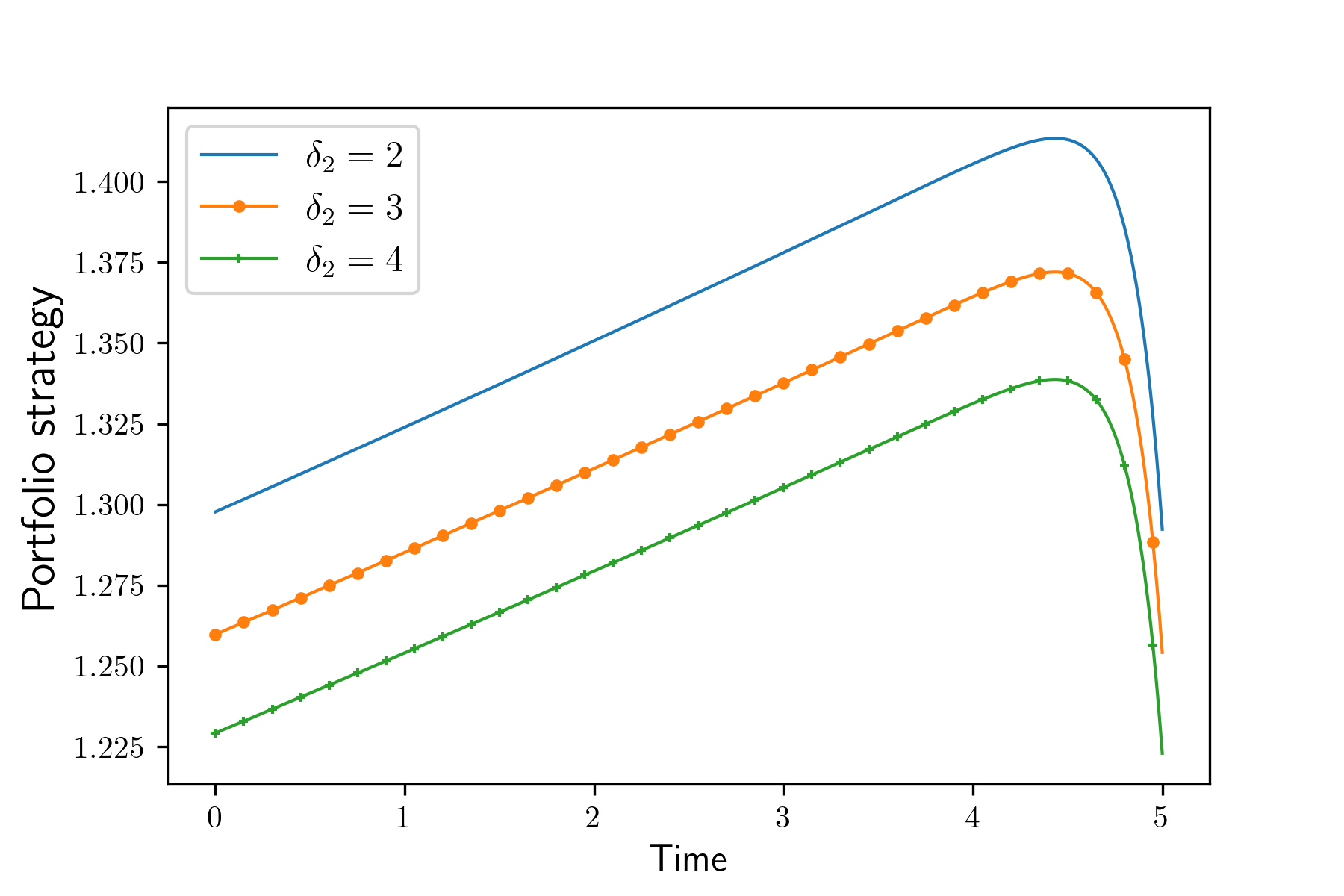}
		\caption{Effects of $ \delta_{2}$ on $ \pi^*_1(t) $}
		\label{delta2_Portfolio}
	\end{minipage}
\end{figure}

Fig.~\ref{delta1_Portfolio} illustrates the relationship between AAI 1's risk aversion parameter, $\delta_1$, and the investment strategy in the financial market. As AAI 1 becomes more risk-averse (with a larger $\delta_1$), she adopts a more conservative approach by reducing her  allocation in risky assets. Consequently, there is a negative impact of $\delta_1$ on the optimal investment strategy, $\pi_1^*$. Additionally, Fig.~\ref{delta2_Portfolio} demonstrates how AAI 2's risk aversion attitude influences AAI 1's investment strategy. When AAI 2 exhibits higher levels of risk aversion, AAI 1 tends to adopt a more conservative investment strategy. Conversely, if AAI 2 is more risk-seeking, AAI 1 is more inclined to pursue a more aggressive investment approach. These findings align with the observations in Figs.~\ref{delta1_Reinsurance} and \ref{delta2_Reinsurance}, which depict the effects of $\delta_1$ and $\delta_2$ on AAI 1's reinsurance strategy. In the financial market, different AAIs' behaviors tend to converge as they respond to each other's risk management strategies.

\begin{figure}[H]
	\centering
	\begin{minipage}[t]{0.45\linewidth} 
		\centering
		\includegraphics[width=0.9\linewidth]{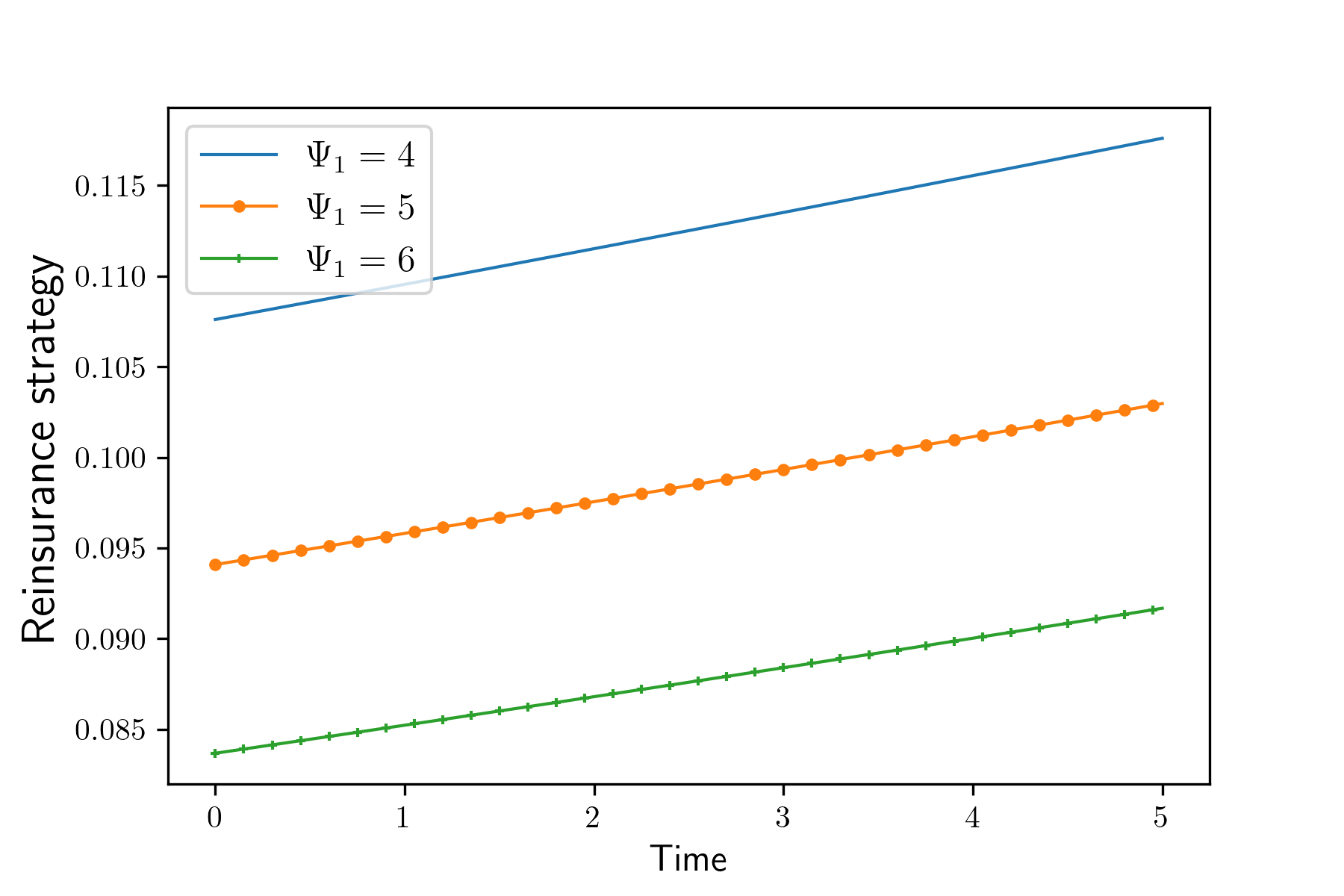}
		\caption{Effects of $ \Psi_{1}$ on $ a^*_1(t) $}
		\label{psi1_Reinsurance}
	\end{minipage}
	\begin{minipage}[t]{0.45\linewidth}
		\centering
		\includegraphics[width=0.9\linewidth]{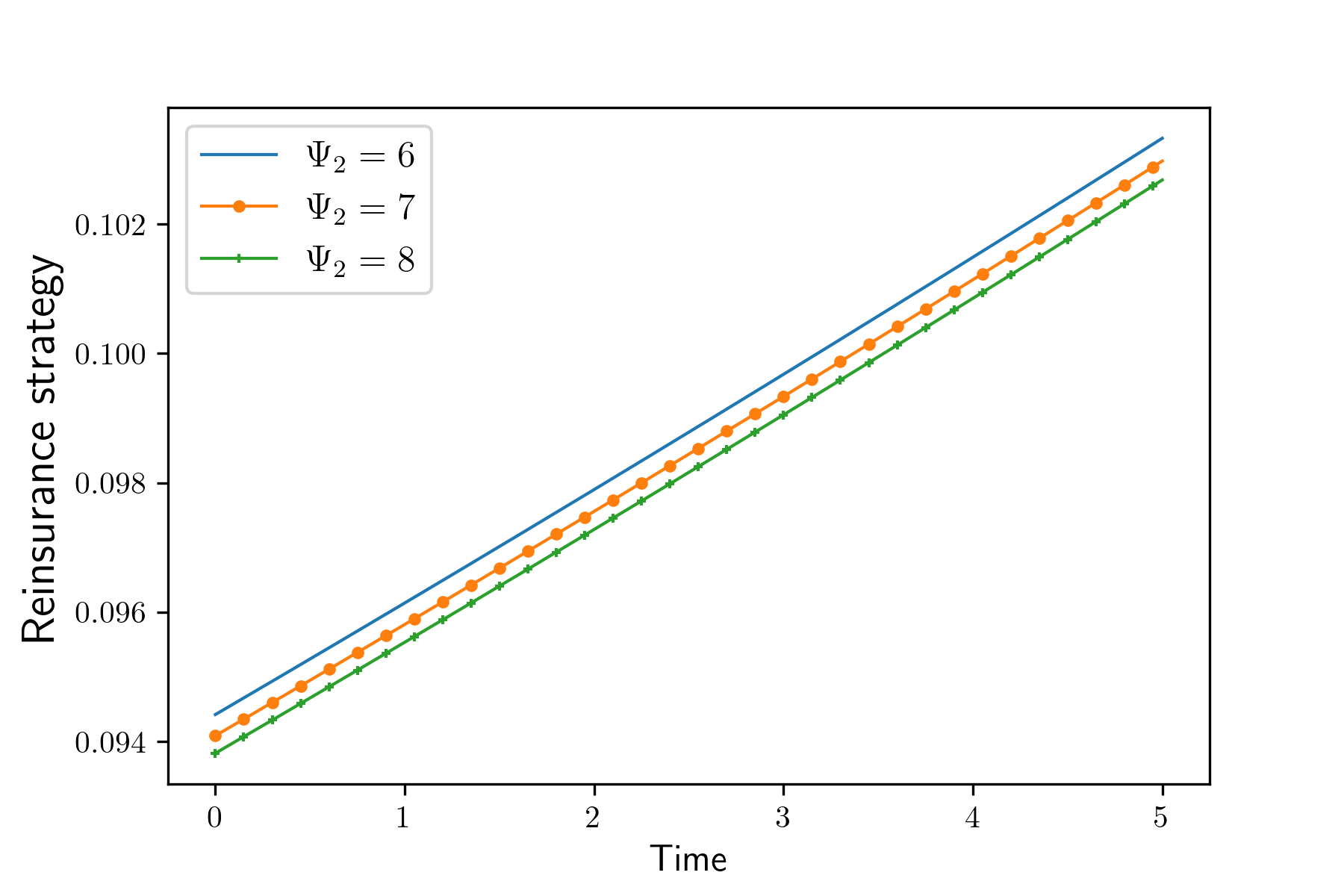}
		\caption{Effects of $ \Psi_{2}$ on $ a^*_1(t) $}
		\label{psi2_Reinsurance}
	\end{minipage}
\end{figure}
\begin{figure}[H]
	\centering
	\begin{minipage}[t]{0.45\linewidth} 
		\centering
		\includegraphics[width=0.9\linewidth]{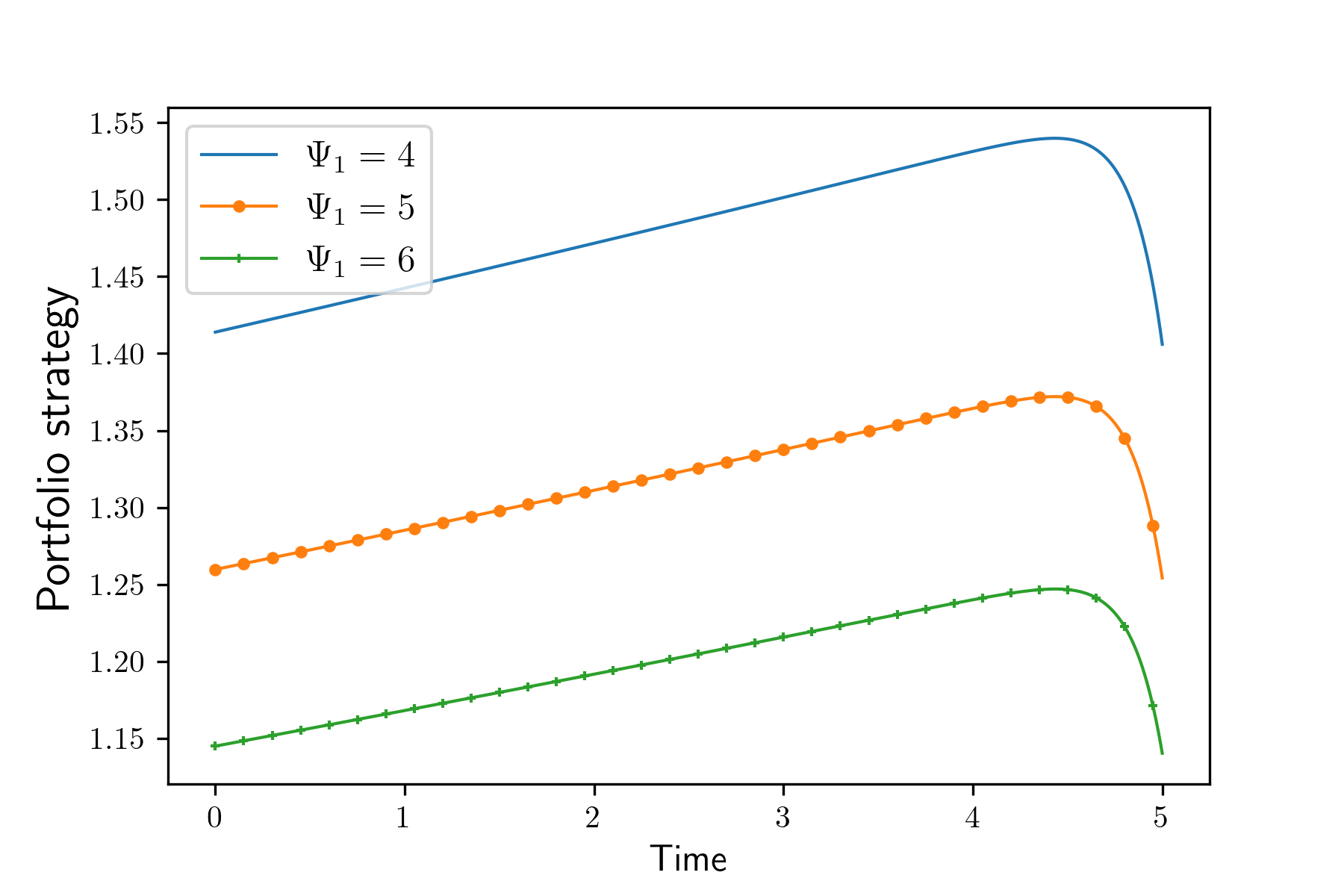}
		\caption{Effects of $ \Psi_{1}$ on $ \pi^*_1(t) $}
		\label{psi1_Portfolio}
	\end{minipage}
	\begin{minipage}[t]{0.45\linewidth}
		\centering
		\includegraphics[width=0.9\linewidth]{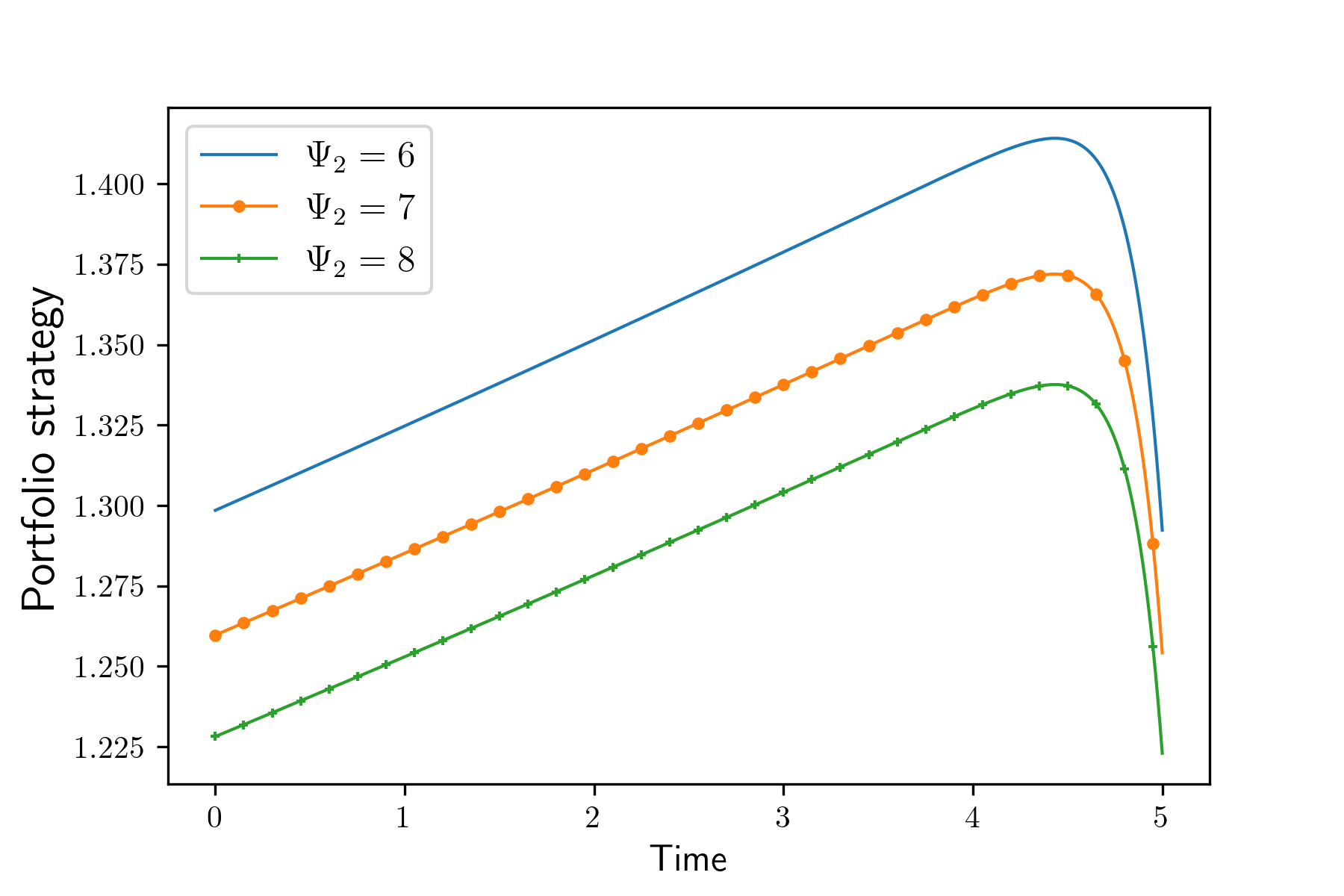}
		\caption{Effects of $ \Psi_{2}$ on $ \pi^*_1(t) $}
		\label{psi2_Portfolio}
	\end{minipage}
\end{figure}

The parameters $\Psi_1$ and $\Psi_2$ quantify the level of ambiguity aversion, representing an AAI's aversion to uncertain outcomes. In Fig.~\ref{psi1_Reinsurance}, it is observed that higher levels of ambiguity aversion for AAI 1 result in reduced confidence in the financial model. As a consequence, AAI 1 becomes more cautious in managing increased model uncertainty, leading to a decrease in her  optimal reinsurance allocation, denoted as $a^*_1$. Furthermore, the impact of AAI 2's ambiguity attitude on AAI 1's reinsurance strategy is examined in Figs.~\ref{psi2_Reinsurance}. These figures demonstrate a positive relationship between AAI 1's reinsurance strategy and AAI 2's ambiguity attitude. Specifically, when AAI 2's ambiguity aversion parameter, $\Psi_2$, increases, both AAI 2 and AAI 1 exhibit a reduced willingness to take risks in the competition game. These findings align with the explanations provided for Figs.~\ref{theta2_Reinsurance} and \ref{delta2_Reinsurance}, where increasing ambiguity aversion leads to a decrease in risk-taking behavior for both AAIs in the competition game.

Fig.~\ref{psi1_Portfolio} provides insights into the impact of ambiguity attitude on investment strategy. It reveals that as AAI 1 becomes more ambiguity averse, she exhibit less confidence in the financial market. Consequently, there is a decrease in the optimal investment strategy, $\pi_1^*$, as $\Psi_1$ increases. A similar influence is observed when considering the ambiguity aversion coefficient, $\Psi_2$, of Insurer 2. As $\Psi_2$ increases, AAI 2 becomes more risk-averse, resulting in a decrease in her allocation to risky assets. Consequently, AAI 1 also decreases her investment in stocks, as depicted in Fig.~\ref{psi2_Portfolio}.

The observed phenomena confirm the presence of a herd effect in a competitive environment. In such a setting, when there are changes in risk exposures, AAIs tend to adjust their risk management strategies in a similar direction to their competitors. This behavior confirms the herd effect of competition in \cite{deng2018non}. 

\section{\bf Conclusion}

This paper integrates model uncertainty and the 4/2 stochastic volatility model to introduce robust n-insurer and mean-field games for competitive insurers under mean-variance criterion. We formulate the robust mean-field game in a non-linear system. By the $n$-dimensional extended HJBI equations, we derive the ``weak equilibrium strategy'', referred to as the robust time-consistent investment-reinsurance response strategy in this paper, for both the $n$-insurer case and the mean-field case. 
We also demonstrate that as the number of AAIs, $n$, tends toward infinity, the robust equilibrium strategies in the $n$-insurer setting converge to those in the robust mean-field equilibrium. By solving the coupled Riccati equations, which is a novel contribution compared to previous research, we obtain the robust equilibrium strategies.  We provide suitable conditions outlined in the verification theorem. Furthermore, we present numerical examples to illustrate insurers' economic behaviors and examine the influence of different parameters on the robust equilibrium strategies. Our findings indicate that even when considering ambiguity, the herd effect of competition still exists. Insurers' strategies now depend not only on the risk attitude of their peers but also on their peers' ambiguity attitude.

\vskip 15pt
{\bf Acknowledgements.} The authors acknowledge the support from the National Natural Science Foundation of China (Grant No.12271290, No.11901574, No.11871036), the MOE Project of Key Research Institute of Humanities and Social Sciences (22JJD910003). The authors thank the members of the group of Actuarial Science and Mathematical Finance at the Department of Mathematical Sciences, Tsinghua University for their feedbacks and useful conversations.
\vskip 15pt
	\appendix
	\renewcommand{\theequation}{\thesection.\arabic{equation}}
	
\section{\bf Proof of Theorem \ref{solution-HJBI}}\label{proof-solution-HJBI}
Based on \eqref{equ:af}, we have\begin{equation*}
		\begin{aligned}
			&\delta_{i} \Upsilon(t,y,z) \mathcal{A}^{\{(\pi_k,a_k)_{k=1}^n\},(\varphi,\chi,\phi,\vartheta)}\Upsilon(t,y,z)-\frac{\delta_{i}}{2}\mathcal{A}^{\{(\pi_k,a_k)_{k=1}^n\},(\varphi,\chi,\phi,\vartheta)}\Upsilon^2(t,y,z)\\
			=&-\nu^2\frac{\delta_{i}}{2}z(\Upsilon_{z})^2\!-\delta_{i}\nu\sqrt{z}\rho\sigma((1\!-\!\frac{\theta_{i}}{n})\pi_i\!-\!\frac{\theta_{i}}{n}\sum_{k\neq i}\pi_{k}(t))\Upsilon_{y}\Upsilon_{z}\\
			&-\frac{\delta_{i}}{2}(1-\frac{\theta_i}{n})^2a_{i}^2(t) \left(\hat{\lambda}+\lambda_{i}\right) \mu_{i 2}(\Upsilon_y)^2-\frac{\delta_{i}\theta_{i}^2}{2n^2}\sum_{k\neq i}(a_{k}(t))^2 \left(\left(\hat{\lambda}\!+\!\lambda_{k}\right) \mu_{k 2}-\hat{\lambda}\mu_{k1}^2\right)(\Upsilon_y)^2\\&-\frac{\delta_{i}\hat{\lambda}\theta_i^2}{2n^2}\left[\sum_{k\neq i}a_k(t)\mu_{k 1}\right]^2\!\!(\Upsilon_y)^2\\
			&+\delta_{i}\hat{\lambda}\frac{\theta_{i}}{n}(1-\frac{\theta_i}{n})a_i(t)\mu_{i 1}\sum_{k\neq i}a_{k}(t)\mu_{k 1}(\Upsilon_y)^2-\frac{\delta_{i}}{2}\sigma^2((1\!-\!\frac{\theta_{i}}{n})\pi_i\!-\!\frac{\theta_{i}}{n}\sum_{k\neq i}\pi_{k}(t))^2(\Upsilon_y)^2.
		\end{aligned}
	\end{equation*}

	Suppose that the two-tuple $ (v^{(i)}, \varUpsilon^{(i)}) $ solves both (\ref{HJBI-MV}) and (\ref{g-MV}). As the operator\\ $ \mathcal{L}_{i}\left({\left\{(\pi_{k}^*,a_k^*)_{k\neq i},(\pi_{i},a_i)\right\},(\varphi_i,\chi_i,\phi_i,\vartheta_{i})},V,\Upsilon,\Psi_i,(t,y,z)\right) $ is quadratic with respect to $ (\varphi_i,\chi_i,\phi_i,\vartheta_{i}) $, we can obtain the values of $\left(\varphi_{i,\pi_i,a_i}^\circ, \chi_{i,\pi_i,a_i}^\circ,\phi_{i,\pi_i,a_i}^\circ, \vartheta_{i,\pi_i,a_i}^\circ\right)$ by applying the first-order condition:
	\begin{equation*}
		\left\{\begin{aligned}
			&\nu\sqrt{z}\rho v^{(i)}_{z}+((1-\frac{\theta_{i}}{n})\pi_i-\frac{\theta_{i}}{n}\sum_{k\neq i}\pi^*_{k})\sigma v^{(i)}_y+\frac{\varphi^\circ_{{i,\pi_i,a_i}}}{\Psi_i} =0,\\
			&\nu\sqrt{z}\sqrt{1-\rho^2} v^{(i)}_{z}+\frac{\chi^\circ_{i,\pi_i,a_i}}{\Psi_i}=0,\\
			&\sqrt{\hat{\lambda}}(1-\frac{\theta_i}{n})\mu_{i1}a_{i}(t)v^{(i)}_y-\sqrt{\hat{\lambda}}\frac{\theta_{i}}{n}\sum_{k\neq i}\mu_{k1}a_{k}^*(t)v^{(i)}_y+\frac{\phi^\circ_{{i,\pi_i,a_i}}}{\Psi_i}=0,\\
			&(1-\frac{\theta_i}{n})a_{i}(t)\sqrt{(\hat{\lambda}+\lambda_{i})\mu_{i 2}-\hat{\lambda}\mu_{i 1}^2}v^{(i)}_y+\frac{\vartheta^\circ_{{i,i,\pi_i,a_i}}}{\Psi_i}=0,\\&-\frac{\theta_i}{n}a_{k}^*(t)\sqrt{(\hat{\lambda}+\lambda_{k})\mu_{k 2}-\hat{\lambda}\mu_{k 1}^2}v^{(i)}_y+\frac{\vartheta^\circ_{{i,k,\pi_i,a_i}}}{\Psi_i}=0,~k\neq i,
		\end{aligned}\right.
	\end{equation*}that is,
\begin{equation}\label{equ:px}
		\left\{\begin{aligned}
			&{\varphi^\circ_{i,\pi_i,a_i}}=-\nu\sqrt{z}\rho v^{(i)}_{z}\Psi_i-((1-\frac{\theta_{i}}{n})\pi_i-\frac{\theta_{i}}{n}\sum_{k\neq i}\pi^*_{k})\sigma v^{(i)}_y\Psi_i,\\
			&{\chi^\circ_{i,\pi_i,a_i}}=-\nu\sqrt{z}\sqrt{1-\rho^2} v^{(i)}_{z}\Psi_i,\\
			&\phi^\circ_{{i,\pi_i,a_i}}=-\sqrt{\hat{\lambda}}\left[(1-\frac{\theta_i}{n})\mu_{i1}a_{i}(t)-\frac{\theta_{i}}{n}\sum_{k\neq i}\mu_{k1}a_{k}^*(t)\right]v^{(i)}_y\Psi_i,\\
			&\vartheta^\circ_{{i,i,\pi_i,a_i}}=-(1-\frac{\theta_i}{n})a_{i}(t)\sqrt{(\hat{\lambda}+\lambda_{i})\mu_{i 2}-\hat{\lambda}\mu_{i 1}^2}v^{(i)}_y\Psi_i,\\
			&\vartheta^\circ_{{i,k,\pi_i,a_i}}=\frac{\theta_i}{n}a_{k}^*(t)\sqrt{(\hat{\lambda}+\lambda_{k})\mu_{k 2}-\hat{\lambda}\mu_{k 1}^2}v^{(i)}_y\Psi_i,~k\neq i.
		\end{aligned}\right.
	\end{equation}
Upon substituting the previous equations into (\ref{HJBI-MV}), it becomes evident that (\ref{HJBI-MV}) is a quadratic equation in terms of $\pi_i$ and $a_i$. By applying the first-order condition, we can derive the values of $\left(\hat{\pi}^\circ_i, \hat{a}^\circ_i\right)$. Consequently, we can obtain the following expressions:
	\begin{equation}\label{equ:v1}
		\left\{	\begin{aligned}
			&\nu\rho\sqrt{z}( v^{(i)}_{yz}-v^{(i)}_yv^{(i)}_{z}\Psi_i-\delta_{i}\Upsilon_{y}\Upsilon_{z})+m\sqrt{z} v^{(i)}_y+\sigma((1-\frac{\theta_i}{n})\hat{\pi}^\circ_i-\frac{\theta_i}{n}\sum_{k\neq i}\pi^*_k)(v^{(i)}_{yy}-(v^{(i)}_y)^2\Psi_i-{\delta_{i}}(\Upsilon_y)^2)=0,\\
			&2\hat{\eta}(\lambda_i+\hat{\lambda})\mu_{i 2}(1-\hat{a}^\circ_i(t))v^{(i)}_y+\hat{\lambda}\frac{\theta_{i}}{n}\sum_{k\neq i}a^*_{k}(t)\mu_{i 1}\mu_{k 1}((v^{(i)}_y)^2\Psi_i-v^{(i)}_{yy})\\&+ (1-\frac{\theta_i}{n})\hat{a}^\circ_i(t)\left[(\lambda_i+\hat{\lambda})\mu_{i2}v^{(i)}_{yy}-(\hat{\lambda}+\lambda_{i})\mu_{i 2}{\Psi_i}(v^{(i)}_y)^2\right]\\&-\delta_{i}(1-\frac{\theta_i}{n})\hat{a}^\circ_i(t) \left(\hat{\lambda}+\lambda_{i}\right) \mu_{i 2}(\Upsilon_y)^2+\delta_{i}\hat{\lambda}\frac{\theta_{i}}{n}\mu_{i 1}\sum_{k\neq i}a_{k}^*(t)\mu_{k 1}(\Upsilon_y)^2=0,\\
			&2\hat{\eta}(\lambda_i+\hat{\lambda})\mu_{i 2}(1-\hat{a}^\circ_i(t))v^{(i)}_y+\hat{\lambda}\frac{\theta_{i}}{n}\sum_{k\neq i}a^*_{k}(t)\mu_{i 1}\mu_{k 1}((v^{(i)}_y)^2\Psi_i-v^{(i)}_{yy}+\delta_{i}(\Upsilon_y)^2)\\&+ (1-\frac{\theta_i}{n})\hat{a}^\circ_i(t)\left[(\lambda_i+\hat{\lambda})\mu_{i2}(v^{(i)}_{yy}-\delta_{i}(\Upsilon_y)^2-\Psi_i(v^{(i)}_y)^2)\right]=0.
		\end{aligned}\right.
	\end{equation}
	
We conjecture \begin{equation*}
		v^{(i)}(t, y,z)= v_{i,1}(t)+ yv_{i,2}(t)+v_{i,3}(t)z,\quad \varUpsilon^{(i)}(t, y,z)= \varUpsilon_{i,1}(t)+ y\varUpsilon_{i,2}(t)+\varUpsilon_{i,3}(t)z.
	\end{equation*}
	Then $ v^{(i)}_y=v_{i,2}(t)$, $v^{(i)}_z=v_{i,3}(t)$, $ v^{(i)}_{yy}=v^{(i)}_{yz}=v^{(i)}_{zz}=0 $, $v^{(i)}_t=v_{i,1}'(t)+ yv_{i,2}'(t)+v_{i,3}'(t)z $,  $ \varUpsilon^{(i)}_y=\varUpsilon_{i,2}(t)$, $\varUpsilon^{(i)}_z=\varUpsilon_{i,3}(t)$, $ \varUpsilon^{(i)}_{yy}=\varUpsilon^{(i)}_{yz}=\varUpsilon^{(i)}_{zz}=0 $, $ v_{i,1}(T)=0$, $ v_{i,2}(T)=1,$  $v_{i,3}(T)=0,$  $\varUpsilon_{i,1}(T)=0,$  $\varUpsilon_{i,2}(T)=1,$  $\varUpsilon_{i,3}(T)=0 $ and \eqref{equ:v1} is equivalent to
	\begin{equation}\label{equ:pia1}
		\left\{	\begin{aligned}
			&\nu\rho\sqrt{z}( -v_{i,2}(t)v_{i,3}(t)\Psi_i\!-\!\delta_{i}\varUpsilon_{i,2}\varUpsilon_{i,3})\!+\!m\sqrt{z} v_{i,2}(t)\!+\!\sigma((1-\frac{\theta_i}{n})\hat{\pi}^\circ_i\!-\!\frac{\theta_i}{n}\sum_{k\neq i}\pi^*_k)(-v_{i,2}^2(t)\Psi_i\!-\!{\delta_{i}}\varUpsilon_{i,2}^2)\!=\!0,\\
			&2\hat{\eta}(\lambda_i+\hat{\lambda})\mu_{i 2}(1-\hat{a}^\circ_i(t))v_{i,2}+\hat{\lambda}\frac{\theta_{i}}{n}\sum_{k\neq i}a^*_{k}(t)\mu_{i 1}\mu_{k 1}(v_{i,2}^2\Psi_i+\delta_{i}\varUpsilon_{i,2}^2)\\&+ (1-\frac{\theta_i}{n})\hat{a}^\circ_i(t)(\lambda_i+\hat{\lambda})\mu_{i2}(-\delta_{i}\varUpsilon_{i,2}^2-\Psi_iv_{i,2}^2)=0.
		\end{aligned}\right.
	\end{equation}
and
	\begin{equation*}
		\left\{	\begin{aligned}
			&-\nu\rho\sqrt{z}( v_{i,2}v_{i,3}\Psi_i+\delta_{i}\varUpsilon_{i,2}\varUpsilon_{i,3})+m\sqrt{z} v_{i,2}(t)=\sigma((1-\frac{\theta_i}{n})\hat{\pi}^\circ_i-\frac{\theta_i}{n}\sum_{k\neq i}\pi^*_k)(v_{i,2}^2(t)\Psi_i+{\delta_{i}}\varUpsilon_{i,2}^2),\\
			&2\hat{\eta}(\lambda_i+\hat{\lambda})\mu_{i 2}(1-\hat{a}^\circ_i(t))v_{i,2}+\hat{\lambda}\frac{\theta_{i}}{n}\sum_{k\neq i}a^*_{k}(t)\mu_{i 1}\mu_{k 1}(v_{i,2}^2\Psi_i+\delta_{i}\varUpsilon_{i,2}^2)\\&= (1-\frac{\theta_i}{n})\hat{a}^\circ_i(t)(\lambda_i+\hat{\lambda})\mu_{i2}(\delta_{i}\varUpsilon_{i,2}^2+\Psi_iv_{i,2}^2).
		\end{aligned}\right.
	\end{equation*}
Based on \eqref{equ:pia1}, we obtain the following expressions:
	\begin{equation}\label{equ:pia2}
		\left\{	\begin{aligned}
			&(1-\frac{\theta_i}{n})\hat{\pi}^\circ_i-\frac{\theta_i}{n}\sum_{k\neq i}\pi^*_k= \frac{m v_{i,2}-\nu\rho( v_{i,2}v_{i,3}\Psi_i+\delta_{i}\varUpsilon_{i,2}\varUpsilon_{i,3})}{v_{i,2}^2(t)\Psi_i+{\delta_{i}}\varUpsilon_{i,2}^2}\frac{ \sqrt{z}}{\sigma},\\
			&2\hat{\eta}(\lambda_i+\hat{\lambda})\mu_{i 2}v_{i,2}+\hat{\lambda}\frac{\theta_{i}}{n}\sum_{k\neq i}a^*_{k}(t)\mu_{i 1}\mu_{k 1}(v_{i,2}^2\Psi_i+\delta_{i}\varUpsilon_{i,2}^2)\\&= \hat{a}^\circ_i(t)(\lambda_i+\hat{\lambda})\mu_{i2}((1-\frac{\theta_i}{n})(\delta_{i}\varUpsilon_{i,2}^2+\Psi_iv_{i,2}^2)+2\hat{\eta}v_{i,2}).
		\end{aligned}\right.
	\end{equation}
	Denote
	\begin{equation*}
		\begin{aligned}
			&S_i(t)=\frac{m v_{i,2}-\nu\rho( v_{i,2}v_{i,3}\Psi_i+\delta_{i}\varUpsilon_{i,2}\varUpsilon_{i,3})}{v_{i,2}^2(t)\Psi_i+{\delta_{i}}\varUpsilon_{i,2}^2},\\
			&R_i(t)=(\lambda_i+\hat{\lambda})\mu_{i2}((1-\frac{\theta_i}{n})(\delta_{i}\varUpsilon_{i,2}^2+\Psi_iv_{i,2}^2)+2\hat{\eta}v_{i,2}),\\
			&Q_i(t)=\hat{\lambda}\theta_i\mu_{i 1}(v_{i,2}^2\Psi_i+\delta_{i}\varUpsilon_{i,2}^2),\quad P_i(t)=2\hat{\eta}(\lambda_i+\hat{\lambda})\mu_{i 2}v_{i,2}.
	\end{aligned}\end{equation*}
	Then \begin{equation*}
		\hat{a}^\circ_i(t)R_i(t)=Q_i(t)\frac{1}{n}\sum_{k\neq i}\mu_{k 1}a^*_k(t)+P_i(t).
	\end{equation*}
Using the above notation, we can derive the following expression from \eqref{equ:pia2}:
 \begin{equation*}
		\left\{\begin{aligned}
			&{\pi}^\circ_i=\frac{\theta_i}{n-\theta_i}\sum_{k\neq i}\pi^*_k+ \frac{ n\sqrt{z}}{(n-\theta_i)\sigma}S_i(t),\\
			&a_i^\circ(t)=(\hat{a}^\circ_i(t)\vee0)\wedge1\\
			&\qquad=\left(\left(\frac{Q_i(t)}{R_i(t)}\frac{1}{n}\sum_{k\neq i}\mu_{k 1}a_k^*(t)+\frac{P_i(t)}{R_i(t)}\right)\vee0\right)\wedge1\\
			&\qquad=\left(\frac{Q_i(t)}{R_i(t)}\frac{1}{n}\sum_{k\neq i}\mu_{k 1}a_k^*(t)+\frac{P_i(t)}{R_i(t)}\right)\wedge1,
		\end{aligned}\right.
	\end{equation*}
	The associated worst-case scenario density generators can be obtained from \eqref{equ:px}, and are as follows:
	\begin{equation*}
		\left\{\begin{aligned}
			&{\varphi^\circ_{i}}=-(\nu \rho v_{i,3}+S_i(t) v_{i,2})\Psi_i\sqrt{z},\\
			&{\chi^\circ_{i}}=-\nu\sqrt{z}\sqrt{1-\rho^2} v_{i,3}\Psi_i,\\
			&\phi^\circ_{{i}}=-\sqrt{\hat{\lambda}}\left[(1-\frac{\theta_i}{n})\mu_{i1}a_{i}^\circ(t)-\frac{\theta_{i}}{n}\sum_{k\neq i}\mu_{k1}a_{k}^*(t)\right]v_{i,2}\Psi_i,\\
			&\vartheta^\circ_{i,i}=-(1-\frac{\theta_i}{n})a_{i}^\circ(t)\sqrt{(\hat{\lambda}+\lambda_{i})\mu_{i 2}-\hat{\lambda}\mu_{i 1}^2}v_{i,2}\Psi_i,\\
			&\vartheta^\circ_{{i,k}}=\frac{\theta_i}{n}a_{k}^*(t)\sqrt{(\hat{\lambda}+\lambda_{k})\mu_{k 2}-\hat{\lambda}\mu_{k 1}^2}v_{i,2}\Psi_i,~k\neq i.
		\end{aligned}\right.
	\end{equation*}
	By \eqref{g-MV}, we have
	\begin{equation*}
		\begin{aligned}
			0=&\mathcal{A}^{\left\{(\pi_{k}^*,a_k^*)_{k\neq i},(\pi_{i}^\circ,a_i^\circ)\right\},(\varphi_i^\circ,\chi_i^\circ,\phi_i^\circ,\vartheta_{i}^\circ)}\varUpsilon^{(i)}(t,y,z)\\
			=&\varUpsilon^{(i)}_t\!+\!\left[{\kappa}(\bar{Z}\!-\!z)+{\nu}\sqrt{z}({\rho}\varphi_i^\circ\!+\!\sqrt{1-{\rho}^2}\chi_i^\circ)\right]\varUpsilon^{(i)}_{z} +ry\varUpsilon^{(i)}_y\\
			&\!+\!(1\!-\!\frac{\theta_i}{n})\!\left[\eta_{i}\!\left(\!\lambda_{i}\!+\!\hat{\lambda}\right) \mu_{i 1}\!-\!\hat{\eta}(1\!-\!a_i^\circ(t))^2\!\left(\!\lambda_{i}\!+\!\hat{\lambda}\!\right)\! \mu_{i 2}\!+\!\pi_{i}(t)\!\left(m\!-\!(\nu \rho v_{i,3}+S_i(t) v_{i,2})\Psi_i\right)\!({a}{z}\!+\!{{b}})\right]\! \varUpsilon^{(i)}_y\\
			&\!-\!\frac{\theta_{i}}{n}\sum_{k\neq i}\left[\eta_{k}\left(\lambda_{k}\!+\!\hat{\lambda}\right) \mu_{k 1}\!-\!\hat{\eta}(1\!-\!a_k^*(t))^2\!\left(\!\lambda_{k}\!+\!\hat{\lambda}\right) \mu_{k 2}\!+\!\pi_{k}(t)\left(m\!-\!(\nu \rho v_{i,3}\!+\!S_i(t) v_{i,2})\Psi_i\right)\!({a}{z}\!+\!{{b}})\right]\!\varUpsilon^{(i)}_y\\
			&+(1-\frac{\theta_i}{n})a_{i}^\circ(t)\left(\sqrt{\hat{\lambda}}\mu_{i1}\phi_{i}^\circ(t)+\sqrt{(\hat{\lambda}+\lambda_{i})\mu_{i 2}-\hat{\lambda}\mu_{i 1}^2}\vartheta_{i,i}^\circ(t)\right)\varUpsilon^{(i)}_y\\
			&-\frac{\theta_{i}}{n}\sum_{k\neq i}a_{k}^*(t) \left(\sqrt{\hat{\lambda}}\mu_{k1}\phi_{i}^\circ(t)+\sqrt{(\hat{\lambda}+\lambda_{k})\mu_{k 2}-\hat{\lambda}\mu_{k 1}^2}\vartheta_{i,k}^\circ(t) \right)\varUpsilon^{(i)}_y,
		\end{aligned}
	\end{equation*} 
	that is 
	\begin{equation*}
		\begin{aligned}
			0=&\varUpsilon_{i,1}' \!+\! y\varUpsilon_{i,2}' \!+\!\varUpsilon_{i,3}'
			z\!+\!\left[{\kappa}\bar{Z}\!-\!{\kappa}z\!-\!{\nu}{z}({\rho}(\nu \rho v_{i,3}\!+\!S_i(t) v_{i,2})\Psi_i\!+ \!\nu ({1\!-\!\rho^2}) v_{i,3}\Psi_i)\right]\!\varUpsilon_{i,3} \!+\!ry\varUpsilon_{i,2}\\
			&+(1-\frac{\theta_i}{n})\left[\eta_{i}\left(\lambda_{i}+\hat{\lambda}\right) \mu_{i 1}-\hat{\eta}(1-a_i^\circ(t))^2\left(\lambda_{i}+\hat{\lambda}\right) \mu_{i 2} \right] \varUpsilon_{i,2}\\
			&-\frac{\theta_{i}}{n}\sum_{k\neq i}\left[\eta_{k}\left(\lambda_{k}+\hat{\lambda}\right) \mu_{k 1}-\hat{\eta}(1-a_k^*(t))^2\left(\lambda_{k}+\hat{\lambda}\right) \mu_{k 2} \right]\varUpsilon_{i,2}\\
			&+ S_i(t)z\left(m-(\nu \rho v_{i,3}+S_i(t) v_{i,2})\Psi_i\right) \varUpsilon_{i,2}\\
			&+(1-\frac{\theta_i}{n})a_{i}^\circ(t)\left(\sqrt{\hat{\lambda}}\mu_{i1}\phi_{i}^\circ(t)-(1-\frac{\theta_i}{n})a_{i}^\circ(t)((\hat{\lambda}+\lambda_{i})\mu_{i 2}-\hat{\lambda}\mu_{i 1}^2)v_{i,2}\Psi_i\right)\varUpsilon_{i,2}\\
			&-\frac{\theta_{i}}{n}\sum_{k\neq i}a_{k}^*(t) \left(\sqrt{\hat{\lambda}}\mu_{k1}\phi_{i}^\circ(t)+\frac{\theta_i}{n}a_{k}^*(t)((\hat{\lambda}+\lambda_{k})\mu_{k 2}-\hat{\lambda}\mu_{k 1}^2)v_{i,2}\Psi_i \right)\varUpsilon_{i,2}.
		\end{aligned}
	\end{equation*} 
	
Comparing the coefficients of $y$ and $z$ in the last equation, we can derive the ordinary differential equations (ODEs) for $\varUpsilon_{i,1}$, $\varUpsilon_{i,2}$, and $\varUpsilon_{i,3}$ as follows:
	\begin{equation}\label{varUpsilon}
		\left\{\begin{aligned}
			0=&\varUpsilon_{i,1}'  +{\kappa}\bar{Z}\varUpsilon_{i,3} 
			+(1-\frac{\theta_i}{n})\left[\eta_{i}\left(\lambda_{i}+\hat{\lambda}\right) \mu_{i 1}-\hat{\eta}(1-a_i^\circ(t))^2\left(\lambda_{i}+\hat{\lambda}\right) \mu_{i 2} \right] \varUpsilon_{i,2}\\
			&-\frac{\theta_{i}}{n}\sum_{k\neq i}\left[\eta_{k}\left(\lambda_{k}+\hat{\lambda}\right) \mu_{k 1}-\hat{\eta}(1-a_k^*(t))^2\left(\lambda_{k}+\hat{\lambda}\right) \mu_{k 2} \right]\varUpsilon_{i,2}\\
			&+(1-\frac{\theta_i}{n})a_{i}^\circ(t)\left(\sqrt{\hat{\lambda}}\mu_{i1}\phi_{i}^\circ(t)+\sqrt{(\hat{\lambda}+\lambda_{i})\mu_{i 2}-\hat{\lambda}\mu_{i 1}^2}\vartheta_{i,i}^\circ(t)\right)\varUpsilon_{i,2}\\
			&-\frac{\theta_{i}}{n}\sum_{k\neq i}a_{k}^*(t) \left(\sqrt{\hat{\lambda}}\mu_{k1}\phi_{i}^\circ(t)+\sqrt{(\hat{\lambda}+\lambda_{k})\mu_{k 2}-\hat{\lambda}\mu_{k 1}^2}\vartheta_{i,k}^\circ(t) \right)\varUpsilon_{i,2},\\
			0=	&\varUpsilon_{i,2}'+r\varUpsilon_{i,2},\\
			0=&\varUpsilon_{i,3}' -\left[{\kappa}+{\nu}{\rho}(\nu \rho v_{i,3}+S_i(t) v_{i,2})\Psi_i+ \nu^2 ({1-\rho^2}) v_{i,3}\Psi_i\right]\varUpsilon_{i,3}\\
			&+ S_i(t)\left(m-(\nu \rho v_{i,3}+S_i(t) v_{i,2})\Psi_i\right) \varUpsilon_{i,2}.
		\end{aligned}\right.	
	\end{equation}
On the other hand, by \eqref{HJBI-MV}, we have
\begin{equation*}
		\begin{aligned}
			0=	&	\mathcal{L}_{i}\left({\left\{(\pi_{k}^*,a_k^*)_{k\neq i},(\pi_{i}^\circ,a_i^\circ)\right\},(\varphi_i^\circ,\chi_i^\circ,\phi_i^\circ,\vartheta_{i}^\circ)},v^{(i)},\varUpsilon^{(i)}, \Psi_i,(t,y,z)\right)\\=&\mathcal{A}^{\left\{(\pi_{k}^*,a_k^*)_{k\neq i},(\pi_{i}^\circ,a_i^\circ)\right\},(\varphi_i^\circ,\chi_i^\circ,\phi_i^\circ,\vartheta_{i}^\circ)}v^{(i)}(t,y,z)+  \frac{\varphi^{2}(t)}{2 \Psi_i} +\frac{\chi^2(t)}{2 \Psi_i} +  \frac{\phi^{2}(t)}{2 \Psi_i} +\frac{\vartheta(t)^T\vartheta(t)}{2 \Psi_i} \\
			&+\delta_{i} \varUpsilon^{(i)}(t,y,z) \mathcal{A}^{\left\{(\pi_{k}^*,a_k^*)_{k\neq i},(\pi_{i}^\circ,a_i^\circ)\right\},(\varphi_i^\circ,\chi_i^\circ,\phi_i^\circ,\vartheta_{i}^\circ)}\varUpsilon^{(i)}(t,y,z)\\&-\frac{\delta_{i}}{2}\mathcal{A}^{\left\{(\pi_{k}^*,a_k^*)_{k\neq i},(\pi_{i}^\circ,a_i^\circ)\right\},(\varphi_i^\circ,\chi_i^\circ,\phi_i^\circ,\vartheta_{i}^\circ)}(\varUpsilon^{(i)})^2(t,y,z)\\
			=&v_{i,1}' + yv_{i,2}' +v_{i,3}'
			z\!+\!\left[{\kappa}\bar{Z}\!-{\kappa}z-{\nu}{z}({\rho}(\nu \rho v_{i,3}+S_i(t) v_{i,2})\Psi_i+ \nu ({1-\rho^2}) v_{i,3}\Psi_i)\right]v_{i,3} +ryv_{i,2}\\
			&+(1-\frac{\theta_i}{n})\left[\eta_{i}\left(\lambda_{i}+\hat{\lambda}\right) \mu_{i 1}-\hat{\eta}(1-a_i^\circ(t))^2\left(\lambda_{i}+\hat{\lambda}\right) \mu_{i 2} \right] v_{i,2}\\
			&-\frac{\theta_{i}}{n}\sum_{k\neq i}\left[\eta_{k}\left(\lambda_{k}+\hat{\lambda}\right) \mu_{k 1}-\hat{\eta}(1-a_k^*(t))^2\left(\lambda_{k}+\hat{\lambda}\right) \mu_{k 2} \right]v_{i,2}\\
			&+ S_i(t)z\left(m-(\nu \rho v_{i,3}+S_i(t) v_{i,2})\Psi_i\right) v_{i,2}\\
			&+(1-\frac{\theta_i}{n})a_{i}^\circ(t)\left(\sqrt{\hat{\lambda}}\mu_{i1}\phi_{i}^\circ(t)-(1-\frac{\theta_i}{n})a_{i}^\circ(t)((\hat{\lambda}+\lambda_{i})\mu_{i 2}-\hat{\lambda}\mu_{i 1}^2)v_{i,2}\Psi_i\right)v_{i,2}\\
			&-\frac{\theta_{i}}{n}\sum_{k\neq i}a_{k}^*(t) \left(\sqrt{\hat{\lambda}}\mu_{k1}\phi_{i}^\circ(t)+\frac{\theta_i}{n}a_{k}^*(t)((\hat{\lambda}+\lambda_{k})\mu_{k 2}-\hat{\lambda}\mu_{k 1}^2)v_{i,2}\Psi_i \right)v_{i,2}\\
			&-\nu^2\frac{\delta_{i}}{2}z\varUpsilon_{i,3}^2\!-\delta_{i}\nu\sqrt{z}\rho\sigma((1\!-\!\frac{\theta_{i}}{n})\pi_i^\circ\!-\!\frac{\theta_{i}}{n}\sum_{k\neq i}\pi_{k}^*(t))\varUpsilon_{i,2}\varUpsilon_{i,3}\\
			&-\frac{\delta_{i}}{2}(1-\frac{\theta_i}{n})^2(a_{i}^\circ(t))^2 \left(\hat{\lambda}+\lambda_{i}\right) \mu_{i 2}\varUpsilon_{i,2}^2-\frac{\delta_{i}\theta_{i}^2}{2n^2}\sum_{k\neq i}(a_{k}^*(t))^2 \left(\left(\hat{\lambda}\!+\!\lambda_{k}\right) \mu_{k 2}-\hat{\lambda}\mu_{k1}^2\right)\varUpsilon_{i,2}^2\\
			&\!\!-\!\frac{\delta_{i}\hat{\lambda}\theta_i^2}{2n^2}\!\!\left[\!\sum_{k\neq i}a_k^*(t)\mu_{k 1}\!\right]^2\!\!\!\!\varUpsilon_{i,2}^2\!+\!\delta_{i}\hat{\lambda}\frac{\theta_{i}}{n}(1\!-\!\frac{\theta_i}{n})a_i^\circ(t)\mu_{i 1}\!\sum_{k\neq i}a_{k}^*(t)\mu_{k 1}\varUpsilon_{i,2}^2\!-\!\frac{\delta_{i}}{2}\sigma^2((1\!-\!\frac{\theta_{i}}{n})\pi_i^\circ\!-\!\frac{\theta_{i}}{n}\sum_{k\neq i}\pi_{k}^*(t))^2\varUpsilon_{i,2}^2\\
			&\quad+\!\frac{1}{2}(\nu \rho v_{i,3}\!+\!S_i(t) v_{i,2})^2\Psi_iz
			\!+\!\frac{1}{2}\nu^2z{(1-\rho^2)} v_{i,3}^2\Psi_i
			+\frac{1}{2}{\hat{\lambda}}\!\left[(1-\frac{\theta_i}{n})\mu_{i1}a_{i}^\circ(t)\!-\!\frac{\theta_{i}}{n}\sum_{k\neq i}\mu_{k1}a_{k}^*(t)\right]^2\!\!\!\!v_{i,2}^2\Psi_i\\
			&\quad+\frac{1}{2}(1-\frac{\theta_i}{n})^2(a_{i}^\circ(t))^2{((\hat{\lambda}+\lambda_{i})\mu_{i 2}-\hat{\lambda}\mu_{i 1}^2)}v_{i,2}^2\Psi_i
			+\frac{1}{2}\sum_{k\neq i}\frac{\theta_i^2}{n^2}(a_{k}^*(t))^2{((\hat{\lambda}+\lambda_{k})\mu_{k 2}-\hat{\lambda}\mu_{k 1}^2)}v_{i,2}^2\Psi_i,
		\end{aligned}
	\end{equation*}
	that is \begin{equation*}
		\begin{aligned}
			0=	&v_{i,1}' + yv_{i,2}' +v_{i,3}'
			z\!+\!\left[{\kappa}\bar{Z}\!-{\kappa}z \right]v_{i,3} +ryv_{i,2}\\
			&+(1-\frac{\theta_i}{n})\left[\eta_{i}\left(\lambda_{i}+\hat{\lambda}\right) \mu_{i 1}-\hat{\eta}(1-a_i^\circ(t))^2\left(\lambda_{i}+\hat{\lambda}\right) \mu_{i 2} \right] v_{i,2}\\
			&-\frac{\theta_{i}}{n}\sum_{k\neq i}\left[\eta_{k}\left(\lambda_{k}+\hat{\lambda}\right) \mu_{k 1}-\hat{\eta}(1-a_k^*(t))^2\left(\lambda_{k}+\hat{\lambda}\right) \mu_{k 2} \right]v_{i,2}\\
			&+ S_i(t)v_{i,2}mz
			-\nu^2\frac{\delta_{i}}{2}z\varUpsilon_{i,3}^2\!-\delta_{i}\nu z\rho S_i(t)\varUpsilon_{i,2}\varUpsilon_{i,3}\\
			&-\frac{\delta_{i}}{2}(1-\frac{\theta_i}{n})^2a_{i}^2(t) \left(\hat{\lambda}+\lambda_{i}\right) \mu_{i 2}\varUpsilon_{i,2}^2-\frac{\delta_{i}\theta_{i}^2}{2n^2}\sum_{k\neq i}(a_{k}(t))^2 \left(\left(\hat{\lambda}\!+\!\lambda_{k}\right) \mu_{k 2}-\hat{\lambda}\mu_{k1}^2\right)\varUpsilon_{i,2}^2\\&-\frac{\delta_{i}\hat{\lambda}\theta_i^2}{2n^2}\left[\sum_{k\neq i}a_k(t)\mu_{k 1}\right]^2\varUpsilon_{i,2}^2+\delta_{i}\hat{\lambda}\frac{\theta_{i}}{n}(1-\frac{\theta_i}{n})a_i(t)\mu_{i 1}\sum_{k\neq i}a_{k}(t)\mu_{k 1}\varUpsilon_{i,2}^2-\frac{\delta_{i}}{2}S_i(t)^2z\varUpsilon_{i,2}^2\\
			&-\!\frac{1}{2}(\nu \rho v_{i,3}\!+\!S_i(t) v_{i,2})^2\Psi_iz\!-\!\frac{1}{2}\nu^2z{(1-\rho^2)} v_{i,3}^2\Psi_i\!-\!\frac{1}{2}{\hat{\lambda}}\left[(1-\frac{\theta_i}{n})\mu_{i1}a_{i}^\circ(t)\!-\!\frac{\theta_{i}}{n}\sum_{k\neq i}\mu_{k1}a_{k}^*(t)\right]^2\!\!\!v_{i,2}^2\Psi_i
			\\
			\end{aligned}
		\end{equation*}
	\begin{equation*}
\begin{aligned}
			&-\frac{1}{2}(1-\frac{\theta_i}{n})^2(a_{i}^\circ(t))^2{((\hat{\lambda}+\lambda_{i})\mu_{i 2}-\hat{\lambda}\mu_{i 1}^2)}v_{i,2}^2\Psi_i-\frac{1}{2}\sum_{k\neq i}\frac{\theta_i^2}{n^2}(a_{k}^*(t))^2{((\hat{\lambda}+\lambda_{k})\mu_{k 2}-\hat{\lambda}\mu_{k 1}^2)}v_{i,2}^2\Psi_i.
		\end{aligned}
	\end{equation*}
By comparing the coefficients of $y$ and $z$ in the last equation, we obtain the ODEs for $v_{i,1}$, $v_{i,2}$ and $v_{i,3}$ as follows:
	\begin{equation}\label{v1}
\begin{aligned}
			0=	&v_{i,1}'+{\kappa}\bar{Z}v_{i,3}+(1-\frac{\theta_i}{n})\left[\eta_{i}\left(\lambda_{i}+\hat{\lambda}\right) \mu_{i 1}-\hat{\eta}(1-a_i^\circ(t))^2\left(\lambda_{i}+\hat{\lambda}\right) \mu_{i 2} \right] v_{i,2}\\
			&-\frac{\theta_{i}}{n}\sum_{k\neq i}\left[\eta_{k}\left(\lambda_{k}+\hat{\lambda}\right) \mu_{k 1}-\hat{\eta}(1-a_k^*(t))^2\left(\lambda_{k}+\hat{\lambda}\right) \mu_{k 2} \right]v_{i,2}\\
			&-\frac{\delta_{i}}{2}(1-\frac{\theta_i}{n})^2a_{i}^2(t) \left(\hat{\lambda}+\lambda_{i}\right) \mu_{i 2}\varUpsilon_{i,2}^2-\frac{\delta_{i}\theta_{i}^2}{2n^2}\sum_{k\neq i}(a_{k}(t))^2 \left(\left(\hat{\lambda}\!+\!\lambda_{k}\right) \mu_{k 2}-\hat{\lambda}\mu_{k1}^2\right)\varUpsilon_{i,2}^2\\
			&-\frac{\delta_{i}\hat{\lambda}\theta_i^2}{2n^2}\left[\sum_{k\neq i}a_k(t)\mu_{k 1}\right]^2\varUpsilon_{i,2}^2+\delta_{i}\hat{\lambda}\frac{\theta_{i}}{n}(1-\frac{\theta_i}{n})a_i(t)\mu_{i 1}\sum_{k\neq i}a_{k}(t)\mu_{k 1}\varUpsilon_{i,2}^2\\
			&-\frac{1}{2}{\hat{\lambda}}\left[(1-\frac{\theta_i}{n})\mu_{i1}a_{i}^\circ(t)-\frac{\theta_{i}}{n}\sum_{k\neq i}\mu_{k1}a_{k}^*(t)\right]^2v_{i,2}^2\Psi_i
		   -\frac{1}{2}(1-\frac{\theta_i}{n})^2(a_{i}^\circ(t))^2{((\hat{\lambda}+\lambda_{i})\mu_{i 2}-\hat{\lambda}\mu_{i 1}^2)}v_{i,2}^2\Psi_i\\
			&-\frac{1}{2}\sum_{k\neq i}\frac{\theta_i^2}{n^2}(a_{k}^*(t))^2{((\hat{\lambda}+\lambda_{k})\mu_{k 2}-\hat{\lambda}\mu_{k 1}^2)}v_{i,2}^2\Psi_i.
								\end{aligned}	
		\end{equation}
	\begin{equation}\label{v2}
\begin{aligned}
			0=	& v_{i,2}'  +rv_{i,2},\\
					\end{aligned}	
		\end{equation}
				\begin{equation}\label{v3}
				\begin{aligned}
			0=	&v_{i,3}'
			-{\kappa} v_{i,3} + S_i(t)v_{i,2}m-\nu^2\frac{\delta_{i}}{2}\varUpsilon_{i,3}^2\!-\delta_{i}\nu \rho S_i(t)\varUpsilon_{i,2}\varUpsilon_{i,3}\\
			& -\frac{\delta_{i}}{2}S_i(t)^2\varUpsilon_{i,2}^2-\frac{1}{2}(\nu \rho v_{i,3}+S_i(t) v_{i,2})^2\Psi_i-\frac{1}{2}\nu^2{(1-\rho^2)} v_{i,3}^2\Psi_i.
		\end{aligned}	
	\end{equation}
	Based on \eqref{varUpsilon} and \eqref{v2}, we directly have
	\begin{equation*}
		v_{i,2}=\varUpsilon_{i,2}=e^{r(T-t)}.
	\end{equation*}
Then \begin{equation*}
		S_i(t)=\frac{m e^{r(T-t)}-\nu\rho( e^{r(T-t)}v_{i,3}\Psi_i+\delta_{i} e^{r(T-t)}\varUpsilon_{i,3})}{(\Psi_i+{\delta_{i}})e^{2r(T-t)}}=\frac{m -\nu\rho( v_{i,3}\Psi_i+\delta_{i} \varUpsilon_{i,3})}{(\Psi_i+{\delta_{i}})e^{r(T-t)}}.
	\end{equation*}
	To solve $ (\pi_{i}^\circ,a_i^\circ)$ and $(\varphi_i^\circ,\chi_i^\circ,\phi_i^\circ,\vartheta_{i}^\circ) $, we need to solve $ v_{i,3}$ and $\varUpsilon_{i,3} $. Combining \eqref{varUpsilon} and \eqref{v3}, we have
	\begin{equation*}
		\left\{	\begin{aligned}
			0=	&v_{i,3}'
			-{\kappa} v_{i,3} + S_i(t)e^{r(T-t)}m-\nu^2\frac{\delta_{i}}{2}\varUpsilon_{i,3}^2\!-\delta_{i}\nu \rho S_i(t)e^{r(T-t)}\varUpsilon_{i,3}\\
			& -\frac{\delta_{i}}{2}S_i(t)^2e^{2r(T-t)}-\frac{1}{2}(\nu \rho v_{i,3}+S_i(t) e^{r(T-t)})^2\Psi_i-\frac{1}{2}\nu^2{(1-\rho^2)} v_{i,3}^2\Psi_i\\
			0=&\varUpsilon_{i,3}' -\left[{\kappa}+{\nu}{\rho}(\nu \rho v_{i,3}+S_i(t) e^{r(T-t)})\Psi_i+ \nu^2 ({1-\rho^2}) v_{i,3}\Psi_i\right]\varUpsilon_{i,3}\\
			&+ S_i(t)e^{r(T-t)}\left(m-(\nu \rho v_{i,3}+S_i(t) e^{r(T-t)})\Psi_i\right),
		\end{aligned}\right.
	\end{equation*}
	that is,\begin{equation}\label{v-Upsilon}
		\left\{\begin{aligned}
			v_{i,3}'=&\frac{1}{2}\nu^2(\rho^2\frac{\Psi_i\delta_{i}}{\Psi_i+\delta_{i}}+(1-\rho^2)\Psi_i)v_{i,3}^2+(\kappa+m\nu\rho\frac{\Psi_i}{\Psi_i+\delta_{i}})v_{i,3}-\nu^2\rho^2\frac{\Psi_i\delta_{i}}{\Psi_i+\delta_{i}}v_{i,3}\varUpsilon_{i,3}\\
			&+m\nu\rho\frac{\delta_{i}}{\Psi_i+\delta_{i}}\varUpsilon_{i,3}+\frac{1}{2}\nu^2\delta_{i}(1-\rho^2\frac{\delta_{i}}{\Psi_i+\delta_{i}} )\varUpsilon_{i,3}^2-\frac{m^2}{2(\Psi_i+\delta_{i})},\\
			\varUpsilon_{i,3}'=&-\nu^2\rho^2\frac{\delta_{i} \Psi_i}{(\Psi_i+{\delta_{i}})^2}\varUpsilon_{i,3}^2+\left[\kappa +m\nu\rho\frac{{\delta_{i}}^2+(\Psi_i)^2}{(\Psi_i+{\delta_{i}})^2}\right]\varUpsilon_{i,3}\\
			&+\nu^2\left[ \rho^2\Psi_i\frac{2\Psi_i{\delta_{i}}}{(\Psi_i+{\delta_{i}})^2}+(1-\rho^2)\Psi_i\right]v_{i,3}\varUpsilon_{i,3}\\
			&+m\nu\rho\frac{2\Psi_i{\delta_{i}}}{(\Psi_i+{\delta_{i}})^2}v_{i,3}-\nu^2\rho^2\Psi_i\frac{\delta_{i} \Psi_i}{(\Psi_i+{\delta_{i}})^2}v_{i,3}^2-\frac{\delta_{i} m^2}{(\Psi_i+\delta_{i})^2}.
		\end{aligned}\right.
	\end{equation}
We conclude that the solution to the coupled Riccati differential equation (\ref{v-Upsilon}) can be expressed in terms of the solution to (\ref{ricc}), which is given by
 \begin{equation*}
\diag(v_{i,3}(t),\varUpsilon_{i,3}(t) )=F_i(T-t).
	\end{equation*}
Based on  Proposition \ref{ricc}, it is easy to see that $ (\varphi_i^\circ,\chi_i^\circ,\phi_i^\circ,\vartheta_{i}^\circ)\in\mathscr{A} $ and $ (\pi_{i}^\circ,a_i^\circ)\in\mathscr{U}_i. $
\section{\bf Proof of Theorem \ref{Verification}}\label{proof-Verification}
For $ \left(\varphi_i(t), \chi_i(t),\phi_i(t),\vartheta_{i}(t)\right)\!\in\!\mathscr{A} $ and  $ (\pi_i,a_i)\!\in\!\mathscr{U}_i$, denote $ \left(\varphi_i(t), \chi_i(t)\right)\!=\!\left(\!\hslash_i(t)\!\sqrt{Z(t)},\hbar_i(t)\!\sqrt{Z(t)}\!\right) $, 
based on It\^{o}'s  formula,
 we only need to verify the following conditions for $ f\in\{v^{(i)},\varUpsilon^{(i)}, (\varUpsilon^{(i)})^2\}$
\begin{equation*}
	\left\{\begin{aligned}
		&\mathbb{E}^{ \mathbb{Q}^{\varphi_i,\chi_i,\phi_i,\vartheta_{i}}}_{s,y,z}\left[\int_{s}^{T}\left(f_z(t, Y(t),Z(t))\right)^2{Z(t)}\rd t\right]<\infty,\\
		&\mathbb{E}^{ \mathbb{Q}^{\varphi_i,\chi_i,\phi_i,\vartheta_{i}}}_{s,y,z}\left[\int_{s}^{T}\left(f_y(t, Y(t),Z(t))\right)^2\left((1-\frac{\theta_{i}}{n})\pi_{i}(t)-\frac{\theta_{i}}{n}\sum_{k\neq i}\pi^*_{k}(t)\right)^2({\Sigma}(t))^2\rd t\right]<\infty,\\
		&\mathbb{E}^{ \mathbb{Q}^{\varphi_i,\chi_i,\phi_i,\vartheta_{i}}}_{s,y,z}\left[\int_{s}^{T}\left(f_y(t, Y(t),Z(t))\right)^2\left((1-\frac{\theta_{i}}{n})a_{i}(t)\mu_{i 1}-\frac{\theta_{i}}{n}\sum_{k\neq i}a^*_{k}(t)\mu_{k1}\right)^2\rd t\right]<\infty,\\
		&\mathbb{E}^{ \mathbb{Q}^{\varphi_i,\chi_i,\phi_i,\vartheta_{i}}}_{s,y,z}\left[\int_{s}^{T}\left(f_y(t, Y(t),Z(t))\right)^2\frac{1}{n}\sum_{k\neq i}(a_{k}^*(t))^2 {((\hat{\lambda}+\lambda_{k})\mu_{k 2}-\hat{\lambda}\mu_{k 1}^2)}\rd t\right]<\infty,\\
		&\mathbb{E}^{ \mathbb{Q}^{\varphi_i,\chi_i,\phi_i,\vartheta_{i}}}_{s,y,z}\left[\int_{s}^{T}\left(f_y(t, Y(t),Z(t))\right)^2a_{i}^2(t) {((\hat{\lambda}+\lambda_{i})\mu_{i 2}-\hat{\lambda}\mu_{i1}^2)}\rd t\right]<\infty.
	\end{aligned}\right.
\end{equation*}In fact, 
\begin{equation*}
	\begin{aligned}
		&\rd Z(t)\\=&{\kappa}(\bar{Z}\!-\!Z(t))\rd t\!+\!{\nu}\sqrt{Z(t)}\!\left[{\rho}(\rd W^{ \mathbb{Q}^{\varphi_i,\chi_i,\phi_i,\vartheta_{i}}}(t)\!+\!\varphi_i(t)\rd t)\!+\!\sqrt{1\!-\!{\rho}^2}(\rd B^{ \mathbb{Q}^{\varphi_i,\chi_i,\phi_i,\vartheta_{i}}}(t)\!+\!\chi_i(t)\rd t)\right]\\
		=&(\kappa\bar{Z}\!-\!\left(\!{\kappa}\!-\!\nu \rho\hslash_i(t)\!-\!\nu \sqrt{1-{\rho}^2}\hbar_i(t)\!\right)\!Z(t))\rd t\!+\!{\nu}\sqrt{Z(t)}\!\left[{\rho}\rd W^{ \mathbb{Q}^{\varphi_i,\chi_i,\phi_i,\vartheta_{i}}}\!(t)\!+\!\sqrt{1-{\rho}^2}\rd B^{ \mathbb{Q}^{\varphi_i,\chi_i,\phi_i,\vartheta_{i}}}\!(t)\right]\!.
	\end{aligned}
\end{equation*}
Let $ \bar{\kappa}= {\kappa}-\nu (\rho+\sqrt{1-{\rho}^2})\mathcal{C} $, as $ \mathcal{C}^2<\frac{\kappa^2}{2\nu^2} $, then $\mathcal{C}< \frac{{\kappa}}{\nu (\rho+\sqrt{1-{\rho}^2})} $, then $ 0<\bar{\kappa}\leqslant {\kappa}-\nu \rho\hslash_i(t)-\nu \sqrt{1-{\rho}^2}\hbar_i(t)$. $\forall s\in[0,T]$, let $ \hat{Z}^s(t) $ be determined by the CIR model as follows \begin{equation*}
	\rd \hat{Z}^s(t)\!=\!(\kappa\bar{Z}-\bar{\kappa}\hat{Z}^s(t))\rd t+{\nu}\sqrt{\hat{Z}^s(t)}\!\left[{\rho}\rd W^{ \mathbb{Q}^{\varphi_i,\chi_i,\phi_i,\vartheta_{i}}}(t)\!+\!\sqrt{1-{\rho}^2}\rd B^{ \mathbb{Q}^{\varphi_i,\chi_i,\phi_i,\vartheta_{i}}}(t)\right]\!\!, t\in[s,T],\hat{Z}^s(s)\!=\!{Z}(s).
\end{equation*}
Then by the comparison theorem, we see \begin{equation*}
	Z(t)\leqslant\hat{Z}^s(t), \quad\forall t\in[s,T],\quad a.s.
\end{equation*}
Thus,  
\begin{equation*}
	\mathbb{E}^{ \mathbb{Q}^{\varphi_i,\chi_i,\phi_i,\vartheta_{i}}}_{s,y,z}\left[\int_{s}^{T}{Z(t)}\rd t\right]\leqslant\mathbb{E}^{ \mathbb{Q}^{\varphi_i,\chi_i,\phi_i,\vartheta_{i}}}_{s,y,z}\left[\int_{s}^{T}{\hat{Z}^s(t)}\rd t\right]=\int_{s}^T\left[ze^{-\bar{\kappa}(t-s)}+\frac{\kappa}{\bar{\kappa}}\bar{Z}(1-e^{-\bar{\kappa}(t-s)})\right]\rd t<\infty.
\end{equation*}
	While for $ (\pi_i,a_i)\in\mathscr{U}_i$, $ \left\{(\pi_{k}^*,a_k^*)_{k\neq i}\right\}\in\mathscr{U}_{-i} $,  $  \exists \{\mathcal{C}_j\}_{j=1}^n\subseteq\mathbb{R}_+$, $ \mathcal{M}:=\sup_{1\leqslant j\leqslant n}\mathcal{C}_j<\infty, $ such that  $ \pi_k^*(t)=\ell_k(t)\frac{Z(t)}{aZ(t)+b}$, $|\ell_k(t)|\leqslant \mathcal{C}_k$, $\forall k\neq i$, $ \pi_i(t)=\ell_i(t)\frac{Z(t)}{aZ(t)+b}$, $|\ell_i(t)|\leqslant \mathcal{C}_i$, $ \forall t\in[0,T]$, $0\leqslant a_i\leqslant1$, $0\leqslant a_k^*\leqslant1 $. Based on Theorem {\ref{solution-HJBI}} and Proposition \ref{ricc}, it is easy to see that $$ \sup_{t\in[s,T]} |f_z(t, Y(t),Z(t))|<\infty,\quad\sup_{t\in[s,T]} |f_y(t, Y(t),Z(t))|<\infty,\quad\forall f\in\{v^{(i)},\varUpsilon^{(i)}, (\varUpsilon^{(i)})^2\},$$
	then it is obvious that
	\begin{equation*}
		\left\{\begin{aligned}
			&\mathbb{E}^{ \mathbb{Q}^{\varphi_i,\chi_i,\phi_i,\vartheta_{i}}}_{s,y,z}\left[\int_{s}^{T}\left(f_y(t, Y(t),Z(t))\right)^2\left((1-\frac{\theta_{i}}{n})a_{i}(t)\mu_{i 1}-\frac{\theta_{i}}{n}\sum_{k\neq i}a^*_{k}(t)\mu_{k1}\right)^2\rd t\right]<\infty,\\
			&\mathbb{E}^{ \mathbb{Q}^{\varphi_i,\chi_i,\phi_i,\vartheta_{i}}}_{s,y,z}\left[\int_{s}^{T}\left(f_y(t, Y(t),Z(t))\right)^2\frac{1}{n}\sum_{k\neq i}(a_{k}^*(t))^2 {((\hat{\lambda}+\lambda_{k})\mu_{k 2}-\hat{\lambda}\mu_{k 1}^2)}\rd t\right]<\infty,\\
			&\mathbb{E}^{ \mathbb{Q}^{\varphi_i,\chi_i,\phi_i,\vartheta_{i}}}_{s,y,z}\left[\int_{s}^{T}\left(f_y(t, Y(t),Z(t))\right)^2a_{i}^2(t) {((\hat{\lambda}+\lambda_{i})\mu_{i 2}-\hat{\lambda}\mu_{i1}^2)}\rd t\right]<\infty.
		\end{aligned}\right.
	\end{equation*}
Moreover,
\begin{equation*}
	\begin{aligned}
		&\mathbb{E}^{ \mathbb{Q}^{\varphi_i,\chi_i,\phi_i,\vartheta_{i}}}_{s,y,z}\left[\int_{s}^{T}\left(f_y(t, Y(t),Z(t))\right)^2\left((1-\frac{\theta_{i}}{n})\pi_{i}(t)-\frac{\theta_{i}}{n}\sum_{k\neq i}\pi^*_{k}(t)\right)^2({\Sigma}(t))^2\rd t\right]\\
		\leqslant&\mathcal{M}^2\sup_{t\in[s,T]} |f_y(t, Y(t),Z(t))|^2	\mathbb{E}^{ \mathbb{Q}^{\varphi_i,\chi_i,\phi_i,\vartheta_{i}}}_{s,y,z}\left[\int_{s}^{T}{Z(t)}\rd t\right]<\infty
	\end{aligned}
\end{equation*}
and
\begin{equation*}
	\mathbb{E}^{ \mathbb{Q}^{\varphi_i,\chi_i,\phi_i,\vartheta_{i}}}_{s,y,z}\!\left[\int_{s}^{T}\left(f_z(t, Y(t),Z(t))\right)^2{Z(t)}\rd t\right]\!\leqslant\! \sup_{t\in[s,T]} |f_z(t, Y(t),Z(t))|^2\mathbb{E}^{ \mathbb{Q}^{\varphi_i,\chi_i,\phi_i,\vartheta_{i}}}_{s,y,z}\!\left[\int_{s}^{T}{Z(t)}\rd t\right]\!<\!\infty.
\end{equation*}The others are similar to that in \cite{Chi2018} and we omit them here.

	\bibliographystyle{apalike}
	\bibliography{RNG-MV}
	
\end{document}